\documentclass[11pt]{article}

\usepackage[T1]{fontenc}
\usepackage[utf8]{inputenc}
\usepackage{mathpazo}
\usepackage{courier}
\usepackage{microtype}
\usepackage{geometry}
\usepackage{amsmath,amssymb,amsthm,mathtools}
\usepackage{mathrsfs}
\usepackage{bm}
\usepackage{booktabs}
\usepackage{array}
\usepackage{enumitem}
\usepackage{tikz-cd}
\usepackage[numbers,sort&compress]{natbib}
\usepackage{xcolor}
\usepackage{hyperref}
\usepackage[nameinlink,noabbrev,capitalise]{cleveref}

\hypersetup{
  colorlinks=true,
  linkcolor=blue!45!black,
  citecolor=green!35!black,
  urlcolor=blue!50!black,
  pdfauthor={Steven Rayan},
  pdftitle={Beyond K-Theory: Geometry and Holomorphy in Hyperbolic Band Theory},
  pdfsubject={Mathematical physics of hyperbolic band theory},
  pdfkeywords={K-theory, hyperbolic band theory, Jacobian, Higgs bundle, quantum geometry, thermodynamic Bloch limit},
  bookmarksnumbered=true,
  pdfdisplaydoctitle=true
}

\setlist{itemsep=0.25em,topsep=0.45em}
\allowdisplaybreaks

\newtheorem{theorem}{Theorem}[section]
\newtheorem{proposition}[theorem]{Proposition}

\newtheorem{corollary}[theorem]{Corollary}
\newtheorem{criterion}[theorem]{Criterion}

\theoremstyle{definition}
\newtheorem{definition}[theorem]{Definition}
\newtheorem{example}[theorem]{Example}

\theoremstyle{remark}
\newtheorem{remark}[theorem]{Remark}

\newcommand{\HH}{\mathbb H}
\newcommand{\CC}{\mathbb C}
\newcommand{\RR}{\mathbb R}
\newcommand{\ZZ}{\mathbb Z}

\newcommand{\UU}{\mathrm U}
\newcommand{\SU}{\mathrm{SU}}
\newcommand{\Jac}{\operatorname{Jac}}
\newcommand{\Pic}{\operatorname{Pic}}
\newcommand{\Hom}{\operatorname{Hom}}
\newcommand{\End}{\operatorname{End}}
\newcommand{\Tr}{\operatorname{Tr}}
\newcommand{\tr}{\operatorname{tr}}
\newcommand{\spec}{\operatorname{Spec}}

\newcommand{\ch}{\operatorname{ch}}
\newcommand{\Ind}{\operatorname{Ind}}
\newcommand{\cR}{\mathcal R}
\newcommand{\cM}{\mathcal M}
\newcommand{\cH}{\mathcal H}
\newcommand{\cP}{\mathcal P}
\newcommand{\cT}{\mathcal T}
\newcommand{\cE}{\mathcal E}
\newcommand{\cO}{\mathcal O}
\newcommand{\cG}{\mathcal G}

\newcommand{\dd}{\mathrm d}
\newcommand{\ii}{\mathrm i}
\newcommand{\e}{\mathrm e}
\newcommand{\eps}{\varepsilon}
\newcommand{\ket}[1]{\lvert #1\rangle}
\newcommand{\bra}[1]{\langle #1\rvert}
\newcommand{\braket}[2]{\langle #1\,|\,#2\rangle}

\title{\bfseries Beyond K-Theory:\\
Geometry and Holomorphy in Hyperbolic Band Theory}

\author{
Steven Rayan\\[0.35em]
\small Centre for Quantum Topology and Its Applications (quanTA)\\
\small Department of Mathematics and Statistics\\
\small University of Saskatchewan, Saskatoon, Saskatchewan S7N 5E6, Canada\\
\small \texttt{rayan@math.usask.ca}
}

\date{\today}

\begin{document}
\maketitle

\begin{abstract}
Topological \(K\)-theory is indispensable in band theory: it keeps track of
stable classes of occupied states and organizes robust phases.  For gapped
free-fermion phases in Euclidean crystals, it supplies the decisive stable
classification principle once symmetry and stabilization are fixed, although
it does not determine nonquantized physics within a phase.  In hyperbolic band
theory, by contrast, \(K\)-theory is less decisive for the full band problem: ordinary
\(K\)-classes collapse geometric and holomorphic variations already present in
generalized momentum sectors. The article has two intertwined purposes: to
explain why these \(K\)-invisible variations are physically important and to
formulate the distinction as an observable-factorization problem that yields new
mathematical results.

For a compact hyperbolic surface \(X\), we distinguish \emph{kinematical
holomorphy}, in which complex geometry organizes the sector spaces, from
\emph{dynamical holomorphy}, in which that geometry informs a Hamiltonian or
projector. We prove nonfactorization in three settings:
fibrewise \(K^0(X)\), occupied-state \(K^0(B)\), and operator-algebraic \(K\)-theory.
The affected quantities are spectra and Higgs spectral curves, the Berry holonomy and the quantum
metric, the partially filled Hall response, and the Fermi surface and nodal geometries.
We also isolate the quantized pairings and local charges retained by topology.
Up to an area factor, the Kotani--Sunada bottom-band Hessian is the Hodge
inner product on \(H^1(X;\RR)\).  Together with the integral intersection
form, this recovers the homology-marked principally polarized Jacobian and hence, by Torelli and
uniformization, the underlying complex and hyperbolic surface, albeit without a
full Teichm\"uller marking.  We derive the relevant hyperbolic-band
normalization and place this reconstruction in the information hierarchy.

Our framework rigorously separates finite-rank sectors from the
thermodynamic bulk.  Locally faithful covers reproduce polynomial traces
exactly and control continuous spectral observables.  For arithmetic
congruence towers, logarithmic systole growth yields power-law convergence in
quotient size for analytic observables and an \(O((\log |G_n|)^{-s})\) rate
for \(C^s\) observables, while established coherent large-rank limits recover
bulk density-of-states moments.  Together these results separate the stable
information retained by topology from the geometric, holomorphic, and
Hamiltonian data required by the remaining physics.
\end{abstract}



\tableofcontents

\section{Introduction}

\subsection{The point of this article}

Stable topological classification is one of the decisive
organizing ideas of modern condensed-matter theory
\cite{KitaevPeriodicTable,ProdanSchulzBaldes}.  It identifies information that
survives gap-preserving deformation, stabilization by trivial bands, and the
appropriate symmetry relations.  Those are exactly the identifications that
a theory of phases is meant to impose.

A phase classification and a determination of the physics within a phase are
different problems.  Two Hamiltonians in the same \(K\)-class may have
different dispersions, densities of states, localization lengths, optical
matrix elements, Berry-curvature profiles, quantum metrics, and responses away
from a quantized limit.  In Euclidean band theory this difference is familiar,
but the simplicity of the Brillouin torus can make it appear secondary.  In
hyperbolic band theory the distinction occurs at the first step: the objects
used as momenta already vary through moduli spaces that ordinary topological
\(K\)-theory collapses to a single class.

The article therefore has two closely related aims.  The first is expository:
to identify the geometric and holomorphic structures native to hyperbolic
band theory, show where they enter the Hamiltonian and its observables, and
explain why they cannot be replaced by a stable phase label.  The second is
mathematical: to make that explanation exact.  For each observable we specify
the information being forgotten, ask whether the observable descends through
the corresponding forgetful map, and prove either factorization or
nonfactorization.  The geometry must first be shown to be dynamically active,
rather than merely an elegant parametrization, before its loss can have
physical content.

Let
\[
    X=\HH^2/\Gamma
\]
be a closed hyperbolic surface of genus \(g\geq 2\), with \(\Gamma\) a
torsion-free cocompact Fuchsian group.  In the abelian sector, a hyperbolic
quasimomentum is a character
\[
   \chi:\Gamma\longrightarrow \UU(1).
\]
It defines a degree-zero flat line bundle \(L_\chi\), and the set of such
characters is the underlying real torus of the Jacobian:
\[
   \Hom(\Gamma,\UU(1))
   \cong H^1(X;\RR)/2\pi H^1(X;\ZZ)
   \cong \Jac(X).
\]
Yet every \(L_\chi\) is topologically trivial.  In particular,
\[
             [L_\chi]=[\cO_X]\qquad\text{in }K^0(X)
             \tag{1.1}\label{eq:intro-line-collapse}
\]
for every \(\chi\), where \(\cO_X\) denotes the trivial holomorphic line
bundle.  If a rank-one hyperbolic band energy
\(E_n(\chi)\) is nonconstant, then it cannot possibly be a function of the
class in \eqref{eq:intro-line-collapse}.

The same phenomenon persists nonabelianly.  An irreducible unitary
representation
\[
       \rho:\Gamma\longrightarrow \UU(r)
\]
corresponds, through the Narasimhan--Seshadri theorem, to a stable
degree-zero holomorphic bundle \(E_\rho\) on \(X\)
\cite{NarasimhanSeshadri,DonaldsonNS}.  At fixed rank,
\[
        [E_\rho]=r[\cO_X]\qquad\text{in }K^0(X),
        \tag{1.2}\label{eq:intro-rank-collapse}
\]
while the stable bundles move in a moduli space of complex dimension
\[
                  r^2(g-1)+1
\]
when the determinant is allowed to vary.  Thus the nonabelian Bloch sectors
live in information that the fibrewise \(K\)-class necessarily forgets.

Equations \eqref{eq:intro-line-collapse} and
\eqref{eq:intro-rank-collapse} are elementary.  They are the topological
starting point, not the conclusion.  Their physical force depends on showing
that the forgotten flat and holomorphic structures change an observable.
Conversely, a list of changing observables would not by itself identify the
precise loss of information.  Observable factorization supplies the bridge:
across three distinct uses of \(K\)-theory, we ask whether spectral and
geometric observables descend to the relevant \(K\)-class.  For discrete
Hamiltonians the basic mechanism is the closed-walk expansion
\[
   \Tr(H_\rho^n)
      =\sum_{\substack{w\ {\rm closed}\\ |w|=n}}
          A_w\,\tr\rho(\gamma_w),
          \tag{1.3}\label{eq:intro-wilson}
\]
and for continuum Hamiltonians by an equivariant heat-kernel formula.  The
coefficients \(A_w\) contain hopping and orbital information, while the Wilson
functions \(\tr\rho(\gamma_w)\) vary on the character variety.  In the
continuum formula the analogous term is weighted by the displacement, and in
constant curvature by the length, of the closed geodesic represented by
\(\gamma_w\).  Geometry and holonomy therefore enter the spectrum before one
asks for its stable topological class.

The continuum theory provides an especially sharp complex-geometric example.
Write
\(\Delta_a=(\dd+\ii a)^*(\dd+\ii a)\) for the scalar Laplacian twisted by a
harmonic real one-form \(a\), and let \(\lambda_0(a)\) denote its lowest
eigenvalue near \(a=0\).  If \(\alpha,\beta\) are harmonic one-forms, we prove
in the normalization used here that
\[
 \operatorname{Hess}_0(\lambda_0)(\alpha,\beta)
   =\frac{2}{\operatorname{Area}(X)}
      \int_X\alpha\wedge *\beta .
\tag{1.4}\label{eq:intro-hodge-hessian}
\]
The identity is an instance of the twisted-ground-state Hessian calculation
used by Kotani and Sunada in their analysis of heat kernels on abelian
coverings \cite{KotaniSunada2000}.  We include a self-contained derivation in
\Cref{thm:hodge-hessian} because its hyperbolic-band interpretation is central
here.  The right-hand side is the Hodge metric on \(H^1(X;\RR)\).  Together
with the integral intersection form, it determines the Hodge star, the period
matrix, and the principally polarized Jacobian with its induced homology
marking.  Torelli and uniformization then recover the underlying complex and
hyperbolic surface, but not the full Teichm\"uller marking: the Torelli group
acts trivially on integral homology \cite{FarbMargalit}.  This reconstruction
mechanism is closely related to the magnetic ground-state determination of a
surface's conformal class proved by Colbois, Provenzano, and Savo
\cite{ColboisProvenzanoSavo2025}.  Thus both the Hessian formula and the
Torelli implication have established antecedents.  Here they are interpreted
within hyperbolic band theory and placed in the \(K\)-theoretic information
hierarchy.  The class of the associated flat line bundle remains
\([\cO_X]\).  This example captures the two roles of the manuscript in a
single calculation: it explains concretely why \(K\)-invisible holomorphic
data matter to the physics, while the surrounding factorization framework
determines exactly what is forgotten and what the measured band data recover.

\subsection{Kinematical and dynamical holomorphy}

To see when holomorphy affects the physics, one must first distinguish two
roles of the complex structure.  It may organize generalized momenta without
affecting a particular Hamiltonian, or it may enter the kinetic operator,
couplings, spectral projection, and response relations.  The former role is
kinematical; the latter is dynamical.  Hyperbolic band theory supplies both
the geometric parameter spaces and the mechanisms by which their geometry can
become observable.

The foundational construction of hyperbolic band theory identifies abelian
Bloch states with automorphic states parametrized by a higher-dimensional
Jacobian \cite{MaciejkoRayan2021}.  Hyperbolic crystallography supplies the
unit cells, translation groups, side pairings, and finite Bloch matrices that
make this construction calculable \cite{BoettcherEtAl2022}.  The automorphic
Bloch theorems show how periodic boundary conditions on finite covers select
both one-dimensional and higher-dimensional irreducible representations
\cite{MaciejkoRayan2022}.  The higher-rank picture naturally leads to stable
bundles and their moduli.  Higgs bundles then enlarge the proposal by
introducing spectral curves and a complex form of momentum
\cite{KienzleRayan2022}.  Finally, the hyperbolic Bloch transform places the
representation-theoretic construction in a noncommutative
harmonic-analytic framework, organizes finite-rank transforms into an
asymptotic construction, and sends the geometric Laplacian to covariant
Laplacians on flat bundles \cite{NagyRayan2024}.

These constructions also belong to a broader analytic tradition.
Koc\'abov\'a and \v{S}\v{t}ov\'i\v{c}ek developed generalized Bloch
decompositions for Riemannian manifolds with a countable discrete symmetry
group, while Gruber formulated periodic elliptic theory through Hilbert
modules over group \(C^*\)-algebras
\cite{KocabovaStovicek2008,Gruber2001}.  These general frameworks complement
the character-space and algebro-geometric constructions used here.

The phrase \emph{noncommutative Bloch theory} also has an earlier,
operator-algebraic meaning.  Marcolli and Mathai studied elliptic operators
on good-orbifold covers that are invariant under a projective action of a
cocompact Fuchsian group.  Their twisted index theory, together with the
\(K\)-theory and canonical trace of twisted group \(C^*\)-algebras, yields
constraints on spectra and spectral gaps in real and complex hyperbolic
settings \cite{MarcolliMathai1999}.  This is not identical to the
representation-space transform of \cite{NagyRayan2024}: the former packages
projectively covariant observables in a noncommutative algebra, whereas the
latter uses finite-dimensional representation sectors as geometric fibres
in an asymptotic Bloch construction.  The two viewpoints are complementary,
and their distinction will matter below when we separate an observable
algebra from its representation-dependent spectral data.  In the second
part of their study, Marcolli and Mathai compute the
higher cyclic pairing underlying the hyperbolic Connes--Kubo formula and
obtain Hall values governed by the orbifold Euler characteristic
\cite{MarcolliMathai2001}.  Their later survey makes explicit the passage from
the Euclidean magnetic noncommutative torus to the hyperbolic good-orbifold
model \cite{MarcolliMathai2006}.  Taken together, these works form an
important operator-algebraic precursor to the questions considered here.

The same programme was presented, from the outset, as part of a broader
algebraic-geometric approach to quantum matter in
\cite{RayanQuantumMatter2026}.  There the Abel--Jacobi map, moduli of bundles,
and arithmetic structure are not decorative analogies: they provide the
natural passage between higher-genus configuration space and generalized
momentum data.  The present paper develops the complementary point that this
passage contains physically relevant information strictly finer than stable
topology.

These works do not say that an abelian Jacobian, or any fixed finite-rank
character space, exhausts the thermodynamic spectrum of a hyperbolic lattice.
Open hyperbolic patches have an extensive boundary, and different limiting
procedures probe different parts of the spectral problem.  Mosseri and Vidal
showed directly that rank-one Bloch-like states occupy a vanishing fraction of
the thermodynamic spectrum \cite{MosseriVidal2023}.  Lux and Prodan explained
more generally why finite-dimensional representations can produce spurious
spectra and why the reduced, regular-representation spectrum must instead be
recovered from coherent limiting sequences
\cite{LuxProdanCayley2024,LuxProdan2023}.  Shankar and Maciejko then proved
that the normalized \(\UU(r)\) hyperbolic-Bloch density-of-states moments
converge, for every fixed moment order, to the exact infinite-lattice moments
as \(r\to\infty\) \cite{ShankarMaciejko2024}.  Thus finite-rank character
spaces remain meaningful sector-resolved probes and approximation spaces,
but no fixed finite-rank sector carries nonzero weight in the thermodynamic
bulk density of states.

Supercell constructions recover nonabelian sectors systematically
\cite{LenggenhagerEtAl2023}.  Magnetic and topological models further show how
generalized momentum and negative curvature affect Dirac cones and topological
phases \cite{IkedaAokiMatsuki,UrwylerEtAl2022}, while nonabelian sectors
support phenomena with no direct two-dimensional Euclidean analogue
\cite{TummuruEtAl2024}.  This broader picture also determines how the
nonfactorization results below should be read.  They apply exactly within each
representation sector and to finite periodic devices.  Thermodynamic bulk
averages, by contrast, arise only after passage to the regular group trace
through a coherent finite-cover or large-rank limit.

One further point is essential.  A Jacobian has a complex structure and a
principal polarization, but a finite matrix \(H_\chi\) built only from group
labels and fixed hopping constants can be independent of that complex
structure after the character torus is identified by a marking.  In such a
model, complex geometry organizes the parameters but does not yet change an
observable.  It becomes dynamically active when the kinetic operator,
coupling strengths, parallel transport, orbital alignment, spectral
projection, or response protocol depends on the hyperbolic metric or on the
associated complex structure.  This distinction between
\emph{kinematical holomorphy} and \emph{dynamical holomorphy} will be maintained
throughout.

The distinction is experimentally relevant.  Hyperbolic and
topological-hyperbolic networks have been implemented in circuit,
photonic, and electrical platforms
\cite{KollarFitzpatrickHouck,YuPiaoPark,ZhangEtAl2022}.  In a recent
superconducting-circuit construction, the hyperbolic metric is encoded
directly into capacitive couplings, and spectral and localization properties
are measured for genus-two and genus-three designs \cite{XuEtAl2025}.  Such an
implementation changes neither the abstract genus nor the relevant
\(K\)-groups when the couplings are varied.  It nevertheless changes the
Hamiltonian.  This is precisely the regime in which geometry can be isolated
from topology in the laboratory.  This gives the paper's explanatory thesis
a direct experimental form: hold the stable topological data fixed, vary a
geometric or holomorphic input through a specified coupling, and measure an
observable that the ensuing sections prove cannot descend to the
\(K\)-class.

\subsection{Principal contributions and relation to prior work}

This subsection separates established mechanisms from the results developed
here.  The first principal mathematical contribution is a unified
nonfactorization framework.  Given a forgetful map from geometric or
Hamiltonian data to a \(K\)-class and an observable \(\mathcal O\), we ask
whether \(\mathcal O\) descends through that map.  In several settings the
answer is negative, under explicit hypotheses and with explicit witnesses.
The resulting statements concern spectra and Wilson moments, Higgs spectral
curves, Berry holonomy and the quantum metric, the partially filled Hall response,
and Fermi surface and nodal geometries.  The point is not the familiar slogan
that topology omits dispersion; it is the common theorem-level formulation
across fibrewise \(K^0(X)\), occupied-family \(K^0(B)\), and
operator-algebraic \(K_*(\mathcal A)\), together with proofs identifying the
data responsible for each failure of descent.

The Hodge--Torelli mechanism is the sharpest continuum illustration of this
first contribution, and its analytic core comes from earlier work.  Kotani
and Sunada computed the twisted-ground-state Hessian in their study of
abelian-cover heat-kernel asymptotics \cite{KotaniSunada2000}.  Colbois,
Provenzano, and Savo subsequently used weak magnetic ground-state data, the
Hodge Gram matrix, and Torelli theory to recover the conformal class of a
surface \cite{ColboisProvenzanoSavo2025}.  We give the calculation in the
normalization of hyperbolic Bloch characters, extract the homologically
marked infinitesimal reconstruction statement needed here, and connect it
with the discrete Wilson, Higgs, and quantum-geometric nonfactorization
results.

The second principal mathematical contribution is a systematic framework
separating finite-rank sector geometry from thermodynamic bulk spectral data.
The vanishing thermodynamic weight of fixed finite-dimensional sectors and
the necessity of large-rank or coherent limits are known from the work of
Mosseri--Vidal, Lux--Prodan, Shankar--Maciejko, and Nagy and the author
\cite{MosseriVidal2023,LuxProdanCayley2024,LuxProdan2023,
ShankarMaciejko2024,NagyRayan2024}.  The finite-cover part lies in the
classical tradition of approximation along residual towers and local spectral
convergence \cite{Luck1994,AbertThomVirag2014}; the exclusion of a fixed
finite-dimensional representation from the reduced problem follows from
standard weak-containment and amenability theory
\cite{Hulanicki1964}.  We derive the exact finite-propagation specialization
needed for finite-range Hamiltonians and its quantitative
continuous-functional-calculus corollary.  For arithmetic principal
congruence towers, we combine the latter with the logarithmic systole bound of
Katz, Schaps, and Vishne to obtain explicit rates in the quotient size
\cite{KatzSchapsVishne2007}.  We also record the short weak-containment
argument in the notation of hyperbolic band theory.  These results connect
sector-resolved calculations, finite periodic devices, and the canonical
regular trace within one framework.

Accordingly, the paper's own contribution is not the underlying approximation
or amenability principles, nor the established flat-bundle \(K\)-collapse,
closed-walk and heat-kernel expansions, bottom-band Hessian, or large-rank
thermodynamic limit.  It lies in the factorization criteria and
counterexamples across the discrete, continuum, holomorphic, and
operator-algebraic settings, and in their synthesis with the sector--bulk
distinction in the strict information hierarchy of
\Cref{thm:strict-hierarchy}.  The finite-cover identities and
weak-containment proposition make the interfaces of that hierarchy exact.
Together, these ingredients permit a direct comparison of the information
retained by geometric sectors, stable topology, finite periodic models, and
the thermodynamic bulk.

The analysis proceeds through four layers, in increasing order of
physical content.  Each layer first identifies what geometric or holomorphic
information is present and then tests, by a theorem or an explicit witness,
whether the relevant observable survives its removal.

\begin{enumerate}[label=\textup{(\Roman*)},leftmargin=2.8em]
\item \textbf{Foundational \(K\)-theoretic collapse.}
  The map from rank-\(r\) unitary Bloch sectors to the \(K\)-class of the
  corresponding flat bundle is constant.  This is proved in
  \Cref{thm:k-collapse}.  We also explain why the universal Poincar\'e bundle
  over \(X\times\Jac(X)\) is not topologically trivial, even though each of its
  fibres over \(\Jac(X)\) has the same class.  Families \(K\)-theory retains
  mixed topological information, but not the pointwise holomorphic or spectral
  data.  The first strict loss occurs when the flat connection or Dolbeault
  operator is forgotten, which locates precisely why holomorphic information
  is finer than the fibrewise class.

\item \textbf{Spectral nonfactorization.}
  For group-labelled finite Hamiltonians, \Cref{thm:walk-formula} proves
  \eqref{eq:intro-wilson}.  Since finitely many moments determine a finite
  spectrum, any nonconstant Wilson contribution obstructs factorization of
  the spectrum through \(K^0(X)\); see
  \Cref{thm:spectral-nonfactorization}.  \Cref{prop:heat-trace} gives the
  continuum heat-kernel counterpart.  \Cref{thm:abelian-blindness} shows
  exactly what is lost by restricting to rank-one momentum: labels in the
  commutator subgroup disappear, although rank-two sectors can recover their
  spectral effect.  In \Cref{thm:metric-nonfactorization} we vary the
  hyperbolic metric on a fixed topological surface and obtain different
  Laplace spectra with identical \(K\)-data.

\item \textbf{Geometric reconstruction and holomorphic refinement.}
  The missing variation is organized by flat connections, Wilson functions,
  geodesic lengths, period matrices, Higgs fields, spectral curves, Berry
  connections, and quantum metrics.  The infinitesimal metric response is
  given in \Cref{prop:metric-variation}, while \Cref{thm:hodge-hessian}
  restates and derives the Kotani--Sunada Hessian identity in the present
  normalization, and \Cref{cor:torelli-from-band} records its homologically
  marked Hodge--Torelli consequence.  These results provide the continuum
  model against which the surrounding nonfactorization framework is developed.
  The underlying \(K\)-class of a Higgs bundle is constant on fixed-rank and
  fixed-degree moduli, while its spectral curve varies; this is made precise
  in \Cref{prop:higgs-collapse}.  For an isolated band,
  \Cref{thm:qgt-not-k} constructs \(K\)-equivalent families with different
  dispersion, Berry holonomy, curvature distribution, and quantum metric.
  The trace and determinant bounds, together with the ideal-band saturation
  condition, have antecedents in the Chern-band literature
  \cite{Roy2014,MeraOzawaKahler,MeraOzawaFlat}.
  \Cref{prop:pseudoholomorphic-bound} gives their intrinsic formulation on a
  parameter Riemann surface and identifies saturation with local
  pseudoholomorphy.  These results and examples explain not merely that
  geometric data have been discarded, but how they reappear as measurable
  spectral or response data.

\item \textbf{Condensed-matter consequences and explicit witnesses.}
  A fully occupied gapped band has quantized Hall response determined by a
  Chern pairing, and a stable Fermi or nodal manifold may carry a local
  \(K\)-theoretic charge.  These are positive results for topology, not
  exceptions to the thesis.  We show in \Cref{prop:hall-dichotomy} that
  partially filled Hall response still depends on dispersion and the local
  Berry-curvature profile.  \Cref{prop:fermi-coarea} relates the density of
  states to the metric geometry of Fermi level sets, and
  \Cref{prop:clifford-nodes} separates the charge of a Dirac or Weyl-type node
  from its position, velocity tensor, tilt, and nodal embedding.  This
  distinction is applied to the hyperbolic nonabelian semimetal.
\end{enumerate}

Two later results sharpen the interfaces with operator algebras and finite
computations.  \Cref{prop:gap-label-scope} separates the deformation-invariant
label of a gap from its location, width, and surrounding spectral measure.
\Cref{prop:differential-k-not-band} shows that even a differential
\(K\)-class can be held fixed while dispersion or quantum metric varies.
Finally, \Cref{prop:finite-cover-moments} records the exact
finite-propagation identity for every fixed spectral moment and identifies
the representation-theoretic weighting implicit in compact periodic
quotients.  \Cref{cor:finite-cover-functional-calculus} turns it into a
quantitative trace estimate for continuous spectral observables in terms of
polynomial approximation and the local injectivity length of the cover, while
\Cref{cor:arithmetic-congruence-rate} gives explicit quotient-size rates for
arithmetic principal congruence towers.
\Cref{prop:no-finite-dimensional-bulk} then states, in the present setting,
the complementary classical weak-containment consequence: no fixed
finite-dimensional representation is a sector of the reduced
regular-representation problem, even though coherent finite-dimensional
families can recover its moments.

It is important that three different uses of \(K\)-theory appear in this
subject:
\[
\begin{array}{ccl}
K^0(X) &:& \text{the topological class of a flat or holomorphic bundle on }X,\\
K^0(B) &:& \text{the occupied-state bundle over a Bloch or moduli space }B,\\
K_*(\mathcal A) &:& \text{the stable topology of an observable algebra }\mathcal A.
\end{array}
\tag{1.5}\label{eq:three-k}
\]
They should not be conflated.  We treat all three.  The conclusion is the same
in each case, for different reasons: \(K\)-theory records robust stable data,
not a metric, a connection, a spectrum, or a Hamiltonian.

There is also a categorical distinction behind \eqref{eq:three-k}.  A Bloch
sector is naturally an object of a representation groupoid, an occupied band
is an object of the symmetric monoidal category of vector bundles on \(B\),
and operator-algebraic \(K\)-theory is built from projective modules or
\(K\)-cycles over \(\mathcal A\).  Passing to orbit spaces, isomorphism
classes, and Grothendieck groups are different operations.  We use categorical
language below where it makes these operations precise: gauge transformations
are morphisms, quotient stacks retain stabilizers, forgetting structure is a
functor, and \(K\)-theory is a stabilized, group-completed target.  The
nonfactorization results can then be read as statements that specified
observables do not descend along those functors.

Three logical levels will remain separate throughout.  The closed-walk,
flat-bundle, Higgs, and quantum-geometric results are exact statements about
a specified finite-dimensional sector or control family; they do not require
that family to carry thermodynamic spectral weight.  Statements about the
infinite lattice instead use the canonical trace of the regular
representation and require a coherent finite-cover or large-rank
approximation whenever they are inferred from finite matrices.  Finally, a
complex or holomorphic parametrization becomes physically restrictive only
after the Hamiltonian or spectral projector is coupled to it.  Keeping these
three levels distinct prevents sectorwise geometry, bulk averaging, and
dynamical holomorphy from being used as substitutes for one another.

\subsection{Organization}

The order of the paper follows a loss-and-recovery argument.  We first retain
the full geometric datum and then forget structure in controlled stages.
\Cref{sec:bloch-data} separates marked topology, hyperbolic and complex
geometry, quotient combinatorics, Hamiltonian coefficients, and
representation sectors.  \Cref{sec:k-collapse} applies the first forgetful
map, proves the fibrewise collapse, and uses the Poincar\'e family to explain
what families \(K\)-theory still retains.  \Cref{sec:walks} then establishes
that the discarded connection data remain spectrally active through
closed-walk Wilson functions, equivariant heat traces, and flux response.

The next three sections develop the geometric and holomorphic information
that is recovered from, or imposed on, the band data.
\Cref{sec:metric-complex} separates topology from metric and complex
structure and presents the Hodge--Torelli reconstruction example.
\Cref{sec:holomorphy} treats stable bundles, Higgs fields, spectral curves,
and the moduli-stack issue.  \Cref{sec:quantum-geometry} passes from the
holomorphic geometry of sector spaces to the differential geometry of
spectral projectors, including pseudoholomorphic bands.
\Cref{sec:fermi-hall} then tests the distinction against physical response:
it identifies the quantized regimes in which topology is sufficient and the
partially filled or gapless regimes in which further geometry is required.
It also marks the limit of the single-particle framework by distinguishing
topological semimetals from interacting metallic regimes that require
frequency-dependent Green-function and correlation data.

The final part tests whether more sophisticated topological or limiting
constructions repair the loss.  \Cref{sec:operator-k} treats
operator-algebraic and differential \(K\)-theory, including noncommutative
tori, twisted surface-group algebras, gap labels, and differential
refinements.  \Cref{sec:quasicrystals-experiments} compares the argument with
quasicrystals, separates finite systems from the regular-trace bulk, proves
the finite-cover approximation results, and proposes experiments.
\Cref{sec:completion} then recombines the preceding analysis into an
observable-dependent information hierarchy and practical geometric records
for the observables under study.  Two appendices supply the surface
\(K\)-theory computation and explicit genus-two models.

\section{Hyperbolic Bloch theory as a hierarchy of data}
\label{sec:bloch-data}

\subsection{The surface group and a marked hyperbolic crystal}

Let \(S\) be a closed, connected, oriented topological surface of genus
\(g\geq2\).  A marking identifies its fundamental group with
\[
 \Gamma_g=
 \left\langle
 a_1,b_1,\ldots,a_g,b_g\ \middle|\
 \prod_{j=1}^{g}[a_j,b_j]=1
 \right\rangle.
 \tag{2.1}\label{eq:surface-group}
\]
A point of Teichm\"uller space \(\cT_g\) equips \(S\) with a complex structure
\(J\), or equivalently with the unique compatible metric \(h_J\) of curvature
\(-1\).  We write \(X_J=(S,J)\) and, after choosing a Fuchsian realization,
\[
               X_J\cong \HH^2/\Gamma_J .
\]
The abstract marked group remains \(\Gamma_g\), while its embedding
\(\Gamma_J\subset\operatorname{PSL}(2,\RR)\) varies with \(J\).
More explicitly, the marking supplies an isomorphism
\(\iota_J:\Gamma_g\to\Gamma_J\).  Whenever the abstract group
\(\Gamma_g\) is written as acting on \(\HH^2\), this action is understood
through \(\iota_J\).

We use \(S\) for the underlying smooth oriented surface and \(X_J\) for the
same surface equipped with its complex structure.  Once \(J\) has been
fixed, we often abbreviate
\[
                    X=X_J,\qquad h=h_J,\qquad \Gamma=\Gamma_J.
                    \tag{2.1a}\label{eq:standing-abbreviations}
\]
Thus \(K^0(S)\) and \(K^0(X)\) denote the same topological \(K\)-group, while
the notation \(X\) signals that holomorphic data are in use.  We write
\(\operatorname{Herm}(N)\) for the Hermitian \(N\times N\) matrices and
\(\operatorname{Gr}(q,N)\) for the Grassmannian of \(q\)-planes in
\(\CC^N\).  Capital \(\Tr\) denotes a full Hilbert-space trace; lower-case
\(\tr\), with a subscript when needed, denotes a trace over an internal
orbital or representation fibre.

A hyperbolic crystal contains more information than \(S\) and \(J\).  One may
start with a \(\Gamma_J\)-invariant graph \(\widetilde{\cG}\subset\HH^2\), with
a finite quotient graph
\[
                   \mathcal Q=\widetilde{\cG}/\Gamma_J .
\]
Each quotient edge remembers which translate of a chosen fundamental domain
contains its endpoint.  Once lifts of quotient vertices are chosen, this is
recorded by a label \(\gamma_e\in\Gamma_g\) on every oriented edge \(e\), with
\(\gamma_{\bar e}=\gamma_e^{-1}\).  Changing the chosen lifts performs a
vertex gauge transformation on the labels.  Hyperbolic crystallography
systematizes precisely this passage from a regular tessellation to a
translation group, fundamental cell, and finite Bloch matrix
\cite{BoettcherEtAl2022}.

The labels determine how a representation of \(\Gamma_g\) enters a boundary
condition.  They do not determine hopping strengths.  Those may depend on
hyperbolic distances, edge angles, orbital orientations, magnetic transport,
or fabrication parameters.  We therefore separate:
\[
\boxed{
\begin{minipage}{0.88\linewidth}
\begin{enumerate}[label=\arabic*.,leftmargin=1.8em]
\item the marked topology \((S,\Gamma_g)\);
\item the hyperbolic and complex geometry \((h_J,J)\);
\item the quotient combinatorics and group labels \((\mathcal Q,\gamma)\);
\item the Hamiltonian data, including hoppings, on-site terms, and transport;
\item the representation or flat connection specifying a Bloch sector.
\end{enumerate}
\end{minipage}}
\tag{2.2}\label{eq:input-layers}
\]
Only after these choices does one obtain spectral projectors and their
topological classes.

\subsection{The abelian Brillouin torus and its complex structure}

Every character of \(\Gamma_g\) factors through the abelianization
\[
     \Gamma_g^{\mathrm{ab}}\cong H_1(S;\ZZ)\cong\ZZ^{2g}.
\]
Hence
\[
 B_{\mathrm{ab}}
 :=\Hom(\Gamma_g,\UU(1))
 \cong \Hom(H_1(S;\ZZ),\UU(1))
 \cong (\RR/2\pi\ZZ)^{2g}.
 \tag{2.3}\label{eq:abelian-bz}
\]
If \(\{[a_j],[b_j]\}_{j=1}^g\) is the symplectic homology basis induced by
\eqref{eq:surface-group}, a character is described by its flux angles
\[
  \chi_{\bm\theta}(a_j)=\e^{\ii\theta_{a_j}},
  \qquad
  \chi_{\bm\theta}(b_j)=\e^{\ii\theta_{b_j}}.
\]
This real torus is defined by topology and a marking.

The complex structure \(J\) adds nontrivial information.  Hodge star on
harmonic one-forms satisfies \(*_{h_J}^2=-1\), and therefore makes
\[
 H^1(S;\RR)/2\pi H^1(S;\ZZ)
\]
into a complex torus.  Dividing harmonic representatives by \(2\pi\) gives
the standard algebro-geometric normalization.  After normalizing a basis of
holomorphic one-forms, one then obtains a period matrix \(\Omega_J\) in the
Siegel upper half-space and
\[
               \Jac(X_J)\cong
               \CC^g/(\ZZ^g+\Omega_J\ZZ^g).
               \tag{2.4}\label{eq:period-jac}
\]
The intersection form determines the principal polarization.  Thus the
underlying real torus in \eqref{eq:abelian-bz} is independent of \(J\), while
its complex structure and polarized isomorphism class need not be
\cite{FarkasKra,BirkenhakeLange}.

For \(\chi\in B_{\mathrm{ab}}\), define
\[
 L_\chi=(\HH^2\times\CC)/\Gamma_g,
 \qquad
 \gamma\cdot(z,v)=(\gamma z,\chi(\gamma)v).
 \tag{2.5}\label{eq:flat-line}
\]
The flat unitary connection makes \(L_\chi\) a holomorphic line bundle of
degree zero.  This identifies the real character torus with
\(\Pic^0(X_J)=\Jac(X_J)\).  The identification is canonical once \(J\) is
fixed, but its holomorphic content depends on \(J\).

\subsection{Nonabelian sectors}

For \(r\geq1\), let
\[
 \cR_r(S):=\Hom(\Gamma_g,\UU(r))/\UU(r)
 \tag{2.6}\label{eq:unitary-character}
\]
be the unitary character space, where the quotient is by simultaneous
conjugation.  It is compact and stratified by representation type.  Its
irreducible locus \(\cR_r^{\mathrm{irr}}(S)\) is smooth after dividing by
\(\operatorname{PU}(r)\), and carries the Atiyah--Bott--Goldman symplectic
structure \cite{AtiyahBott,GoldmanSymplectic}.

A representation \(\rho\) defines the flat bundle
\[
 E_\rho=(\HH^2\times\CC^r)/\Gamma_g,
 \qquad
 \gamma\cdot(z,v)=(\gamma z,\rho(\gamma)v).
 \tag{2.7}\label{eq:flat-vector}
\]
For fixed \(J\), the Narasimhan--Seshadri correspondence identifies the
irreducible representations with stable holomorphic vector bundles of rank
\(r\) and degree zero.  Allowing polystable bundles completes the
correspondence across the reducible strata.  The stable moduli space
\(\cM_X^s(r,0)\) has complex dimension
\[
       \dim_\CC\cM_X^s(r,0)=r^2(g-1)+1,
       \tag{2.8}\label{eq:moduli-dim}
\]
and the fixed-determinant locus has dimension
\((r^2-1)(g-1)\).

Unlike \(B_{\mathrm{ab}}\), the space \(\cR_r(S)\) is not a group for
\(r>1\), and no single global set of momentum coordinates linearizes it.
Wilson functions
\[
          W_\gamma([\rho])=\tr\rho(\gamma),
          \qquad \gamma\in\Gamma_g,
          \tag{2.9}\label{eq:wilson-functions}
\]
provide conjugation-invariant functions and will play the role of
nonabelian Fourier modes.  They separate semisimple complex representations
when all words are allowed.  A finite Hamiltonian samples only finitely many
of them at any fixed spectral moment.

\subsection{Three bases and two different bundles}

Two bundles enter discussions of topology and are easily confused.

\begin{enumerate}[label=\textup{(\alph*)},leftmargin=2.5em]
\item The bundle \(E_\rho\to X\) belongs to a single representation sector.
  Its base is position space.  Its \(K\)-class is the object in
  \eqref{eq:intro-rank-collapse}.

\item If \(B\subset\cR_r(S)\) is a smooth region on which a collection of
  eigenvalues is separated by a gap, the corresponding spectral projectors
  may define an occupied-state bundle \(\mathscr E_{\mathrm{occ}}\to B\), or
  more naturally an equivariant or stack-theoretic bundle.  Its base is
  momentum or representation space, and its class lies in \(K^0(B)\), not in
  \(K^0(X)\).
\end{enumerate}

The first \(K\)-class cannot distinguish representation sectors at fixed
rank.  The second may carry Chern classes and classify stable band topology,
but it does not determine eigenvalues or differential geometry.  The
operator-algebraic \(K\)-groups in \eqref{eq:three-k} are a third object and
will be treated in \Cref{sec:operator-k}.

\begin{definition}[Geometric hyperbolic band datum]
\label{def:geometric-band-datum}
A \emph{geometric hyperbolic band datum} consists of a marked surface
\((S,\Gamma_g)\), a hyperbolic complex structure \((h_J,J)\), a finite
group-labelled quotient graph \((\mathcal Q,\gamma)\) or a continuum kinetic
operator, Hamiltonian coefficients \(\mathsf t\), and a smooth sector map
\(\rho_B:B\to\cR_r(S)\) from a parameter space \(B\).  In a finite model,
\(\mathsf t\) denotes the collection of on-site and hopping matrices,
including any metric-dependent transport rule.  Where a rank-\(q\) spectral
cluster is separated by a gap, the datum also includes its projector
\(P:B\to\operatorname{Gr}(q,N)\).
\end{definition}

There are forgetful maps from this datum to flat or occupied bundles, from
bundles with connection to topological bundles, and from topological bundles
to \(K\)-classes.  This definition is the common source for both strands of
the paper: it displays the geometry and holomorphy that the physical model
may use, and it supplies the domain of the forgetful maps tested by the new
nonfactorization results.  The next section locates the first strict loss;
\Cref{sec:walks} then proves that the discarded structure remains visible in
the spectrum.

\section{The \texorpdfstring{\(K\)}{K}-theoretic collapse of Bloch sectors}
\label{sec:k-collapse}

\subsection{Complex \texorpdfstring{\(K\)}{K}-theory of a closed surface}

The complex \(K\)-theory of a closed oriented surface is especially simple:
\[
       K^0(S)\cong \ZZ\oplus\ZZ,
       \qquad
       K^1(S)\cong\ZZ^{2g}.
       \tag{3.1}\label{eq:k-surface}
\]
The two components of \(K^0(S)\) may be taken to be rank and first Chern
number.  More precisely,
\[
 \widetilde K^0(S)\xrightarrow[\cong]{\ c_1\ }H^2(S;\ZZ)\cong\ZZ .
 \tag{3.2}\label{eq:reduced-k}
\]
This follows either from obstruction theory for stable unitary bundles or
from the Atiyah--Hirzebruch spectral sequence; details are recalled in
\Cref{app:k-surface}.  Consequently two complex vector bundles over \(S\)
have the same stable class exactly when they have the same rank and degree.

\begin{theorem}[\(K\)-theoretic collapse]
\label{thm:k-collapse}
Let \(S\) be a closed oriented surface of genus \(g\geq2\).
For every unitary representation \(\rho:\Gamma_g\to\UU(r)\), the associated
flat bundle satisfies
\[
                    [E_\rho]=r[\underline{\CC}]
                    \quad\text{in }K^0(S).
                    \tag{3.3}\label{eq:k-collapse}
\]
Here \(\underline{\CC}=S\times\CC\) is the trivial smooth complex line
bundle.
After a complex structure is chosen, the same statement holds for every
polystable holomorphic bundle of rank \(r\) and degree zero.  Therefore the
map
\[
  \kappa_r:\cR_r(S)\longrightarrow K^0(S),
  \qquad [\rho]\longmapsto[E_\rho],
  \tag{3.4}\label{eq:kappa-map}
\]
is constant.
\end{theorem}

\begin{proof}
The flat unitary connection on \(E_\rho\) has zero curvature.  Chern--Weil
theory therefore gives \(c_1(E_\rho)=0\) in \(H^2(S;\RR)\).  Since
\(H^2(S;\ZZ)\cong\ZZ\) is torsion-free, the integral first Chern class also
vanishes.  The class of a complex vector bundle on a two-dimensional CW
complex is determined stably by its rank and first Chern class, by
\eqref{eq:reduced-k}.  Hence \([E_\rho]\) is the class of a trivial rank-\(r\)
bundle.  A degree-zero holomorphic bundle has the same rank and vanishing
first Chern number, so the second assertion follows as well.
\end{proof}

\begin{corollary}[No fibrewise \(K\)-complete observable]
\label{cor:k-observable-constant}
Let \(Y\) be any set and let
\(\mathcal O:\cR_r(S)\to Y\) be an observable that factors through
\(\kappa_r\).  Then \(\mathcal O\) is constant.  In particular, a spectral,
dynamical, or response quantity that varies with \(\rho\) is not determined
by the \(K^0(S)\)-class of \(E_\rho\).
\end{corollary}

\begin{proof}
If \(\mathcal O=F\circ\kappa_r\), then
\(\mathcal O([\rho])=F(r[\underline{\CC}])\) for every \([\rho]\), by
\Cref{thm:k-collapse}.
\end{proof}

For \(r=1\), \Cref{thm:k-collapse} says that the entire Jacobian is sent to
one point of \(K^0(S)\).  For \(r>1\), it collapses every stable moduli space
in \eqref{eq:moduli-dim}, as well as all of its polystable strata.  No
subtle property of \(K\)-theory is involved: the collapse is forced by the
classification of stable complex bundles on a surface.

\subsection{What the collapse does not imply}

There are three reasons not to overstate \Cref{thm:k-collapse}.

First, equality in \(K^0(S)\) is stable equality.  It does not assert that two
flat connections are gauge equivalent, that two holomorphic bundles are
holomorphically isomorphic, or that their spaces of sections are canonically
identified.  On a surface, a degree-zero smooth bundle is topologically
trivial, but a flat or Dolbeault operator on it can be highly nontrivial.

Second, the representation \(\rho\) is additional structure on the smooth
bundle.  Parallel transport around a loop \(\gamma\) gives
\(\rho(\gamma)\), and the class functions \(W_\gamma\) in
\eqref{eq:wilson-functions} vary even while \([E_\rho]\) does not.  Forgetting
the flat connection before applying \(K\)-theory is precisely what removes
this information.

Third, a family of fibrewise trivial bundles need not be trivial as a family.
This is already visible in rank one.

\subsection{The Poincar\'e family: fibres and families}

Let \(\mathcal J_X=\Jac(X)\), and choose a base point \(x_0\in X\).  The
symbol \(\mathcal J_X\) is used here to avoid confusing the Jacobian with the
complex structure \(J\).  There is a normalized Poincar\'e line bundle
\[
                   \cP\longrightarrow X\times\mathcal J_X
\]
whose restriction to \(X\times\{\chi\}\) is \(L_\chi\), and whose restriction
to \(\{x_0\}\times\mathcal J_X\) is trivial.  Its first Chern class lies in
the mixed
K\"unneth summand
\[
 c_1(\cP)\in H^1(X;\ZZ)\otimes H^1(\mathcal J_X;\ZZ)
          \subset H^2(X\times\mathcal J_X;\ZZ).
 \tag{3.5}\label{eq:poincare-c1-space}
\]
Under the natural identification
\(H^1(\mathcal J_X;\ZZ)\cong H_1(X;\ZZ)\), it is the canonical evaluation
tensor.  We now fix the conventions used below.  Choose a symplectic
integral basis
\[
 \alpha_1,\ldots,\alpha_g,\beta_1,\ldots,\beta_g
       \quad\text{of }H^1(X;\ZZ),
 \qquad
 \int_X\alpha_i\wedge\beta_j=\delta_{ij},
\]
and corresponding integral flux classes
\(u_1,\ldots,u_g,v_1,\ldots,v_g\in H^1(\mathcal J_X;\ZZ)\).
We normalize these classes so that
\[
 c_1(\cP)
   =\sum_{i=1}^g
       \bigl(\alpha_i\otimes u_i+\beta_i\otimes v_i\bigr),
 \qquad
 \theta:=\sum_{i=1}^g u_i\wedge v_i,
 \tag{3.6}\label{eq:poincare-c1}
\]
where \(\theta\) is the theta class of the principal polarization.  This
choice fixes the orientation of the flux coordinates and the sign in the
index formula.

\begin{proposition}[Fibrewise collapse versus family topology]
\label{prop:poincare-family}
The restriction of \([\cP]\) to \(K^0(X\times\{\chi\})\) equals
\([\cO_X]\) for every \(\chi\in\mathcal J_X\), but
\[
                  [\cP]\neq[\cO_{X\times\mathcal J_X}]
                  \quad\text{in }K^0(X\times\mathcal J_X).
                  \tag{3.7}\label{eq:poincare-nontrivial}
\]
Moreover, the families Dolbeault index
\[
       \Ind(\bar\partial_{\cP})
       :=(p_{\mathcal J})_![\cP]\in K^0(\mathcal J_X)
       \tag{3.8}\label{eq:family-index}
\]
has
\[
 \ch_1\Ind(\bar\partial_{\cP})=-\theta,
 \tag{3.8a}\label{eq:family-index-theta}
\]
where \(p_{\mathcal J}:X\times\mathcal J_X\to\mathcal J_X\) is the
projection.
\end{proposition}

\begin{proof}
Restriction to \(X\times\{\chi\}\) annihilates the
\(H^1(\mathcal J_X;\ZZ)\)-factor in \eqref{eq:poincare-c1-space}, so each
restricted
line bundle has zero first Chern class.  The mixed class
\eqref{eq:poincare-c1} is nonzero on \(X\times\mathcal J_X\), proving
\eqref{eq:poincare-nontrivial}.  The families index theorem gives
\[
 \ch\Ind(\bar\partial_{\cP})
   =(p_{\mathcal J})_*\!\left(
       \e^{c_1(\cP)}\operatorname{Td}(TX)
     \right).
 \tag{3.9}\label{eq:family-index-ch}
\]
Here \(\operatorname{Td}(TX)\) is the Todd class of the tangent bundle.
The degree-zero and degree-two components on \(\mathcal J_X\) are
\[
 \ch_0\Ind(\bar\partial_{\cP})=1-g,
 \qquad
 \ch_1\Ind(\bar\partial_{\cP})
     =\frac12(p_{\mathcal J})_*\!\left(c_1(\cP)^2\right).
 \tag{3.9a}\label{eq:family-index-low-components}
\]
The first identity is Riemann--Roch for a degree-zero line bundle.  In the
second, the two factors of the mixed class \(c_1(\cP)\) provide precisely the
two degrees along \(X\) that survive fibre integration.  With the convention
in \eqref{eq:poincare-c1}, graded commutativity gives
\[
 (p_{\mathcal J})_*\!\left(c_1(\cP)^2\right)
      =-2\sum_{i=1}^g u_i\wedge v_i=-2\theta.
\]
Thus \(\ch_1\Ind(\bar\partial_{\cP})=-\theta\), proving
\eqref{eq:family-index-theta}.  This is the usual index-theoretic
construction of the principal polarization
\cite{AtiyahSingerFamilies,BirkenhakeLange}.
\end{proof}

\begin{remark}
\Cref{prop:poincare-family} refines the preceding conclusion.  The
\(K\)-theory of a \emph{family} can retain global mixed topology that
fibrewise \(K^0(X)\) loses.  It still does not determine the flat connection
on each \(L_\chi\), the period matrix of \(X\), a twisted Hamiltonian
\(H_\chi\), or the functions \(\chi\mapsto E_n(\chi)\).  Families
\(K\)-theory is therefore richer than the map \(\kappa_1\), while remaining a
stable topological quotient of the geometric band datum.
\end{remark}

\subsection{Where the first strict loss occurs}

It is useful to locate the collapse precisely.  Let
\(\operatorname{Vect}^{\mathrm{top}}_{r,0}(S)\) denote the set of smooth
isomorphism classes of rank-\(r\), degree-zero complex vector bundles on
\(S\), and let \(\cM_X^{\mathrm{ps}}(r,0)\) denote the moduli space of
polystable holomorphic bundles of that rank and degree.  For fixed \(J\), the
relevant maps are
\[
\begin{tikzcd}[row sep=1.15em]
\Hom(\Gamma_g,\UU(r))
  \arrow[d, "\text{conjugacy quotient}"]\\
\cR_r(S)
  \arrow[d, "\mathrm{NS}_J"', "\cong"]\\
\cM_X^{\mathrm{ps}}(r,0)
  \arrow[d, "\text{forget }\bar\partial_E"]\\
\operatorname{Vect}^{\mathrm{top}}_{r,0}(S)
  \arrow[d, "\text{stable class}"]\\
K^0(S).
\end{tikzcd}
\tag{3.10}\label{eq:forgetful-flat-holo-k}
\]
The first arrow removes a choice of fibre framing.  The middle arrow is the
Narasimhan--Seshadri correspondence, not a loss of information: a polystable
degree-zero holomorphic bundle determines its compatible flat unitary
connection up to unitary gauge.  The strict geometric loss occurs when the
flat or holomorphic structure is forgotten.  Indeed,
\(\operatorname{Vect}^{\mathrm{top}}_{r,0}(S)\) has a single element, because
rank and first Chern class classify smooth complex bundles on a surface.

At the last arrow the situation is simpler.  On a surface and at fixed rank,
passage from a smooth bundle to its \(K\)-class introduces no further
ambiguity: the \(K\)-class already determines rank and degree.  Thus the
sectorwise collapse in \Cref{thm:k-collapse} occurs when the connection or
Dolbeault operator is forgotten, rather than through stabilization alone.  On
a higher-dimensional parameter space \(B\), by contrast, stabilization can
introduce additional losses.  In later sections we place the Hamiltonian,
spectral projector, Berry connection, and quantum metric to the left of
\eqref{eq:forgetful-flat-holo-k}; the expanded hierarchy in
\Cref{sec:completion} records exactly which forgetful step is responsible for
each missing observable.

\section{Closed walks, Wilson functions, and spectral nonfactorization}
\label{sec:walks}

\subsection{A finite group-labelled Hamiltonian}

We now give a model-independent reason that flat-bundle data affect the
spectrum.  Let \(\mathcal Q=(V,\cE)\) be a finite directed graph.  Every
geometric edge
occurs with both orientations \(e\) and \(\bar e\), where
\[
 s(\bar e)=t(e),\qquad t(\bar e)=s(e).
\]
Choose labels \(\gamma_e\in\Gamma_g\) and hopping matrices
\(T_e\in\operatorname{Mat}_m(\CC)\) such that
\[
              \gamma_{\bar e}=\gamma_e^{-1},
              \qquad T_{\bar e}=T_e^*.
              \tag{4.1}\label{eq:reverse-edge}
\]
The integer \(m\) counts internal orbitals at each quotient vertex.  We also
allow an on-site Hermitian matrix \(V_v\in\operatorname{Herm}(m)\) at every
vertex.

For a unitary representation \(\rho:\Gamma_g\to\UU(r)\), define the twisted
Hamiltonian on
\[
       \cH_\rho=\CC^V\otimes\CC^m\otimes\CC^r
\]
by
\begin{align}
H_\rho
  ={}&\sum_{v\in V}\ket v\bra v\otimes V_v\otimes I_r \notag\\
   &+\sum_{e\in\cE^+}
      \left(
       \ket{t(e)}\bra{s(e)}\otimes T_e\otimes\rho(\gamma_e)
       +
       \ket{s(e)}\bra{t(e)}\otimes T_e^*\otimes\rho(\gamma_e)^*
      \right),
      \tag{4.2}\label{eq:twisted-H}
\end{align}
where \(\cE^+\) contains one orientation of each geometric edge.  This is the
usual finite Bloch matrix, with scalar phases replaced by matrix-valued
holonomies.  It includes the abelian matrices of
\cite{MaciejkoRayan2021,BoettcherEtAl2022} and the higher-dimensional sectors
of \cite{MaciejkoRayan2022,LenggenhagerEtAl2023}.

If \(\rho'=u\rho u^{-1}\), then
\[
 H_{\rho'}=(I\otimes I\otimes u)H_\rho(I\otimes I\otimes u^{-1}).
 \tag{4.3}\label{eq:rho-conjugacy}
\]
The spectrum therefore descends from the representation space to
\(\cR_r(S)\).  A change of lifts of quotient vertices also conjugates
\eqref{eq:twisted-H} by a block-diagonal unitary.  The spectrum is independent
of this vertex gauge, as it must be.

It is convenient to regard an on-site term as a labelled loop of group label
\(1\).  A directed walk
\[
             w=(e_1,e_2,\ldots,e_n),
             \qquad t(e_j)=s(e_{j+1}),
\]
has ordered group element and orbital amplitude
\[
   \gamma(w)=\gamma_{e_n}\cdots\gamma_{e_2}\gamma_{e_1},
   \qquad
   T(w)=T_{e_n}\cdots T_{e_2}T_{e_1}.
   \tag{4.4}\label{eq:walk-data}
\]
The order convention is chosen to match operator composition.  For scalar
hoppings, \(T(w)\) is simply the product of hopping amplitudes along \(w\).

\begin{theorem}[Closed-walk formula]
\label{thm:walk-formula}
For every positive integer \(n\),
\[
  \Tr_{\cH_\rho}(H_\rho^n)
   =
  \sum_{\substack{w\ {\rm directed\ closed\ walk}\\ |w|=n}}
        \tr_{\CC^m}T(w)\,
        \tr_{\CC^r}\rho(\gamma(w)).
  \tag{4.5}\label{eq:walk-trace}
\]
The sum includes the on-site loops described above.  In particular,
\(\Tr(H_\rho^n)\) is a finite linear combination of Wilson functions on
\(\cR_r(S)\).
\end{theorem}

\begin{proof}
Expand \(H_\rho^n\) using \eqref{eq:twisted-H}.  A product of matrix units
\[
 \ket{t(e_n)}\bra{s(e_n)}
 \cdots
 \ket{t(e_1)}\bra{s(e_1)}
\]
is nonzero exactly when the edges concatenate.  Its trace over \(\CC^V\)
vanishes unless the resulting walk is closed, in which case it is one.
The trace over a tensor product factors.  The internal-orbital product is
\(T(w)\), and the representation product is
\[
 \rho(\gamma_{e_n})\cdots\rho(\gamma_{e_1})
 =\rho(\gamma(w)).
\]
Summing the surviving terms gives \eqref{eq:walk-trace}.
\end{proof}

\begin{remark}
Formula \eqref{eq:walk-trace} separates two kinds of geometry.  The conjugacy
class of \(\gamma(w)\) records how the walk closes through the side pairings
of the fundamental cell.  The coefficient \(\tr T(w)\) records orbital,
metric, and fabrication data along the same walk.  The spectrum mixes them;
the class \([E_\rho]\in K^0(S)\) records neither.
\end{remark}

\subsection{A nonfactorization theorem}

Write \(N=|V|mr\).  The unordered spectrum of \(H_\rho\), including
multiplicity, is determined by the first \(N\) power sums
\[
          M_n(\rho):=\Tr(H_\rho^n),\qquad 1\leq n\leq N,
\]
through the Newton identities.  This turns the closed-walk formula into a
sharp obstruction.

\begin{theorem}[Spectral nonfactorization]
\label{thm:spectral-nonfactorization}
Suppose that \(M_n:\cR_r(S)\to\CC\) is nonconstant for some
\(1\leq n\leq N\).  Then there is no map
\[
           F:K^0(S)\longrightarrow
             \{\text{unordered \(N\)-point spectra in \(\RR\)}\}
\]
such that
\[
                 \spec(H_\rho)=F([E_\rho])
                 \qquad\text{for every }\rho.
                 \tag{4.6}\label{eq:no-spectrum-factor}
\]
More generally, neither the empirical density of states nor any observable
that determines \(M_n\) can factor through \([E_\rho]\).
\end{theorem}

\begin{proof}
By \Cref{thm:k-collapse}, \([E_\rho]=r[\underline{\CC}]\) for every
\(\rho\).  If \eqref{eq:no-spectrum-factor} held, the spectrum would be
constant, and therefore every power sum of its eigenvalues would be constant.
This contradicts the hypothesis.  The same argument applies to any invariant
from which \(M_n\) can be recovered.
\end{proof}

The criterion is intentionally easy to check.  One does not need to solve the
eigenvalue problem.  It is enough to find a closed-walk coefficient in
\eqref{eq:walk-trace} that produces a nonconstant Wilson combination.

\begin{example}[One labelled loop]
\label{ex:one-loop}
Take one quotient vertex, one orbital, and an oriented loop labelled by
\(\gamma\in\Gamma_g\).  Let the hopping be \(t\neq0\) and the on-site energy
be \(v\in\RR\).  For a character \(\chi\),
\[
             H_\chi
              =v+t\chi(\gamma)+\overline t\,\chi(\gamma)^{-1}.
              \tag{4.7}\label{eq:one-loop-energy}
\]
If the image of \(\gamma\) in \(H_1(S;\ZZ)\) is nonzero, evaluation at
\(\gamma\) maps the character torus onto \(\UU(1)\).  Consequently
\[
 \{H_\chi:\chi\in\Jac(X)\}
      =[v-2|t|,\,v+2|t|].
 \tag{4.8}\label{eq:one-loop-range}
\]
Every \(L_\chi\) has the same \(K^0(S)\)-class, while even this one-level
spectrum varies through an interval.
\end{example}

\begin{example}[Abelian Fourier expansion]
\label{ex:abelian-fourier}
Fix the marking in \eqref{eq:surface-group}.  For
\([\gamma]\in H_1(S;\ZZ)\cong\ZZ^{2g}\), write
\[
 \chi_{\bm\theta}(\gamma)
    =\exp\!\bigl(\ii\langle[\gamma],\bm\theta\rangle\bigr).
\]
Grouping the closed walks in \eqref{eq:walk-trace} by homology class gives
\[
       M_n(\bm\theta)
          =\sum_{\bm q\in\ZZ^{2g}}A_{n,\bm q}
             \e^{\ii\langle\bm q,\bm\theta\rangle},
          \tag{4.9}\label{eq:moment-fourier}
\]
with only finitely many nonzero coefficients for fixed \(n\).  Thus the
ordinary Fourier modes of the higher-dimensional Brillouin torus are weighted
closed-walk homology classes.  \(M_n\) is constant exactly when all
coefficients with \(\bm q\neq0\) cancel.
\end{example}

\begin{criterion}[Closed-walk test]
\label{crit:closed-walk}
For an abelian hyperbolic Bloch Hamiltonian, a nonzero coefficient
\(A_{n,\bm q}\) with \(\bm q\neq0\) in \eqref{eq:moment-fourier} proves that
the spectrum is not determined by \(K^0(S)\).  For a nonabelian Hamiltonian,
it is enough that the Wilson combination on the right-hand side of
\eqref{eq:walk-trace} be nonconstant on the representation locus under
consideration.
\end{criterion}

\begin{proposition}[Positive hoppings force spectral variation]
\label{prop:positive-hopping-variation}
Consider a connected one-orbital quotient graph with real, strictly positive
hopping amplitudes on its geometric edges.  On-site terms may be arbitrary
real numbers.  Suppose that the graph contains a closed walk whose group
label has nonzero image in \(H_1(S;\ZZ)\).  Then the abelian spectrum is
nonconstant on \(B_{\mathrm{ab}}\), and hence does not factor through
\(K^0(S)\).
\end{proposition}

\begin{proof}
Among all closed walks with nonzero abelianized label, choose a minimal
length \(n_0\), and let \(\bm q\neq0\) be the homology class of one such walk.
No length-\(n_0\) walk contributing to \(A_{n_0,\bm q}\) can contain an
on-site step: deleting that step would leave a shorter closed walk with the
same homology class.  Every contribution to \(A_{n_0,\bm q}\) is therefore a
product of strictly positive edge hoppings.  At least one such contribution
exists, so
\[
                         A_{n_0,\bm q}>0.
                         \tag{4.9a}\label{eq:positive-fourier-coefficient}
\]
The Fourier polynomial \(M_{n_0}\) is nonconstant by
\Cref{crit:closed-walk}, and \Cref{thm:spectral-nonfactorization} applies.
\end{proof}

\Cref{prop:positive-hopping-variation} removes the nonconstancy hypothesis
from a broad class of scalar models.  Spectral variation is not a delicate
effect requiring complex hopping phases: a homologically nontrivial closed
trajectory and positive hopping are already sufficient.  Matrix-valued
orbitals can introduce trace cancellations, but such cancellations are
additional algebraic constraints rather than a consequence of \(K\)-theory.

\subsection{Band velocities and perturbations}

On a smooth parameter region \(B\subset\cR_r(S)\), let \(x^\mu\) be local
coordinates and suppose \(E_j(x)\) is a simple eigenvalue with normalized
eigenvector \(u_j(x)\).  The Feynman--Hellmann formula gives
\[
       \frac{\partial E_j}{\partial x^\mu}
          =\bra{u_j}\frac{\partial H}{\partial x^\mu}\ket{u_j}.
          \tag{4.10}\label{eq:feynman-hellmann}
\]
For abelian flux coordinates, differentiating a holonomy factor gives
\[
 \frac{\partial}{\partial\theta_\mu}\chi_{\bm\theta}(\gamma)
       =\ii[\gamma]_\mu\,\chi_{\bm\theta}(\gamma).
 \tag{4.11}\label{eq:holonomy-derivative}
\]
The velocity in a flux direction therefore measures weighted winding of the
hopping terms around the corresponding surface cycle.  It vanishes only when
symmetry or cancellation forces it to vanish.  The \(K^0(S)\)-class is locally
and globally constant and contains no derivative information of this kind.

If a control parameter \(s\) changes the metric-dependent hoppings rather
than the representation, then
\[
    \dot E_j(s)=\bra{u_j(s)}\dot H(s)\ket{u_j(s)}
    \tag{4.12}\label{eq:metric-fh}
\]
at a simple eigenvalue.  This gives the most direct experimental separation
between topology and geometry: vary capacitances or couplings while preserving
the quotient graph and its \(K\)-theoretic phase, and measure the resulting
resonance shifts.

\subsection{Continuum Hamiltonians and the equivariant heat kernel}

The same structure is present without a tight-binding approximation.  Let
\((X,h)=(\HH^2/\Gamma,h)\), let \(E_\rho\) carry its flat unitary connection
\(\nabla_\rho\), and consider
\[
       \cH_{h,\rho}
        =\nabla_\rho^*\nabla_\rho+V
        \tag{4.13}\label{eq:continuum-H}
\]
on \(L^2(X,E_\rho)\), with \(V\) a smooth \(\Gamma\)-periodic Hermitian
potential.  Let \(\widetilde K_t(\widetilde x,\widetilde y)\) be the heat
kernel of the lifted scalar or orbital operator on \(\HH^2\), and let
\(\mathfrak F\subset\HH^2\) be a fundamental domain.

\begin{proposition}[Twisted heat-trace formula]
\label{prop:heat-trace}
For \(t>0\), under the standard heat-kernel convergence hypotheses,
\[
 \Tr\!\left(\e^{-t\cH_{h,\rho}}\right)
 =
 \sum_{\gamma\in\Gamma}
   \tr\rho(\gamma)
   \int_{\mathfrak F}
      \tr_{\mathrm{orb}}
      \widetilde K_t(\widetilde x,\gamma\widetilde x)
      \,\dd\mathrm{vol}_h(\widetilde x).
 \tag{4.14}\label{eq:twisted-heat}
\]
The identity contribution is local.  The nonidentity contributions couple
Wilson functions to the geometry of the deck transformations.
\end{proposition}

\begin{proof}
The heat kernel on \(E_\rho\) is obtained by the method of images:
\[
 K_{\rho,t}(x,y)
   =\sum_{\gamma\in\Gamma}
      \widetilde K_t(\widetilde x,\gamma\widetilde y)\rho(\gamma),
 \]
with an inverse on \(\gamma\) if the opposite equivariance convention is
chosen.  Restricting \(x=y\), tracing over the fibre, and integrating over a
fundamental domain gives \eqref{eq:twisted-heat}.  Gaussian heat-kernel bounds
and proper discontinuity justify the exchange of sum and integral.
\end{proof}

For \(V=0\) and constant curvature, regrouping the nonidentity terms by
conjugacy class yields the twisted Selberg trace formula.  Its coefficients
depend on the lengths \(\ell_h(\gamma)\) of closed geodesics and on
\(\tr\rho(\gamma)\) \cite{Selberg,Buser}.  Attar and Boettcher give a direct
application of the Selberg trace formula to hyperbolic band theory
\cite{AttarBoettcher2022}.  This is the continuum counterpart of
\eqref{eq:walk-trace}:
\[
\begin{array}{c|c|c}
 &\text{discrete model}&\text{continuum model}\\
\hline
\text{closed object}&\text{labelled walk }w&\text{closed geodesic }[\gamma]\\
\text{geometry}&\tr T(w)&\ell_h(\gamma)\text{ and orbital transport}\\
\text{Bloch sector}&\tr\rho(\gamma(w))&\tr\rho(\gamma).
\end{array}
\tag{4.15}\label{eq:trace-dictionary}
\]

For the scalar Laplacian this coupling can be written explicitly.  Let
\(\mathscr P_{\mathrm{or}}\) be a set of primitive oriented conjugacy classes
that is closed under inversion.  With \(\gamma_p\) representing
\(p\in\mathscr P_{\mathrm{or}}\), the heat-kernel specialization of the
twisted Selberg trace formula is
\begin{align}
\Tr\!\left(\e^{-t\Delta_\rho}\right)
={}&
 \frac{r\,\operatorname{Area}(X)}{4\pi}
 \int_{-\infty}^{\infty}
     \e^{-t(\xi^2+1/4)}\,
     \xi\tanh(\pi\xi)\,\dd\xi \notag\\
&+
 \frac{\e^{-t/4}}{\sqrt{4\pi t}}
 \sum_{p\in\mathscr P_{\mathrm{or}}}
 \sum_{k=1}^{\infty}
 \frac{\ell(p)\,\tr\rho(\gamma_p^k)}
      {2\sinh(k\ell(p)/2)}
 \exp\!\left(-\frac{k^2\ell(p)^2}{4t}\right).
 \tag{4.15a}\label{eq:twisted-selberg-heat}
\end{align}
Pairing \(p\) with its reverse replaces the two Wilson factors by twice their
real part.
The identity term knows only the rank and area.  Every nonidentity term
multiplies a length-dependent weight by a Wilson trace.  Thus the full
formula, not only its schematic structure, separates the metric and flat
connection data that collapse in \(K^0(S)\)
\cite{Selberg,Buser}.

The local small-\(t\) heat coefficients do not see a flat holonomy, because
its curvature vanishes.  Holonomy enters through global closed trajectories.
This is another reason a local characteristic class cannot reconstruct the
full spectrum: the distinction between two flat connections is intrinsically
nonlocal.

\begin{corollary}
\label{cor:heat-nonfactor}
If the heat trace in \eqref{eq:twisted-heat} is nonconstant as a function of
\([\rho]\), then the spectrum of \(\cH_{h,\rho}\) does not factor through
\([E_\rho]\in K^0(S)\).
\end{corollary}

\begin{proof}
The heat trace is the Laplace transform of the discrete spectral measure on
the compact surface.  Equal spectra give equal heat traces for every \(t>0\).
Apply \Cref{thm:k-collapse}.
\end{proof}

\subsection{Abelianization and information lost at rank one}

The abelian Brillouin torus does not merely simplify the nonabelian theory:
it applies a definite quotient to the real-space labels.  Let
\[
 \operatorname{ab}:\Gamma_g\longrightarrow
 \Lambda_g:=H_1(S;\ZZ)\cong\ZZ^{2g}
\tag{4.16}\label{eq:abelianization}
\]
be the abelianization map.  Every character factors uniquely through it:
\[
 \chi=\widehat\chi\circ\operatorname{ab},
 \qquad \widehat\chi\in\Hom(\Lambda_g,\UU(1)).
\]
Hence an abelian Bloch matrix depends on an edge label \(\gamma_e\) only
through the displacement
\[
                       q_e=\operatorname{ab}(\gamma_e)\in\Lambda_g.
                       \tag{4.17}\label{eq:abelian-edge-displacement}
\]
Words in the commutator subgroup are invisible to every rank-one sector.
This is, in fact, a loss of geometric organization before any \(K\)-theory is applied.

\begin{theorem}[Abelian blindness and nonabelian recovery]
\label{thm:abelian-blindness}
Consider two group-labelled Hamiltonians with the same quotient graph,
on-site matrices, and hopping matrices, but with edge labels
\(\gamma_e\) and \(\gamma'_e\).

\begin{enumerate}[label=\textup{(\alph*)},leftmargin=2.5em]
\item If
  \(\operatorname{ab}(\gamma_e)=\operatorname{ab}(\gamma'_e)\) for every
  edge, then
  \[
                         H_\chi=H'_\chi
  \]
  for every character \(\chi:\Gamma_g\to\UU(1)\).
\item For every \(g\geq2\), there are labelings satisfying the hypothesis in
  \textup{(a)} and a representation
  \(\rho:\Gamma_g\to\SU(2)\) for which
  \(\spec(H_\rho)\neq\spec(H'_\rho)\).
\end{enumerate}
\end{theorem}

\begin{proof}
For part \textup{(a)}, factor \(\chi\) through
\eqref{eq:abelianization}.  Equality of the abelianized labels gives
\(\chi(\gamma_e)=\chi(\gamma'_e)\) term by term in
\eqref{eq:twisted-H}.

For part \textup{(b)}, take a one-vertex, one-orbital graph with one loop.
Label it first by \(1\) and then by the commutator
\(c=[a_1,b_1]\).  These labels have the same image under
\(\operatorname{ab}\).  Choose noncommuting \(A,B\in\SU(2)\) and define
\[
 \rho(a_1)=A,\quad \rho(b_1)=B,\quad
 \rho(a_2)=B,\quad \rho(b_2)=A,
 \tag{4.18}\label{eq:commutator-representation}
\]
with all remaining generators sent to the identity.  Because
\([B,A]=[A,B]^{-1}\), the surface relation is satisfied.  Moreover,
\(C:=\rho(c)=[A,B]\neq I\).

With a nonzero real hopping \(t\) and zero on-site energy, the two rank-two
Hamiltonians are
\[
             H_\rho=2tI_2,\qquad
             H'_\rho=t(C+C^*).
             \tag{4.19}\label{eq:commutator-loop-H}
\]
If the eigenvalues of \(C\) are \(\e^{\pm\ii\varphi}\), then
\(H'_\rho=2t\cos\varphi\,I_2\).  Since \(C\neq I\), one has
\(\cos\varphi\neq1\), so the spectra differ.
\end{proof}

The theorem separates two successive losses of information.  Passing to
rank-one momentum replaces group words by their homology classes.  Passing
next to fibrewise \(K^0(S)\) identifies every degree-zero character line bundle
with the same stable class.  Thus, if one retains only topological invariants
of rank-one bands, both losses have already occurred.  Nonabelian Bloch
sectors can recover information lost in the first passage, but not in the
second: their Wilson functions can respond to
commutator words even though all \([E_\rho]\) remain equal.

This does not mean that a single fixed rank reconstructs every microscopic
edge label.  Trace identities and vertex gauge can identify different
presentations.  The physically relevant statement is more modest and more
useful: nonabelian sectors supply observables on conjugacy data that cannot be
functions of abelian momentum.

\subsection{The abelian cover and position--momentum duality}

There is a corresponding real-space statement.  The rank-one theory naturally
uses the maximal abelian cover
\[
           S_{\mathrm{ab}}
             =\widetilde S/[\Gamma_g,\Gamma_g]
             \longrightarrow S,
           \qquad \operatorname{Deck}(S_{\mathrm{ab}}/S)=\Lambda_g.
           \tag{4.20}\label{eq:maximal-abelian-cover}
\]
After a choice of quotient cell, a finite-range tight-binding Hilbert space
on this cover has the form
\(\ell^2(\Lambda_g)\otimes\CC^d\).  Its Fourier transform is
\[
 (\mathscr F_{\mathrm{ab}}\psi)(\bm\theta)
   =\frac{1}{(2\pi)^g}
      \sum_{\bm q\in\Lambda_g}
       \e^{-\ii\langle\bm q,\bm\theta\rangle}\psi(\bm q),
 \qquad \bm\theta\in(\RR/2\pi\ZZ)^{2g}.
 \tag{4.21}\label{eq:abelian-bloch-transform}
\]
This is the rank-one part of the hyperbolic Bloch transform developed in
\cite{MaciejkoRayan2022,NagyRayan2024}.  It is an ordinary Fourier transform
on the deck group of the abelian cover, not on the nonabelian group
\(\Gamma_g\).

For \(1\leq\mu\leq2g\), let \(N_\mu\) be multiplication by \(q_\mu\) on its
natural dense domain:
\[
                       (N_\mu\psi)(\bm q)=q_\mu\psi(\bm q).
                       \tag{4.22}\label{eq:abelian-position}
\]
These are valid position operators for homological displacement.  They are
not a complete coordinate system on the universal cover: every displacement
in \([\Gamma_g,\Gamma_g]\) has all \(2g\) coordinates equal to zero.

\begin{proposition}[Position and flux derivative]
\label{prop:position-flux}
Under \(\mathscr F_{\mathrm{ab}}\),
\[
        \mathscr F_{\mathrm{ab}}N_\mu\mathscr F_{\mathrm{ab}}^{-1}
             =\ii\frac{\partial}{\partial\theta_\mu}.
        \tag{4.23}\label{eq:position-derivative}
\]
If a bounded, finite-range, \(\Lambda_g\)-equivariant Hamiltonian \(H\)
transforms into the matrix-valued multiplier \(H(\bm\theta)\), then on smooth
Bloch sections
\[
    \mathscr F_{\mathrm{ab}}\,\ii[H,N_\mu]\,
       \mathscr F_{\mathrm{ab}}^{-1}
             =\frac{\partial H(\bm\theta)}{\partial\theta_\mu}.
    \tag{4.24}\label{eq:current-derivative}
\]
\end{proposition}

\begin{proof}
Termwise differentiation of \eqref{eq:abelian-bloch-transform} gives
\[
 \ii\partial_{\theta_\mu}
  \e^{-\ii\langle\bm q,\bm\theta\rangle}
   =q_\mu\e^{-\ii\langle\bm q,\bm\theta\rangle},
\]
which proves \eqref{eq:position-derivative} on finitely supported vectors and
then on the natural domains.  In Bloch representation, \(H\) is
multiplication by \(H(\bm\theta)\).  Therefore
\[
 [H(\bm\theta),\ii\partial_{\theta_\mu}]f
       =-\ii(\partial_{\theta_\mu}H)f.
\]
Multiplication by \(\ii\) yields \eqref{eq:current-derivative}.
\end{proof}

The proposition gives a precise meaning to velocity in a hyperbolic
rank-one flux direction.  It is dual to displacement around a chosen
homology cycle.  It also states the limitation precisely.  No collection of
derivatives on the Jacobian detects a loop whose group label lies in the
commutator subgroup, because that loop acts trivially on the maximal abelian
cover.  To probe such directions one needs matrix-valued holonomies, Wilson
observables, or the full nonabelian transform.

\subsection{First- and second-order spectral response}

Flux dependence contains more information than a list of energies at one
point.  Put
\[
       J_\mu(\bm\theta)=\partial_\mu H(\bm\theta),\qquad
       D_{\mu\nu}(\bm\theta)=\partial_\mu\partial_\nu H(\bm\theta).
       \tag{4.25}\label{eq:current-diamagnetic}
\]
By \Cref{prop:position-flux}, \(J_\mu\) is the Bloch representation of a
current operator defined by homological position.  If \(E_j\) is simple, the
Feynman--Hellmann formula \eqref{eq:feynman-hellmann} gives its first
response.  A second differentiation gives
\begin{align}
 \partial_\mu\partial_\nu E_j
   ={}&\langle u_j,D_{\mu\nu}u_j\rangle \notag\\
    &+2\Re\sum_{k\neq j}
       \frac{
        \langle u_j,J_\mu u_k\rangle
        \langle u_k,J_\nu u_j\rangle}
       {E_j-E_k}.
       \tag{4.26}\label{eq:second-order-response}
\end{align}
The first line is the direct or diamagnetic contribution.  The second is an
interband contribution controlled by current matrix elements and energy
denominators.  Thus band velocity, effective-mass tensors, and optical
oscillator strengths retain the detailed hopping labels and eigenvectors.
None can be reconstructed from the constant class
\([L_\chi]=[\cO_X]\).

At a band crossing, individual eigenvalue derivatives need not exist.  The
appropriate object is then the Riesz projection onto an isolated cluster,
\[
            P(\bm\theta)
              =\frac{1}{2\pi\ii}\oint_\mathcal C
                 (z-H(\bm\theta))^{-1}\,\dd z,
                 \tag{4.27}\label{eq:riesz-cluster}
\]
where \(\mathcal C\) encloses the cluster and no other spectrum.  As long as
the cluster gap remains open, \(P\), the cluster energy
\(\Tr(PH)\), and their derivatives are smooth even when internal
degeneracies occur.  This distinction is important experimentally: a
singular choice of eigenvector is not itself a singularity of the occupied
subspace.

The response formulas also distinguish topology from geometry without
changing a phase.  Multiplying the hopping amplitudes by a smooth positive
control \(s\mapsto t_e(s)\), or changing them through distances and angles,
changes \(J_\mu\), \(D_{\mu\nu}\), and the right-hand side of
\eqref{eq:second-order-response}.  If the Fermi gap stays open, the
occupied \(K\)-class stays fixed throughout the deformation.

\subsection{What momentum spectroscopy can reconstruct}

The Fourier expansion \eqref{eq:moment-fourier} has a useful inverse.  With
\(\dd^{2g}\bm\theta\) the Haar measure in the marked flux coordinates,
\[
 A_{n,\bm q}
   =\frac{1}{(2\pi)^{2g}}
     \int_{[0,2\pi]^{2g}}
       M_n(\bm\theta)
       \e^{-\ii\langle\bm q,\bm\theta\rangle}
       \,\dd^{2g}\bm\theta.
       \tag{4.28}\label{eq:moment-fourier-inverse}
\]
Measuring the eigenvalues across the Jacobian therefore determines the
aggregate weight of length-\(n\) closed walks in every homology class.  All
walks with the same \(\bm q\), however, contribute to the same coefficient.

\begin{proposition}[Homological tomography and its limit]
\label{prop:homological-tomography}
For a finite abelian Bloch Hamiltonian, knowledge of the spectral moments
\(M_n(\bm\theta)\) for all \(\bm\theta\) and
\(1\leq n\leq N\) determines the unordered spectrum at every momentum and all
coefficients \(A_{n,\bm q}\).  It does not, in general, determine how
\(A_{n,\bm q}\) decomposes among closed group words with the same
abelianization.  In particular, rank-one momentum spectroscopy cannot
distinguish insertion of commutator labels as in
\Cref{thm:abelian-blindness}.
\end{proposition}

\begin{proof}
Newton identities reconstruct the spectrum from the first \(N\) moments, and
\eqref{eq:moment-fourier-inverse} reconstructs the coefficients.  Every
character takes the same value on two words with the same abelianization, so
their contributions enter only through their sum.  The example in
\Cref{thm:abelian-blindness} shows that this failure of word-level
reconstruction occurs even for a one-loop quotient graph.
\end{proof}

The hierarchy of probes can be summarized as follows.
\begin{center}
\small
\begin{tabular}{@{}
  >{\raggedright\arraybackslash}p{0.19\textwidth}
  >{\raggedright\arraybackslash}p{0.31\textwidth}
  >{\raggedright\arraybackslash}p{0.40\textwidth}@{}}
\toprule
\textbf{probe} & \textbf{resolved datum} & \textbf{unresolved datum}\\
\midrule
rank-one flux scan
  & closed-walk homology and abelian dispersion
  & commutator words and nonabelian holonomy\\
rank-\(r\) sector
  & rank-\(r\) Wilson combinations
  & metric and orbital coefficients not separately varied\\
metric or coupling scan
  & variation of closed-walk amplitudes
  & holonomy unless momentum is also scanned\\
\(K\)-class
  & stable rank, Chern data, and index pairings
  & all of the preceding spectral resolution\\
\bottomrule
\end{tabular}
\end{center}

The first two rows distinguish abelian and nonabelian geometry.  The third
suggests a separation protocol.  Let \(s\) be a metric or fabrication
parameter and write
\[
 M_n(\bm\theta;s)
   =\sum_{\bm q}A_{n,\bm q}(s)
      \e^{\ii\langle\bm q,\bm\theta\rangle}.
 \tag{4.29}\label{eq:mixed-response}
\]
A flux scan at fixed \(s\) recovers the homological sectors, while varying
\(s\) recovers \(\partial_s A_{n,\bm q}\).  The mixed derivative
\[
 \partial_s\partial_{\theta_\mu}M_n
   =\sum_{\bm q}\ii q_\mu\,
      \partial_sA_{n,\bm q}(s)
      \e^{\ii\langle\bm q,\bm\theta\rangle}
 \tag{4.30}\label{eq:mixed-response-derivative}
\]
measures how a geometric deformation reweights paths winding around a chosen
surface cycle.  This is a response coefficient with an immediate closed-walk
interpretation and no fibrewise \(K\)-theoretic reconstruction.

In the continuum, the analogous operation differentiates the geometric
coefficient in \eqref{eq:twisted-heat} while leaving the Wilson factor
explicit.  Suppose that a family of metrics \(h_s\) is smoothly trivialized
on the fixed marked surface and that uniform heat-kernel bounds justify
differentiation under the sum and integral.  Then
\[
 \partial_s\Tr\!\left(\e^{-t\cH_{h_s,\rho}}\right)
   =\sum_{\gamma\in\Gamma}
      \tr\rho(\gamma)\,
      \partial_s
      \int_{\mathfrak F_s}
       \tr_{\mathrm{orb}}
       \widetilde K_{t,s}(\widetilde x,\gamma\widetilde x)
       \,\dd\mathrm{vol}_{h_s}.
 \tag{4.31}\label{eq:mixed-heat-response}
\]
For the hyperbolic Laplacian, the nonidentity terms can be reorganized by
closed geodesics, so their variation includes the derivatives of the length
functions \(\ell_{h_s}(\gamma)\).  Discrete and continuum response therefore
tell the same story: holonomy selects a closed trajectory, while geometry
sets its weight.

These reconstruction statements concern the real torus underlying the
Jacobian.  The period matrix equips that torus with a polarization and a
complex structure, but ordinary Fourier inversion does not use them.  Complex
geometry becomes dynamically restrictive only when the Hamiltonian or its
projector obeys relations adapted to that complex structure.  This distinction
between kinematical and dynamical holomorphy is the subject of
\Cref{sec:metric-complex,sec:quantum-geometry}.

\section{Metric and complex structure beyond topology}
\label{sec:metric-complex}

\subsection{Teichm\"uller variation at fixed
  \texorpdfstring{\(K\)}{K}-theory}

The topology of \(S\) fixes \(\Gamma_g\), \(H^*(S;\ZZ)\), and
\eqref{eq:k-surface}.  It does not fix a point of Teichm\"uller space.
For \(g\geq2\),
\[
              \dim_\RR\cT_g=6g-6.
              \tag{5.1}\label{eq:teich-dim}
\]
The geodesic length spectrum, Laplace spectrum, Hodge star, period matrix, and
polarized Jacobian vary over this space.

The following standard spectral-geometric fact is enough to disprove any
factorization of the continuum spectrum through the topology of the surface.

\begin{theorem}[Metric nonfactorization]
\label{thm:metric-nonfactorization}
For every \(g\geq2\), there is a family of curvature-\(-1\) metrics \(h_s\) on
the fixed marked surface \(S\) such that
\[
                   \lambda_1(S,h_s)\longrightarrow0,
                   \tag{5.2}\label{eq:lambda-pinch}
\]
where \(\lambda_1\) is the first positive scalar Laplace eigenvalue.  In
particular, the map \(h\mapsto\spec(\Delta_h)\) is not constant on
\(\cT_g\), whereas \(K^*(S)\) and the class of the trivial line bundle remain
constant.
\end{theorem}

\begin{proof}
Choose a separating simple closed curve \(c\subset S\), and take a family of
hyperbolic metrics for which its geodesic representative has length
\(\ell_s\to0\).  The collar lemma supplies an embedded cylinder around \(c\)
whose conformal modulus tends to infinity.  Define a function \(f_s\) that is
approximately constant with opposite signs on the two components of
\(S\setminus c\), and interpolate harmonically across the collar.  Subtract
its mean.  The \(L^2\)-norm of \(f_s\) stays bounded below because both
components have positive limiting area, while the Dirichlet energy of the
interpolation tends to zero as the collar modulus tends to infinity.  The
Rayleigh quotient therefore tends to zero.  The min--max principle gives
\eqref{eq:lambda-pinch}.  See \cite{Buser} for the collar estimates and their
spectral consequences.
\end{proof}

\begin{remark}
The theorem does not assert that a spectrum determines a hyperbolic metric.
Isospectral nonisometric surfaces exist.  The required conclusion is weaker:
the metric-to-spectrum map is not constant, so topology alone cannot determine
it.
\end{remark}

\begin{proposition}[Infinitesimal metric response]
\label{prop:metric-variation}
Let \(h_s\) be a smooth family of Riemannian metrics with
\(h_0=h\) and \(\dot h_0=k\).  Suppose
\(\Delta_{h_s}u_s=\lambda_su_s\) is a smooth real-valued branch through a simple
eigenvalue, normalized by
\(\int_S|u_s|^2\,\dd\mathrm{vol}_{h_s}=1\).  Then
\begin{align}
\dot\lambda_0
={}&-\int_S
       \langle k,\dd u_0\otimes\dd u_0\rangle_h\,
       \dd\mathrm{vol}_h \notag\\
&+\frac12\int_S
       \operatorname{tr}_h(k)
       \left(
         |\dd u_0|_h^2-\lambda_0|u_0|^2
       \right)
       \dd\mathrm{vol}_h .
\tag{5.3}\label{eq:metric-variation}
\end{align}
For a transverse-traceless representative of a Teichm\"uller tangent vector,
the second line vanishes.
\end{proposition}

\begin{proof}
Differentiate the Rayleigh identity
\[
 \lambda_s=\int_S|\dd u_s|_{h_s}^2\,
                   \dd\mathrm{vol}_{h_s}.
\]
The standard variations
\(\frac{\dd}{\dd s}h_s^{-1}|_0=-h^{-1}kh^{-1}\) and
\(\frac{\dd}{\dd s}\dd\mathrm{vol}_{h_s}|_0
 =\frac12\operatorname{tr}_h(k)\dd\mathrm{vol}_h\)
give the displayed metric terms.  Integration by parts converts the terms
containing \(\dot u_0\) to
\(2\lambda_0\operatorname{Re}\int_S\overline{u_0}\dot u_0\).
Differentiating the normalization cancels this expression and supplies the
\(-\frac12\lambda_0\operatorname{tr}_h(k)|u_0|^2\) term.
\end{proof}

\Cref{prop:metric-variation} identifies the observable paired with a metric
deformation: it is the stress tensor of the eigenstate.  The formula has a
twisted analogue with \(\dd u\) replaced by \(\nabla_\rho u\), together with
the variation of any connection or potential that is allowed to depend on
the metric.  It makes the nonfactorization in
\Cref{thm:metric-nonfactorization} infinitesimal and experimentally
accessible.

For a twisted Laplacian, \eqref{eq:twisted-heat} makes the two sources of
variation simultaneous.  One may change the length weight by moving in
\(\cT_g\), the Wilson factor by moving in \(\cR_r(S)\), or both.  Neither
motion changes the class in \eqref{eq:k-collapse}.

\subsection{The Jacobian as a real torus and as a polarized complex torus}

A marking gives a fixed real vector space and lattice
\[
       V_S:=H^1(S;\RR),\qquad
       \Lambda_S:=2\pi H^1(S;\ZZ)\subset V_S.
\]
We identify a cohomology class with its \(h_J\)-harmonic representative.  If
\(a\in V_S\), then
\[
       \chi_a(\gamma)
          =\exp\!\left(\ii\int_\gamma a\right),
          \qquad \gamma\in\Gamma_g,
\]
depends only on the class of \(a\) modulo \(\Lambda_S\).  This gives the
marked identification
\[
                  B_{\mathrm{ab}}=V_S/\Lambda_S.
\]
The factor \(2\pi\) reflects the convention that holonomy angles have period
\(2\pi\); dividing \(a\) by \(2\pi\) recovers the standard integral
normalization used in \eqref{eq:period-jac}.

The Hodge star of \(h_J\) defines an endomorphism
\[
             \mathbb J_J[\eta]=[*_{h_J}\eta]
             \quad\text{on }H^1(S;\RR),
             \qquad \mathbb J_J^2=-I.
             \tag{5.4}\label{eq:hodge-complex}
\]
It is compatible with the intersection form.  We use the unnormalized Hodge
pairing
\begin{align}
 \omega_S(\alpha,\beta)
    &:=\int_S\alpha\wedge\beta,\notag\\
 G_J(\alpha,\beta)
    &:=\omega_S(\alpha,\mathbb J_J\beta)
      =\int_S\alpha\wedge *_{h_J}\beta .
 \tag{5.4a}\label{eq:hodge-polarization-data}
\end{align}
The form
\[
                E_S:=\frac{1}{(2\pi)^2}\omega_S
\]
is integral and unimodular on \(\Lambda_S\), while \(G_J\) is positive
definite.  Moreover,
\[
 \omega_S(\mathbb J_J\alpha,\mathbb J_J\beta)
       =\omega_S(\alpha,\beta).
\]
The triple \((\Lambda_S,E_S,\mathbb J_J)\) therefore makes
\(B_{\mathrm{ab}}\) a principally polarized complex torus.  Topology and the
marking fix \((V_S,\Lambda_S,E_S)\); the point \(J\in\cT_g\) supplies
the varying objects \(\mathbb J_J\) and \(G_J\).

To express the same data in holomorphic coordinates, choose the symplectic
homology basis \(\{[a_j],[b_j]\}_{j=1}^g\) and holomorphic one-forms
\(\{\zeta_k\}_{k=1}^g\) normalized by
\[
       \int_{a_j}\zeta_k=\delta_{jk},\qquad
       (\Omega_J)_{jk}:=\int_{b_j}\zeta_k .
       \tag{5.4b}\label{eq:period-matrix-definition}
\]
The Riemann bilinear relations give
\(\Omega_J^{\mathsf T}=\Omega_J\) and
\(\operatorname{Im}\Omega_J>0\), and the resulting complex torus is
\[
             \Jac(X_J)
                =\CC^g/(\ZZ^g+\Omega_J\ZZ^g).
\]
Thus \(\mathbb J_J\), \(G_J\), and \(\Omega_J\) are three descriptions of the
same complex-geometric refinement of the marked real flux torus.  The
distinction is physically important: a topological Fourier model may use
only \(V_S/\Lambda_S\), whereas the continuum bottom band below recovers the
metric \(G_J\) and hence the full polarized complex structure.

\subsection{The bottom band recovers the Hodge metric}

The complex geometry of the Jacobian is already visible in the dispersion of
the simplest continuum Hamiltonian.  Represent a character near the identity
by a harmonic real one-form
\[
                a(\bm\theta)=\sum_{\mu=1}^{2g}
                    \theta^\mu\eta_\mu ,
                \tag{5.5}\label{eq:harmonic-flat-potential}
\]
where \(\{\eta_\mu\}_{\mu=1}^{2g}\) are the harmonic representatives of a
chosen integral cohomology basis and the periods of \(a\) give the logarithms
of the holonomies.  On the smoothly trivial line bundle, define
\[
              \Delta_{a}
                =(\dd+\ii a)^*(\dd+\ii a).
                \tag{5.6}\label{eq:magnetic-flat-laplacian}
\]
A local gauge transformation by \(\e^{\ii f}\) replaces \(a\) by
\(a+\dd f\).  Globally, a \(\UU(1)\)-valued gauge transformation adds a
closed form with periods in \(2\pi\ZZ\).  The spectrum therefore depends only
on the corresponding point of
\(H^1(S;\RR)/2\pi H^1(S;\ZZ)\).

The local calculation that follows belongs to the classical spectral theory
of abelian coverings.  For harmonic one-forms \(\omega,\nu\), in the
character coordinate
\(\chi_\omega(\gamma)=\exp(2\pi\ii\int_\gamma\omega)\), Kotani and Sunada
obtain the quadratic form
\[
  \operatorname{Hess}_0\lambda_0(\omega,\nu)
    =\frac{8\pi^2}{\operatorname{Area}(X)}
       \int_X\omega\wedge *\nu
\]
as the Gaussian term governing their long-time heat-kernel expansion
\cite{KotaniSunada2000}.  Setting \(a=2\pi\omega\) gives precisely the
normalization below.  We include the proof both to fix conventions and to
make its band-theoretic role transparent.

\begin{theorem}[Bottom-band Hessian; cf.\ Kotani--Sunada]
\label{thm:hodge-hessian}
Let \(\lambda_0(\bm\theta)\) be the eigenvalue of
\(\Delta_{a(\bm\theta)}\) issuing from the simple zero eigenvalue of
\(\Delta_0\).  In a neighbourhood of \(\bm\theta=0\) it is the lowest
eigenvalue and is real analytic.  Its first derivative vanishes, and
\[
 \left.
 \frac{\partial^2\lambda_0}
      {\partial\theta^\mu\partial\theta^\nu}
 \right|_{\bm\theta=0}
   =
 \frac{2}{\operatorname{Area}(X)}
 \int_X\eta_\mu\wedge *\eta_\nu .
 \tag{5.7}\label{eq:hodge-hessian}
\]
Consequently, the marked Hessian of the bottom band determines the Hodge
inner product on \(H^1(S;\RR)\).
\end{theorem}

\begin{proof}
Let \(u_0=\operatorname{Area}(X)^{-1/2}\) be the normalized constant
eigenfunction.  For a real one-form \(a\),
\[
 \Delta_a f
   =\Delta f+\ii(\dd^*a)f
      -2\ii\langle a,\dd f\rangle+|a|^2f .
 \tag{5.8}\label{eq:magnetic-laplacian-expanded}
\]
Each \(\eta_\mu\) is harmonic, so \(\dd^*\eta_\mu=0\).  It follows that
\[
   \left(\partial_{\theta^\mu}\Delta_a\right)_0u_0=0,
 \qquad
   \left(
     \partial_{\theta^\mu}\partial_{\theta^\nu}\Delta_a
   \right)_0u_0
       =2\langle\eta_\mu,\eta_\nu\rangle u_0 .
 \tag{5.9}\label{eq:magnetic-operator-derivatives}
\]
Analytic perturbation theory applies because zero is a simple isolated
eigenvalue of \(\Delta_0\).  The first-order formula gives
\(\partial_\mu\lambda_0(0)=0\).  In the second-order formula, the usual
sum over excited states vanishes because the first operator derivative
annihilates \(u_0\).  Taking the expectation of the second derivative in
\eqref{eq:magnetic-operator-derivatives} gives
\eqref{eq:hodge-hessian}.
\end{proof}

In the notation of \eqref{eq:hodge-polarization-data}, the unnormalized
intersection form and Hodge metric are
\[
 \omega_S(\alpha,\beta)=\int_X\alpha\wedge\beta,
 \qquad
 G_J(\alpha,\beta)=\int_X\alpha\wedge *\beta .
 \tag{5.10}\label{eq:intersection-hodge-pair}
\]
They satisfy
\[
                    G_J(\alpha,\beta)
                       =\omega_S(\alpha,\mathbb J_J\beta).
                       \tag{5.11}\label{eq:hodge-from-symplectic}
\]
The integral lattice and \(\omega_S\) are fixed by the marked topology.
\Cref{thm:hodge-hessian} supplies \(G_J\), and
\eqref{eq:hodge-from-symplectic} therefore reconstructs
\(\mathbb J_J\).

\begin{corollary}[Homologically marked Hodge--Torelli consequence]
\label{cor:torelli-from-band}
For a marked closed hyperbolic surface, the Hessian in
\eqref{eq:hodge-hessian} determines the principally polarized Jacobian with
the homology marking induced by the surface marking.  By the classical
Torelli theorem it determines the underlying Riemann surface \(X_J\) up to
conformal isomorphism, and uniformization determines its curvature-\(-1\)
metric up to isometry.  The Hessian does not determine the full
Teichm\"uller marking.
\end{corollary}

\begin{proof}
Gauss--Bonnet fixes
\(\operatorname{Area}(X)=4\pi(g-1)\), so the area factor in
\eqref{eq:hodge-hessian} is already known.  The preceding discussion
recovers the Hodge star on the marked integral cohomology lattice, hence the
period point and principal polarization.  Torelli recovers the Riemann
surface from its principally polarized Jacobian
\cite{FarkasKra,BirkenhakeLange}; uniformization recovers its hyperbolic
metric.  A mapping class in the Torelli group acts trivially on
\(H_1(S;\ZZ)\), and therefore leaves the homology-marked period data
unchanged, while it can change the Teichm\"uller marking
\cite{FarbMargalit}.
\end{proof}

Colbois, Provenzano, and Savo obtain the corresponding global reconstruction
statement from a countable magnetic ground-state spectrum: its weak-field
asymptotics recover the Hodge Gram matrix, and Torelli then recovers the
conformal class \cite{ColboisProvenzanoSavo2025}.  The preceding corollary is
the homologically marked infinitesimal form of the same mechanism.  These
earlier results establish the analytic and geometric reconstruction used
here; the present formulation places it in the continuum part of the
information hierarchy and makes explicit what the measured flux coordinates
do and do not retain.

The conclusion is stronger than mere spectral nonconstancy.  The class
\([L_\chi]\in K^0(S)\) is constant, but the quadratic term of one band near
the identity character reconstructs the complex geometry that distinguishes
the flat bundles as a holomorphic family.  This is the continuum analogue of
recovering an Albanese torus from the Bloch spectrum of a quantum graph
\cite{RueckriemenBloch}.

\begin{proposition}[Kinematical versus dynamical holomorphy]
\label{prop:kinematical-holomorphy}
Suppose a finite abelian Bloch matrix has the form
\[
 H_{\bm\theta}
   =H_0+\sum_{\nu=1}^M
      \left(T_\nu\e^{\ii\langle\bm q_\nu,\bm\theta\rangle}
      +T_\nu^*\e^{-\ii\langle\bm q_\nu,\bm\theta\rangle}\right),
 \tag{5.12}\label{eq:topological-bloch-matrix}
\]
where \(\bm q_\nu\in H_1(S;\ZZ)\), and \(H_0,T_\nu\) are independent of
\(J\).  After identifying the character tori by the marking, the matrix family
\eqref{eq:topological-bloch-matrix} is independent of the period matrix
\(\Omega_J\).  Hence the complex structure of \(\Jac(X_J)\) is kinematical
for this model.

If instead a coefficient \(T_\nu(J)\) or a kinetic operator varies with
\(J\), then a simple eigenvalue varies along any Teichm\"uller direction for
which the corresponding Feynman--Hellmann matrix element is nonzero.
Independently, if the spectral projector has a nonzero derivative in such a
direction, its quantum metric is nonzero there.  In either case the complex
structure is dynamically active.
\end{proposition}

\begin{proof}
The first assertion follows by inspection: only the real character
\(\bm\theta\mapsto\e^{\ii\langle\bm q,\bm\theta\rangle}\) appears, and that
pairing is fixed by the marking.  For the second, differentiate the resulting
Hamiltonian with respect to a tangent vector in \(\cT_g\) and apply
\eqref{eq:metric-fh}.  If the projector varies, the projector formula
\eqref{eq:projector-qgt} gives
\(\tfrac12\Tr((\partial P)^2)>0\) for every nonzero Hermitian derivative
\(\partial P\).
\end{proof}

This proposition makes the physical point precise.  A complex-geometric
parametrization does not by itself change an observable; the Hamiltonian or
its states must couple to that geometry.

\subsection{Mechanisms that activate the geometry}

Several such mechanisms arise naturally.

\begin{enumerate}[label=\textup{(\roman*)},leftmargin=2.6em]
\item \textbf{Continuum kinetic energy.}
  The adjoint in \(\nabla_\rho^*\nabla_\rho\), the volume form, and the heat
  kernel all depend on \(h_J\).  The length factors in the trace formula are
  measurable spectral data.

\item \textbf{Distance- and angle-dependent hopping.}
  If
  \[
        T_e(J)=F_e\bigl(d_{h_J}(x_{s(e)},x_{t(e)}),
                 \angle_{h_J}(e),\ldots\bigr),
        \tag{5.13}\label{eq:metric-hopping}
  \]
  then a Teichm\"uller deformation changes \(H\) while leaving the abstract
  labelled graph fixed.  Resonator and capacitive implementations make this
  dependence explicit \cite{XuEtAl2025}.

\item \textbf{Parallel transport and orbital frames.}
  Spin connection, magnetic connection, or orientation-sensitive orbitals
  attach geometric transport matrices to edges.  Closed-walk coefficients
  then include both \(\rho(\gamma_w)\) and a metric connection holonomy.

\item \textbf{Holomorphic state dependence.}
  A family of eigenprojectors may be holomorphic or pseudoholomorphic with
  respect to the complex structure on a Jacobian or a two-dimensional slice
  of a character variety.  This imposes a pointwise relation between Berry
  curvature and quantum metric; see
  \Cref{prop:pseudoholomorphic-bound}.

\item \textbf{Higgs or spectral data.}
  A Higgs field changes a spectral curve and can enter a covariant kinetic
  operator or a matrix-valued potential.  The underlying bundle class stays
  fixed while the operator changes.
\end{enumerate}

These mechanisms do not compete with topological band theory.  They supply
the continuous data on which topological invariants impose constraints.  The
next section develops this statement on the sector side, where complex
structures organize flat and Higgs bundles; \Cref{sec:quantum-geometry}
develops it on the eigenstate side, where holomorphy constrains the
differential of the spectral projector.

\section{Holomorphic bundles, Higgs fields, and spectral curves}
\label{sec:holomorphy}

\subsection{The Narasimhan--Seshadri complex structure}

The stratified real-analytic character space \(\cR_r(S)\) is defined without
choosing \(J\); its smooth irreducible locus carries the symplectic structure
used below.
The identification
\[
        \cR_r^{\mathrm{irr}}(S)
           \xrightarrow{\ \mathrm{NS}_J\ }
        \cM_{X_J}^s(r,0)
        \tag{6.1}\label{eq:NS-J}
\]
depends on \(J\).  The right-hand side is a complex algebraic variety, and
its K\"ahler metric may be described through gauge theory
\cite{NarasimhanSeshadri,DonaldsonNS,AtiyahBott,GarciaRabosoRayan}.  As \(J\)
varies, the same topological representation space receives a varying
complex-geometric interpretation.

The dependence can be seen directly in the tangent geometry.  At a smooth
irreducible point \([\rho]\), deformation theory identifies
\[
       T_{[\rho]}\cR_r^{\mathrm{irr}}(S)
          \cong H^1(S;\operatorname{ad}\rho),
       \tag{6.1a}\label{eq:character-tangent}
\]
where \(\operatorname{ad}\rho\) is the flat bundle with fibre
\(\mathfrak u(r)\) and monodromy induced by conjugation.  Choose an
\(\operatorname{Ad}\)-invariant inner product on \(\mathfrak u(r)\), and
represent cohomology classes by harmonic one-forms for the metric \(h_J\).
Then the Atiyah--Bott--Goldman symplectic form, the Hodge metric, and the
associated complex structure are
\begin{align}
 \omega_\rho([\alpha],[\beta])
    &=\int_S\langle\alpha\wedge\beta\rangle,\notag\\
 G_{J,\rho}([\alpha],[\beta])
    &=\int_S\langle\alpha\wedge *_{h_J}\beta\rangle,\notag\\
 \mathbb J_{J,\rho}[\alpha]
    &=[*_{h_J}\alpha].
 \tag{6.1b}\label{eq:nonabelian-kahler-data}
\end{align}
Here the brackets in the integrands pair the Lie-algebra coefficients.
Because \(*_{h_J}^2=-1\) on one-forms,
\[
        G_{J,\rho}(\,\cdot\,,\,\cdot\,)
          =\omega_\rho(\,\cdot\,,\mathbb J_{J,\rho}\,\cdot\,).
        \tag{6.1c}\label{eq:nonabelian-kahler-relation}
\]
The symplectic pairing is determined by the oriented surface, the cup
product, and \(\rho\); the metric and complex structure depend on \(J\)
\cite{AtiyahBott,GoldmanSymplectic}.  Thus the same real character space
acquires a family of K\"ahler geometries as the marked surface moves in
Teichm\"uller space.  Equation \eqref{eq:NS-J} identifies this Hodge-theoretic
complex structure with the complex structure on the stable-bundle moduli
space.

At rank one this is simply the passage from the real flux torus to the
polarized complex torus in \eqref{eq:period-jac}.  At higher rank it is more
substantial: stable bundles, theta divisors, determinant lines, and
holomorphic sections depend on \(X_J\).  None of these is recovered from the
constant class \(r[\cO_X]\).

This gives two logically distinct reasons that \(K^0(S)\) is insufficient.
It identifies all degree-zero rank-\(r\) bundles at the fibrewise level, and
it does not retain the \(J\)-dependent differential geometry of the
representation space in which those fibres move.  A response coefficient
formed from tangent vectors in \eqref{eq:character-tangent}, for example,
can depend on \(G_{J,\rho}\) even when every represented bundle has class
\(r[\cO_X]\).

This moduli-theoretic viewpoint is central to the development of hyperbolic
band theory in
\cite{MaciejkoRayan2022,KienzleRayan2022,RayanQuantumMatter2026}.
It also explains why the phrase ``higher-dimensional momentum'' is only part
of the story.  A nonabelian representation variety has singular strata,
symplectic leaves, complex structures, and natural line bundles.  Treating it
as a featureless set of labels discards structures that can control
degeneracy, transport, and quantization.

\subsection{Higgs fields vary at fixed \texorpdfstring{\(K\)}{K}-class}

Let \(L\to X\) be a holomorphic line bundle.  An \(L\)-twisted Higgs bundle is
a pair
\[
                   (E,\Phi),\qquad
       \Phi\in H^0\!\left(X,\End(E)\otimes L\right).
       \tag{6.2}\label{eq:twisted-higgs}
\]
The ordinary Higgs-bundle case has \(L=K_X\)
\cite{HitchinSelfDuality,SimpsonHiggsLocal}.  The characteristic polynomial
of \(\Phi\) defines the Hitchin map.  Writing \(\cM_L(r,d)\) for the moduli
space of rank-\(r\), degree-\(d\), \(L\)-twisted Higgs bundles, it is
\[
  h_L:\cM_L(r,d)\longrightarrow
       \bigoplus_{j=1}^r H^0(X,L^j),
  \qquad
  (E,\Phi)\longmapsto(a_1(\Phi),\ldots,a_r(\Phi)).
  \tag{6.3}\label{eq:hitchin-map}
\]
In the total space of \(L\), with tautological section \(\eta\), the spectral
curve is
\[
          \Sigma_\Phi
          =\left\{\det(\eta-\pi^*\Phi)=0\right\}
          \subset\operatorname{Tot}(L).
          \tag{6.4}\label{eq:spectral-curve}
\]
Here \(\pi:\operatorname{Tot}(L)\to X\) is the bundle projection.
For a smooth spectral curve, \(E\) can be reconstructed as the pushforward of
a line bundle on \(\Sigma_\Phi\), subject to the usual spectral
correspondence \cite{HitchinSelfDuality,BNR}.  This is one precise sense in
which a complex momentum variable can carry both lattice and band data in the
proposal of \cite{KienzleRayan2022}.

\begin{proposition}[Higgs collapse and spectral variation]
\label{prop:higgs-collapse}
Fix \(r\), \(d\), \(X\), and \(L\).  The forgetful map
\[
       \cM_L(r,d)\longrightarrow K^0(X),
       \qquad (E,\Phi)\longmapsto[E],
       \tag{6.5}\label{eq:higgs-to-k}
\]
is constant, with value determined by rank \(r\) and degree \(d\).  By
contrast, if the moduli space contains a pair \((E,\Phi)\) for which some
characteristic coefficient \(a_j(\Phi)\) is nonzero, then the Hitchin map
\eqref{eq:hitchin-map} is nonconstant.  This hypothesis holds, in particular,
for the ordinary \(L=K_X\) Hitchin system.

If \(\Sigma_\Phi\) is smooth, then
\[
 g(\Sigma_\Phi)
    =1+r(g(X)-1)+\frac{r(r-1)}2\deg L.
    \tag{6.6}\label{eq:spectral-genus}
\]
For \(L=K_X\), this becomes
\[
                    g(\Sigma_\Phi)=1+r^2(g(X)-1).
                    \tag{6.7}\label{eq:hitchin-genus}
\]
Thus the \(K\)-class of \(E\) does not determine the Higgs field, its
characteristic coefficients, or its spectral curve.
\end{proposition}

\begin{proof}
On a compact Riemann surface, the stable topological class of a complex vector
bundle is determined by rank and degree, so \eqref{eq:higgs-to-k} is
constant.  If \(a_j(\Phi)\neq0\), the scaling family
\((E,s\Phi)\), \(s\in\CC^\times\), satisfies
\[
                         a_j(s\Phi)=s^j a_j(\Phi),
\]
and hence has nonconstant Hitchin image.  In the ordinary Hitchin system,
the usual spectral construction supplies such points and the Hitchin map is
the nonconstant Hitchin fibration.

For the genus formula, the spectral equation is a section of
\(\pi^*L^r\).  The canonical bundle of \(\operatorname{Tot}(L)\) is
\(\pi^*(K_X\otimes L^{-1})\).  Adjunction gives
\[
 K_{\Sigma_\Phi}
    \cong
 \pi^*(K_X\otimes L^{r-1})\big|_{\Sigma_\Phi}.
\]
Because \(\pi:\Sigma_\Phi\to X\) has degree \(r\),
\[
 2g(\Sigma_\Phi)-2
   =r\bigl(2g(X)-2+(r-1)\deg L\bigr),
\]
which is equivalent to \eqref{eq:spectral-genus}.  Substituting
\(\deg K_X=2g(X)-2\) gives \eqref{eq:hitchin-genus}.
\end{proof}

\begin{example}[A \(K\)-invisible Higgs direction]
\label{ex:higgs-scaling}
Fix a stable degree-zero bundle \(E\) and a nonzero
\(\Phi\in H^0(X,\End(E)\otimes L)\).  The family
\((E,s\Phi)\), \(s\in\CC\), has constant underlying \(K\)-class.  Its
characteristic coefficients satisfy
\[
                   a_j(s\Phi)=s^j a_j(\Phi).
                   \tag{6.8}\label{eq:higgs-scaling}
\]
Except on the nilpotent locus, the spectral curve moves with \(s\).  Even on
the nilpotent locus, the Higgs field and its filtration data can vary while
\([E]\) remains fixed.
\end{example}

\subsection{How Higgs data can enter an operator}

The presence of a Higgs field in a parametrization is not by itself a
measurement.  A definite coupling is needed.  In the ordinary case
\(L=K_X\), a natural class is obtained from the Hitchin--Simpson connection
\[
                  D_{A,\Phi}=D_A+\Phi+\Phi^{*_{h_E}},
                  \tag{6.9}\label{eq:hitchin-simpson}
\]
where \(A\) is a unitary connection and \(h_E\) is a Hermitian metric.  For a
harmonic bundle, this connection is flat.  More generally, the self-adjoint
elliptic operator
\[
        \cH_{A,\Phi}
         =D_{A,\Phi}^*D_{A,\Phi}+V
         \tag{6.10}\label{eq:higgs-hamiltonian}
\]
depends on \(\Phi\) through first- and zeroth-order terms.  Along a smooth
family \(\Phi_s\), a simple eigenvalue obeys
\[
       \dot E_n(s)
         =\bra{u_n(s)}
           \frac{\dd}{\dd s}\cH_{A,\Phi_s}
           \ket{u_n(s)}.
           \tag{6.11}\label{eq:higgs-FH}
\]
There is no reason for this matrix element to vanish identically.  Hence
\Cref{prop:higgs-collapse} becomes physically operative in any system whose
effective connection or potential contains the Higgs data.

Equation \eqref{eq:higgs-hamiltonian} gives one illustrative mechanism.  A
general hyperbolic tight-binding Hamiltonian need not be a
Hitchin--Simpson Laplacian.  For an \(L\)-twisted Higgs field, \(\Phi\) can
enter a connection only after the model supplies a contraction or bundle map
relating \(L\) to the cotangent bundle or to another physical coupling
bundle.  Other couplings can instead use the spectral curve to define hopping
channels, orbital multiplets, or symmetry-reduced momenta.  In every case the
Hamiltonian must depend on \(\Phi\), or on data reconstructed from it, before
the Higgs geometry becomes dynamical.

\subsection{Coarse moduli, stacks, and the geometry of eigenstates}

There is a further distinction between spectral invariants and eigenstates.
The quotient in \eqref{eq:unitary-character} is a coarse space.  The quotient
stack
\[
        \mathfrak R_r(S)
          =[\Hom(\Gamma_g,\UU(r))/\UU(r)]
          \tag{6.12}\label{eq:rep-stack}
\]
retains stabilizer groups and carries the tautological representation family.
Over the coarse moduli space of stable bundles of rank \(r\) and degree
\(d\), a universal vector bundle is obstructed in general.  On the
fixed-determinant moduli space, Ramanan's nonexistence theorem rules out a
universal vector bundle whenever \(\gcd(r,d)>1\); in particular, fixing the
determinant does not remove the obstruction in degree zero when \(r>1\)
\cite{Ramanan1973}.  The coarse stable space nevertheless carries a universal
projective bundle.  For curves of genus at least three, its Brauer class
generates a cyclic Brauer group of order \(\gcd(r,d)\)
\cite{BiswasSengupta2018}.  The product of \(X\) with the moduli stack carries
the tautological vector bundle, while a framing at a base point rigidifies
the scalar automorphisms.  This distinction is reflected by the central
\(\UU(1)\)-stabilizer of an irreducible unitary representation, or by its
finite central remnant after the determinant is fixed.

More concretely, \(\mathfrak R_r(S)\) is presented by the action groupoid
\[
 \mathbf{Rep}^{\mathrm U}_r(\Gamma_g)
   :=\UU(r)\ltimes\Hom(\Gamma_g,\UU(r)).
\]
Its objects are representations, and a morphism
\(u:\rho\to\rho'\) is a unitary matrix satisfying
\(\rho'=u\rho u^{-1}\).  The orbit space
\(\cR_r(S)\) remembers only the isomorphism class of an object; the groupoid
also remembers its automorphism group.  Reducible sectors therefore carry
categorical information that is invisible on the coarse character space.

The Hamiltonians in \eqref{eq:twisted-H} form an equivariant family before
quotienting.  Their characteristic polynomials and spectra descend to the
coarse character space because of \eqref{eq:rho-conjugacy}.  Eigenvectors need
not descend as an ordinary globally framed vector bundle.  They may define a
projective, twisted, or stack-theoretic object.  This matters for Berry
transport: a spectral value is a conjugation-invariant function, whereas a
Berry connection requires control of how eigenspaces are glued.

In categorical terms, \(\rho\mapsto H_\rho\) is an equivariant construction:
a morphism \(u:\rho\to\rho'\) induces the corresponding unitary intertwiner
between \(H_\rho\) and \(H_{\rho'}\).  Taking the spectrum is a
set-valued invariant and descends to the orbit space.  Assigning an eigenspace
is instead a functorial family with descent data.  Passing to coarse moduli is
thus a genuine loss of isotropy and gluing information, distinct from the
later stabilization that produces a \(K\)-class.

For physical work, one can proceed in any of three compatible local ways:
work on a framed representation space, choose local slices for the conjugation
action, or formulate the eigenbundle equivariantly.  Globally, the stack is
the cleanest receptacle.  Stack-theoretic treatments of families of Higgs
data \cite{AzamRayanQuiver2026} and diffeological treatments of nonabelian
Hodge moduli \cite{AzamRayan2026} provide natural languages for such
parameter spaces.  Stratification can also become physically active at gap
closings and in entanglement-sensitive constructions
\cite{IkedaRayan2026}.  The present paper uses only the local smooth geometry
needed for spectral projections.  Globally, however, isotropy and descent
data remain part of the eigenstate problem.

Having identified the holomorphic geometry of the sector spaces, we now pass
to the geometry of the eigenspaces themselves.  This is a different level of
structure: a character variety or Higgs moduli space parametrizes sectors,
whereas the Berry connection and quantum metric are pulled back by the
spectral projector of a chosen Hamiltonian on that parameter space.

\section{Quantum geometry and pseudoholomorphic band data}
\label{sec:quantum-geometry}

\subsection{Spectral projectors and the quantum geometric tensor}

Let \(B\) be a smooth parameter region in the abelian Jacobian, an irreducible
character variety, a local slice, or a flux-control space.  Suppose a
rank-\(q\) cluster of eigenvalues of a smooth finite Hamiltonian \(H(x)\) is
separated from the rest of the spectrum.  Its spectral projector
\[
                  P:B\longrightarrow\operatorname{Gr}(q,N)
                  \tag{7.1}\label{eq:projector-map}
\]
is smooth \cite{Kato}.  Its spectral bundle is
\(\mathscr E_P=\operatorname{Ran}P\to B\); when the selected cluster consists
of all states below a Fermi gap, we also write
\(\mathscr E_{\mathrm{occ}}\).
The ambient trivialization \(B\times\CC^N\) induces the Berry connection
and its Hermitian curvature
\[
       \nabla^{\mathrm B}s=P\,\dd s,\qquad
       F^{\mathrm B}:=\ii(\nabla^{\mathrm B})^2
          =\ii P\,\dd P\wedge\dd P\,P .
       \tag{7.1a}\label{eq:berry-connection-curvature}
\]
For a line bundle, the real two-form denoted below by \(F\) is
\(\tr F^{\mathrm B}\).  This convention gives
\((2\pi)^{-1}\int_C F=c_1(\mathscr E_P)[C]\).

For a nondegenerate band, choose a local normalized eigenvector \(u\), so that
\(P=\ket u\bra u\).  The quantum geometric tensor is
\[
    Q_{\mu\nu}
       =\bra{\partial_\mu u}(I-P)\ket{\partial_\nu u}.
       \tag{7.2}\label{eq:qgt}
\]
Its real part is the quantum metric,
\[
                g_{\mu\nu}^{\mathrm Q}=\operatorname{Re}Q_{\mu\nu},
                \tag{7.3}\label{eq:quantum-metric}
\]
and, with our convention, the Berry curvature is
\[
                F_{\mu\nu}=-2\operatorname{Im}Q_{\mu\nu}.
                \tag{7.4}\label{eq:berry-curvature}
\]
Both are gauge invariant.  In projector form,
\begin{align}
g_{\mu\nu}^{\mathrm Q}
   &=\frac12\Tr\bigl((\partial_\mu P)(\partial_\nu P)\bigr),\notag\\
F_{\mu\nu}
   &=\ii\Tr\bigl(P[\partial_\mu P,\partial_\nu P]\bigr).
   \tag{7.5}\label{eq:projector-qgt}
\end{align}
Indeed,
\[
 \Tr\bigl(P[\partial_\mu P,\partial_\nu P]\bigr)
    =2\ii\,\operatorname{Im}Q_{\mu\nu},
\]
so the second identity in \eqref{eq:projector-qgt} agrees with the convention
in \eqref{eq:berry-curvature}.
For higher-rank occupied spaces, the Grassmannian metric has the first form,
while \(P\,\dd P\wedge\dd P\,P\) gives the nonabelian Berry curvature.

Geometrically, \(g^{\mathrm Q}\) and \(F\) are the pullbacks of the
Fubini--Study metric and K\"ahler form on projective space, up to the
normalization in \eqref{eq:berry-curvature}
\cite{ProvostVallee,Berry,SimonHolonomy}.  This statement already shows why
their content exceeds \(K\)-theory.  The \(K\)-class depends on the stable
homotopy class of \(P\), whereas the pullback tensors depend on its
differential.

\subsection{What \texorpdfstring{\(K^0(B)\)}{K0(B)} does and does not determine}

The following result isolates four independent losses: dispersion, connection
holonomy, local curvature, and metric.

\begin{theorem}[Quantum geometry is not a \(K\)-class]
\label{thm:qgt-not-k}
\begin{enumerate}[label=\textup{(\alph*)},leftmargin=2.5em]
\item On any compact parameter manifold \(B\), a fixed smooth projector
  \(P:B\to\operatorname{Gr}(q,N)\) is compatible with infinitely many smooth
  gapped Hamiltonians having different dispersions.

\item On \(B=S^1\), there are projector families with the same class in
  \(K^0(B)\) but different quantum metrics.

\item On \(B=S^1\), there are projector families with the same class in
  \(K^0(B)\) but different Berry holonomies.  On \(B=S^2\), there are
  homotopic projector families with different pointwise Berry-curvature
  profiles and the same integral Chern class.  On any base that admits a
  nonconstant projector family, there are homotopic families with different
  quantum-metric profiles.
\end{enumerate}
\end{theorem}

\begin{proof}
For (a), choose arbitrary smooth functions
\(\eps_-(x)<\eps_+(x)\) and set
\[
             H_{\eps_-,\eps_+}(x)
               =\eps_-(x)P(x)+\eps_+(x)(I-P(x)).
               \tag{7.6}\label{eq:fixed-projector-H}
\]
The occupied projector and its \(K\)-class are unchanged, while both energy
functions are arbitrary subject to the gap.

For (b), take \(B=S^1\) and let
\[
 u_a(k)=
 \begin{pmatrix}
   \cos(a\sin k)\\
   \sin(a\sin k)
 \end{pmatrix},
 \qquad k\in\RR/2\pi\ZZ .
 \tag{7.7}\label{eq:metric-circle-family}
\]
The line spanned by \(u_a(k)\) is a trivial line bundle for every \(a\), so
all projectors \(P_a=\ket{u_a}\bra{u_a}\) have the same \(K^0(S^1)\)-class.
Direct calculation gives
\[
                 g_a^{\mathrm Q}=a^2\cos^2(k)\,\dd k^2.
                 \tag{7.8}\label{eq:circle-metric}
\]
The metric vanishes for \(a=0\) and is nonzero for \(a\neq0\).

For Berry holonomy, take
\[
 v_\vartheta(k)=
 \begin{pmatrix}
  \cos(\vartheta/2)\\
  \e^{\ii k}\sin(\vartheta/2)
 \end{pmatrix}.
 \tag{7.9}\label{eq:berry-circle-family}
\]
Again the underlying line bundle over \(S^1\) is trivial for every
\(\vartheta\).  Its Berry phase is
\[
 \oint_{S^1}\ii\braket{v_\vartheta}{\dd v_\vartheta}
       =-2\pi\sin^2(\vartheta/2)
       \quad \bmod 2\pi,
 \tag{7.10}\label{eq:berry-circle}
\]
with the Hermitian Berry convention of
\eqref{eq:berry-connection-curvature}.  It varies continuously with
\(\vartheta\), while \(K^0(S^1)\) records only the rank.

For curvature redistribution, take the degree-one line projector
\(P:S^2\to\CC\mathrm P^1\) under the standard identification
\(S^2\cong\CC\mathrm P^1\).  Its Berry curvature is a nonzero multiple of
the area form.  Choose an orientation-preserving diffeomorphism
\(\varphi:S^2\to S^2\), isotopic to the identity, that is not
area-preserving.  Then \(P\) and \(P\circ\varphi\) are homotopic, but the
second curvature is \(\varphi^*F\neq F\); their integrals agree.

Finally, let \(P\) be any nonconstant projector family.  Its differential,
and hence its quantum metric, is nonzero at some point.  A diffeomorphism
supported in a coordinate neighbourhood of that point can be chosen
isotopic to the identity and with derivative that does not preserve the
metric tensor there.  Thus \(P\) and \(P\circ\varphi\) are homotopic, while
their quantum metrics \(g^{\mathrm Q}\) and
\(\varphi^*g^{\mathrm Q}\) differ.
\end{proof}

\begin{remark}
Part (a) shows that even complete knowledge of the occupied projector is
insufficient to recover dispersion.  Parts (b) and (c) show that the
\(K\)-class of that projector is insufficient to recover its differential
geometry.  These are distinct losses and should not be compressed into a
single word such as ``topology.''
\end{remark}

\subsection{Pointwise bounds and pseudoholomorphicity}

The tensor \(Q\) is positive semidefinite as a Hermitian form.  On an oriented
two-dimensional parameter surface with local coordinates \((x^1,x^2)\), this
implies
\[
       \det(g^{\mathrm Q})
          \geq \frac14F_{12}^2.
          \tag{7.11}\label{eq:det-bound}
\]
If the coordinates are orthonormal for a chosen conformal structure, it also
implies the trace bound
\[
       g_{11}^{\mathrm Q}+g_{22}^{\mathrm Q}
          \geq |F_{12}|.
          \tag{7.12}\label{eq:trace-bound}
\]
Roy derived the trace and determinant bounds and related their saturation to
an ideal-band condition in two-dimensional Chern-band coordinates
\cite{Roy2014}; see also
\cite{MeraOzawaKahler,MeraOzawaFlat,LiuQiangLuXie}.  The formulation below is
intrinsic to a parameter Riemann surface and makes the local
Cauchy--Riemann condition explicit.

\begin{proposition}[Pseudoholomorphic saturation]
\label{prop:pseudoholomorphic-bound}
Let \((B,J_B)\) be a Riemann surface and let
\(P:B\to\CC\mathrm P^{N-1}\) describe a nondegenerate band.  In an isothermal
coordinate \(z=x^1+\ii x^2\), the following statements hold.

\begin{enumerate}[label=\textup{(\alph*)},leftmargin=2.5em]
\item The determinant and trace inequalities
  \eqref{eq:det-bound}--\eqref{eq:trace-bound} hold pointwise.

\item If \(P\) is \(J_B\)-holomorphic or antiholomorphic, then
  \[
       g_{11}^{\mathrm Q}=g_{22}^{\mathrm Q},\qquad
       g_{12}^{\mathrm Q}=0,\qquad
       g_{11}^{\mathrm Q}+g_{22}^{\mathrm Q}=|F_{12}|.
       \tag{7.13}\label{eq:holomorphic-saturation}
  \]

\item Conversely, at an immersed point, equality in
  \eqref{eq:trace-bound} with a fixed orientation is equivalent to the
  horizontal derivative of a local state satisfying a Cauchy--Riemann
  equation
  \[
       (I-P)\partial_2u=\pm\ii(I-P)\partial_1u.
       \tag{7.14}\label{eq:horizontal-CR}
  \]
  Thus \(P\) is locally pseudo\-holomorphic or anti-pseudoholomorphic there.
\end{enumerate}
\end{proposition}

\begin{proof}
Put \(D_\mu u=(I-P)\partial_\mu u\).  Then
\[
   g_{\mu\nu}^{\mathrm Q}
       =\operatorname{Re}\braket{D_\mu u}{D_\nu u},
   \qquad
   F_{12}=-2\operatorname{Im}\braket{D_1u}{D_2u}.
\]
Cauchy--Schwarz gives
\[
 |\braket{D_1u}{D_2u}|^2
   \leq\|D_1u\|^2\|D_2u\|^2.
\]
Separating real and imaginary parts yields
\eqref{eq:det-bound}.  The arithmetic--geometric mean inequality then gives
\eqref{eq:trace-bound}.

If \(P\) is holomorphic or antiholomorphic, its horizontal differential is
complex linear or antilinear.  This is exactly
\eqref{eq:horizontal-CR}, which implies
\eqref{eq:holomorphic-saturation}.  Conversely, equality in the two
inequalities used above forces equal norms and equality in Cauchy--Schwarz
with purely imaginary proportionality factor of modulus one.  This gives
\eqref{eq:horizontal-CR}.
\end{proof}

\Cref{prop:pseudoholomorphic-bound} gives a concrete meaning to holomorphy in
band physics.  It is not merely a choice of complex notation.  It is a
pointwise constraint relating two measurable tensors.  \(K\)-theory fixes the
integral
\[
                \frac1{2\pi}\int_C F
                =\langle c_1(\mathscr E_P),[C]\rangle
                \tag{7.15}\label{eq:chern-integral}
\]
on a closed two-cycle \(C\), but it does not fix
\eqref{eq:holomorphic-saturation}, the distribution of \(F\), or
\(g^{\mathrm Q}\).

\subsection{Higher-genus momentum and complex directions}

The abelian hyperbolic Brillouin zone has real dimension \(2g\), rather than
two.  A complex structure on \(\Jac(X)\) organizes its tangent space into
\(g\) complex directions.  For a band projector
\[
             P:\Jac(X)\longrightarrow\operatorname{Gr}(q,N),
\]
one can ask whether \(P\) is holomorphic with respect to this structure, or
whether its restriction to a complex curve or a real two-plane is
pseudoholomorphic.

More precisely, identify
\[
 T_\chi\Jac(X)\cong H^1(S;\RR)
\]
by translation and equip it with the Hodge complex structure
\(\mathbb J_J\) of \eqref{eq:hodge-complex}.  Let
\(\mathbb J_{\mathrm{Gr}}\) denote the standard complex structure on the
Grassmannian.  The antiholomorphic part of the projector differential is
\[
 \bar\partial_JP
   :=\frac12\left(
       \dd P+
       \mathbb J_{\mathrm{Gr}}\circ\dd P\circ\mathbb J_J
     \right).
 \tag{7.16}\label{eq:dbar-projector}
\]
This expression is intrinsic: it uses the projector map rather than a choice
of eigenvector gauge.  The family is holomorphic exactly when
\(\bar\partial_JP=0\).

For a nondegenerate band, let \(v\) be a unit tangent vector for the Hodge
metric.  On the \(\mathbb J_J\)-invariant plane spanned by
\(v\) and \(\mathbb J_Jv\), \eqref{eq:trace-bound} becomes
\[
 g^{\mathrm Q}(v,v)
   +g^{\mathrm Q}(\mathbb J_Jv,\mathbb J_Jv)
       \geq
       \left|F(v,\mathbb J_Jv)\right|.
 \tag{7.17}\label{eq:complex-plane-bound}
\]
Equality with a consistent orientation is equivalent, at immersed points,
to the restriction of \(\bar\partial_JP\) vanishing or to its
antiholomorphic analogue vanishing.  Scanning a basis of
\(\mathbb J_J\)-invariant planes therefore tests more than isolated
two-dimensional slices: it measures the failure of the full
\(2g\)-dimensional family to respect the Jacobian complex structure.  For
higher-rank projectors the same statement is expressed using the
Grassmannian metric and traced Berry curvature.

The question is geometric and testable:

\begin{enumerate}[label=\textup{(\roman*)},leftmargin=2.6em]
\item use the period matrix to form a Hodge-complex pair of flux directions
  \(v,\mathbb J_Jv\) (an integral symplectic pair need not itself be such a
  pair);
\item measure transition rates or fidelity susceptibility to obtain the
  quantum metric;
\item measure Berry curvature through adiabatic transport or transverse
  response;
\item compare the two sides of \eqref{eq:trace-bound}.
\end{enumerate}

The period matrix matters here because it selects the complex combinations of
flux directions.  A different point of Teichm\"uller space changes those
combinations.  This is a dynamically active use of the Jacobian's complex
structure when the projector responds to it, in contrast with the fixed
Fourier polynomial of \Cref{prop:kinematical-holomorphy}.

\subsection{Physical consequences of the lost geometry}

The standard response formulae transfer directly to any smooth flux or
representation parameter space, with suitable interpretation of the
controls.  Several consequences are particularly relevant.

\paragraph{Dispersion.}
The group velocity and effective-mass tensor use first and second derivatives
of \(E_n\).  These can vary in a \(K\)-fixed family even when the occupied
projector is held fixed, by \eqref{eq:fixed-projector-H}.

\paragraph{Adiabatic and Hall-type response.}
Berry curvature enters semiclassical anomalous velocity and Kubo response
\cite{XiaoChangNiu}.  A Chern number determines a quantized response in a
fully gapped, fully filled setting, but partially filled response and local
dynamics depend on the curvature profile.  Hyperbolic Chern models already
exhibit linear-response structures associated with their
higher-dimensional reciprocal space \cite{SunEtAl2024}.

\paragraph{Localization.}
The Brillouin-zone integral of the quantum metric contributes to the
gauge-invariant part of the Wannier spread
\cite{MarzariVanderbilt,MarzariReview}.  The topological class may obstruct
exponential localization, but even in an unobstructed class it does not
specify the localization length or the optimal spread.

\paragraph{Spectroscopic access.}
Periodic modulation of parameters measures components of the quantum metric
through excitation rates \cite{OzawaGoldman}.  In a hyperbolic circuit, the
natural parameters are flux twists and tunable couplings, so the protocol
can distinguish two devices with the same quotient topology and band
\(K\)-class.

\paragraph{Interacting flat bands.}
Geometric contributions to superfluid weight survive when ordinary
dispersion is suppressed \cite{PeottaTorma}.  This is an especially clear
case in which stable topology alone does not determine whether a nearly flat
band is physically useful.  Its quantum metric and holomorphic quality
matter.

The need to go beyond \(K\)-theory is therefore not an abstract gesture.  It is
forced by quantities that enter derivatives, transition rates,
localization bounds, and interaction scales.  The next step is to place these
losses beside observables that topology does determine.  That comparison is
the starting point of the Hall and gapless analysis.

\section{Hall response, Fermi geometry, and hyperbolic semimetals}
\label{sec:fermi-hall}

\subsection{Where topology is enough: quantized Hall response}

The central thesis is clearest when placed beside a case in which topology
does determine an observable.  Let \(C\subset B\) be a closed oriented
two-cycle in a gapped parameter space, and let \(\mathscr E_{\mathrm{occ}}\)
be the occupied bundle.  For a nondegenerate occupied band, our curvature
convention gives
\[
       c_1(\mathscr E_{\mathrm{occ}})[C]
          =\frac{1}{2\pi}\int_C F\in\ZZ .
          \tag{8.1}\label{eq:chern-pairing-hall}
\]
For several occupied bands one uses the trace of the nonabelian Berry
curvature.  If \(C\) is an ordinary two-dimensional Brillouin torus, the
Thouless--Kohmoto--Nightingale--den Nijs (TKNN) formula identifies this
integer with the Hall conductance,
\[
       \sigma_{12}
          =\frac{e^2}{h}\,
             c_1(\mathscr E_{\mathrm{occ}})[C],
          \tag{8.2}\label{eq:tknn-hall}
\]
for a fully filled, gapped, noninteracting system
\cite{TKNN}.  If \(C\) is instead a two-torus of boundary twists and the
many-body ground state remains separated by a gap, the
Niu--Thouless--Wu construction gives the corresponding many-body statement
\cite{NiuThoulessWu}.

There is a parallel noncommutative formulation on the hyperbolic plane.  A
Fermi projection in a twisted Fuchsian-group algebra can be paired with the
hyperbolic Connes--Kubo cyclic two-cocycle.  For a compact good hyperbolic
orbifold \(\Sigma_{g,\bm\nu}\), with cone-point orders
\(\bm\nu=(\nu_1,\ldots,\nu_n)\), Marcolli and Mathai identify this cocycle
with the area cocycle in cyclic cohomology and show that its normalized
pairing takes values in
\[
             \phi\ZZ,\qquad
             \phi=-\chi_{\mathrm{orb}}(\Sigma_{g,\bm\nu})
                  =2g-2+n-\sum_{j=1}^n\frac1{\nu_j}.
\]
Thus the allowed Hall plateaux in spectral gaps are integral multiples of an
orbifold Euler-characteristic scale
\cite{MarcolliMathai2001,MarcolliMathai2006}.  For the smooth torsion-free
surfaces considered in most of this article, \(n=0\) and
\(\phi=2g-2\) is integral.  The genuinely fractional step size in that
construction comes from orbifold isotropy; it should not be attributed to
negative curvature alone.

Two complementary developments make the geometric input to this positive
topological result especially explicit.  Mathai and Wilkin study a family of
Landau Hamiltonians on a compact complex two-orbifold parametrized by its
Jacobian.  In the large-field regime, the associated stable holomorphic
spectral orbibundles have computable fractional degree, yielding fractional
conductance and charge transport through an orbifold Nahm transform
\cite{MathaiWilkin2019}.  Mathai and Thiang use noncommutative geometry and
T-duality to formulate topological phases on the hyperbolic plane and propose
a fractional bulk--boundary correspondence guided by hyperbolic boundary
geometry \cite{MathaiThiang2019}.  In both cases the quantized output is
topological, but its construction retains geometric structures---Jacobian
families, orbifold holomorphy, or boundary geometry---that a bare numerical
\(K\)-pairing does not display.

Marcolli and Seipp subsequently extended this construction from a
single-particle orbifold to symmetric products describing indistinguishable
particles.  To avoid collision with the number \(n\) of cone points above,
write \(N_{\mathrm p}\) for the particle number.  The configuration orbifold
is
\(\operatorname{Sym}^{N_{\mathrm p}}(\Sigma_{g,\bm\nu})\), and projective
covariance is encoded by the magnetic multiplier \(\sigma_{N_{\mathrm p}}\)
on the wreath product
\[
 \Gamma_{g,\bm\nu}^{\,N_{\mathrm p}}\rtimes
       \mathfrak S_{N_{\mathrm p}},
\]
or equivalently by its twisted reduced group \(C^*\)-algebra.  Their higher
twisted index theorem gives the range
\[
 \chi_{\mathrm{orb}}\!\left(
       \operatorname{Sym}^{N_{\mathrm p}}(\Sigma_{g,\bm\nu})\right)\ZZ
   =
 \frac{\chi_{\mathrm{orb}}(\Sigma_{g,\bm\nu})^{N_{\mathrm p}}}
      {N_{\mathrm p}!}\,\ZZ
\]
for the corresponding Hall pairing, where \(\chi_{\mathrm{orb}}\) denotes
the Satake orbifold Euler characteristic.  The same framework organizes the
orbifold \(K\)-groups across particle number into a Fock-space-type graded
object and develops composite-fermion and finite-dimensional anyon
representations through orbifold braid groups and Seifert data
\cite{MarcolliSeipp2017}.  Its proposed route to Laughlin-type wave functions
is explicitly exploratory.  This many-particle extension therefore
strengthens the positive case for index pairings while also showing that the
quantized pairing does not, by itself, specify the wave function or the full
dynamics.

This language is particularly natural on a compact hyperbolic surface.
There are \(2g\) independent abelian twist angles.  Choosing two of them, or
more invariantly a two-cycle in the Jacobian, defines a Hall-type adiabatic
response associated with those flux controls.  The resulting integer is a
pairing of the occupied \(K\)-class with \([C]\).  In this precise setting,
\(K\)-theory \emph{is} enough.

Higher-dimensional reciprocal spaces permit higher Chern pairings.  Let
\(Y\subset B\) be a four-cycle on which
\(c_1(\mathscr E_{\mathrm{occ}})\) vanishes.  In the standard physics
normalization, the degree-four Chern-character pairing is
\[
       \mathcal C_2[Y]
         =\frac{1}{8\pi^2}
            \int_Y\tr(F^{\mathrm B}\wedge F^{\mathrm B})\in\ZZ .
         \tag{8.3}\label{eq:second-chern-response}
\]
With the Hermitian curvature convention
\eqref{eq:berry-connection-curvature}, this is exactly
\(\langle\ch_2(\mathscr E_{\mathrm{occ}}),[Y]\rangle\).  Under the stated
hypothesis \(c_1=0\), it equals
\(-\langle c_2(\mathscr E_{\mathrm{occ}}),[Y]\rangle\).  Without that
hypothesis, the Chern--Weil representative of \(c_2\) also contains
\(\tr(F^{\mathrm B})\wedge\tr(F^{\mathrm B})\), so
\eqref{eq:second-chern-response} remains the degree-four Chern-character
pairing rather than the second Chern number.
Hyperbolic band models with a four-dimensional abelian reciprocal space have
realized a nontrivial degree-four Chern structure
\cite{ZhangEtAlSecondChern,TummuruEtAl2024}.  Its nonlinear-response
interpretation requires four independent control directions; the fact that
the real-space network is two-dimensional does not remove that experimental
requirement.

Quantization relies on a filled subspace and a gap.  At partial filling, the
intrinsic Berry contribution takes the schematic form
\[
 \sigma_{ab}^{\mathrm{int}}(\mu_{\mathrm F},T)
   =-\frac{e^2}{\hbar}
      \sum_n\int_B
       f_T(E_n(x)-\mu_{\mathrm F})\,F_{ab}^{(n)}(x)\,
       \dd\mu_B(x),
 \tag{8.4}\label{eq:partial-hall}
\]
where \(a,b\) label response directions, \(\mu_{\mathrm F}\) is the chemical
potential, \(f_T\) is the Fermi--Dirac function at temperature \(T\), and
\(\dd\mu_B\) is a chosen normalized measure on the momentum or control family.
For a finite-rank hyperbolic character space this is a sector-sampling
measure, not the Plancherel measure of the infinite lattice.  The exact
thermodynamic bulk response must instead be formulated using the regular
group trace and the appropriate current cocycle, or obtained through a
specified convergent finite-cover or large-rank scheme.  Equation
\eqref{eq:partial-hall} nevertheless isolates the relevant distinction for a
controlled sector family: the occupation function depends on dispersion, and
the integrand depends on the local curvature profile.

\begin{proposition}[The Hall dichotomy]
\label{prop:hall-dichotomy}
\begin{enumerate}[label=\textup{(\alph*)},leftmargin=2.5em]
\item For a fully occupied gapped bundle over a closed two-cycle, the Chern
  pairing \eqref{eq:chern-pairing-hall} is determined by the occupied
  \(K\)-class.  Whenever the physical hypotheses of the TKNN or
  Niu--Thouless--Wu construction identify that pairing with a Hall response,
  the quantized response is therefore determined as well.
\item The same \(K\)-class does not determine the Berry-curvature profile or
  the partially filled response \eqref{eq:partial-hall}.
\end{enumerate}
\end{proposition}

\begin{proof}
Part \textup{(a)} is the Chern-character pairing of the \(K\)-class with
\([C]\).  For part \textup{(b)}, precompose a projector with a diffeomorphism
isotopic to the identity, as in the proof of
\Cref{thm:qgt-not-k}.  This preserves the \(K\)-class and total Chern number
but can redistribute \(F\).  For the partially filled response, fix a band
projector whose relevant curvature component is nonzero on an open set.
The construction \eqref{eq:fixed-projector-H} allows its isolated-band energy
to be varied there while preserving both the projector and its gap to the
complementary band.  Choose two energies for which the corresponding
Fermi--Dirac weights at some \(T>0\) differ on a smaller region where the
curvature has fixed sign.  The weighted integrals in
\eqref{eq:partial-hall} are then different, although the projector and its
\(K\)-class are unchanged.
\end{proof}

This distinction is also visible in hyperbolic Chern-insulator response:
topological pairings constrain quantized limits, while the full linear
response retains the geometry of the higher-dimensional reciprocal space
\cite{SunEtAl2024}.  The correct message is not that Hall conductance escapes
topology, but that topology determines it only under the hypotheses that make
it quantized.

\subsection{Fermi surfaces and van Hove data}

Let \(B\) be a compact smooth \(D\)-dimensional momentum or parameter
manifold with a Riemannian metric \(G\), and let
\(w\,\dd\mathrm{vol}_G\) be the measure used to sample sectors, with
\(w\) a smooth nonnegative density.  For a smooth band
\(E_n:B\to\RR\), define the generalized Fermi surface at chemical potential
\(\mu_{\mathrm F}\) by
\[
          \mathcal F_{n,\mu_{\mathrm F}}
             =E_n^{-1}(\mu_{\mathrm F}).
                    \tag{8.5}\label{eq:fermi-level-set}
\]
If \(\mu_{\mathrm F}\) is a regular value, this is a smooth hypersurface of
dimension \(D-1\).  In the abelian hyperbolic theory \(D=2g\), so a regular
Fermi surface has dimension \(2g-1\).  The word ``surface'' here refers to
the generalized reciprocal space.  The formula below is a sector-ensemble
coarea formula and describes, for example, an externally scanned flux family.
For a nonamenable surface group it is not a decomposition of the exact
thermodynamic density of states over a fixed finite-dimensional character
space.  The relation to the real-space bulk measure requires the coherent
large-rank or finite-cover limits described in
\Cref{sec:finite-boundaries}.

\begin{proposition}[Fermi geometry and the density of states]
\label{prop:fermi-coarea}
Suppose \(\mu_{\mathrm F}\) is a regular value of each band meeting that
energy.  The integrated state count
\[
       \mathcal N(E)
          =\sum_n\int_B
              \mathbf1_{\{E_n(x)\leq E\}}\,
              w(x)\,\dd\mathrm{vol}_G(x)
              \tag{8.6}\label{eq:integrated-state-count}
\]
is differentiable at \(E=\mu_{\mathrm F}\), with
\[
       \nu(\mu_{\mathrm F}):=\mathcal N'(\mu_{\mathrm F})
          =\sum_n\int_{\mathcal F_{n,\mu_{\mathrm F}}}
             \frac{w(x)}
                  {|\nabla_G E_n(x)|_G}\,
             \dd A_G(x).
             \tag{8.7}\label{eq:fermi-coarea}
\]
Here \(\dd A_G\) is the hypersurface measure induced by \(G\).
Neither \(\mathcal F_{n,\mu_{\mathrm F}}\), its induced geometry, nor
\(\nu(\mu_{\mathrm F})\) is determined by the \(K\)-class of the band
projector.  At a metallic Fermi level, a global gapped occupied-bundle class
need not exist at all.
\end{proposition}

\begin{proof}
Equation \eqref{eq:fermi-coarea} is the coarea formula applied to each
\(E_n\).  To prove the final assertion, hold the projector fixed and vary
the energy functions as in \eqref{eq:fixed-projector-H}.  The \(K\)-class
does not change, while the level sets, their gradients, and the right-hand
side of \eqref{eq:fermi-coarea} can all change.
\end{proof}

At a critical value of \(E_n\), the denominator in
\eqref{eq:fermi-coarea} vanishes and the regular-level formula breaks down.
The resulting van Hove behaviour is governed locally by the Hessian and
higher jets of \(E_n\).  These are exactly the spectral derivatives retained
by \Cref{prop:position-flux} and \eqref{eq:second-order-response}, and exactly
the kind of data forgotten by a \(K\)-class.  Thus Fermi-surface shape,
Fermi velocity, effective mass, and van Hove singularities provide familiar
condensed-matter manifestations of the paper's general nonfactorization
principle.

\subsection{Dirac points, Weyl points, and nodal manifolds}

A Fermi point or nodal manifold is not described by a globally gapped
occupied bundle.  Topology is instead assigned on a small sphere transverse
to the degeneracy, or through a relative \(K\)-class.  The following local
model makes the division between charge and geometry precise.

\begin{proposition}[Clifford nodes: charge versus geometry]
\label{prop:clifford-nodes}
Let \(B\) be a smooth \(D\)-manifold and let
\[
        H(x)=d_0(x)I+\sum_{a=1}^{2m+1}d_a(x)\Gamma_a,
        \qquad
        \Gamma_a\Gamma_b+\Gamma_b\Gamma_a=2\delta_{ab}I,
        \tag{8.8}\label{eq:clifford-hamiltonian}
\]
be a smooth local Hamiltonian.  Put
\(\bm d=(d_1,\ldots,d_{2m+1})\) and
\(\bm\Gamma=(\Gamma_1,\ldots,\Gamma_{2m+1})\).

\begin{enumerate}[label=\textup{(\alph*)},leftmargin=2.5em]
\item If \(0\) is a regular value of \(\bm d\), the degeneracy locus
  \[
                         \mathcal Z=\bm d^{-1}(0)
                         \tag{8.9}\label{eq:nodal-locus}
  \]
  is a smooth submanifold of codimension \(2m+1\).

\item On a small normal sphere
  \(S^{2m}\) linking a component of \(\mathcal Z\), the normalized map
  \(\widehat{\bm d}=\bm d/|\bm d|\) has a degree.  Up to the orientation
  convention for the Clifford matrices, this degree is the normalized top
  Chern-character pairing of the lower-energy eigenbundle---the quantity
  conventionally called the \(m\)-th Chern number in band physics.  It is
  the local topological charge of the node.

\item This charge does not determine the embedding
  \(\mathcal Z\hookrightarrow B\), the nodal energy \(d_0|_{\mathcal Z}\),
  the tilt \(\dd d_0\), the transverse velocity map
  \(\dd\bm d|_{N\mathcal Z}\), or higher-order dispersion.
\end{enumerate}
\end{proposition}

\begin{proof}
Part \textup{(a)} is the regular-value theorem.  On a normal sphere the
Hamiltonian is gapped, and spectral flattening replaces it by
\(\widehat{\bm d}\cdot\bm\Gamma\).  The standard Clifford projector defines
the Bott generator; its normalized top Chern-character pairing is the degree
of \(\widehat{\bm d}\), proving \textup{(b)}.  For \textup{(c)}, translations
of the zero set, changes of \(d_0\), and orientation-preserving deformations
of the derivative of \(\bm d\) leave the linking-sphere degree unchanged.
They can vary all of the listed differential data.
\end{proof}

For \(m=1\), a two-band Hamiltonian has three Pauli components.  In a
three-dimensional effective momentum space, a regular zero is an isolated
Weyl point.  In local coordinates,
\[
 H(\bm k_*+\bm q)
   =E_*I+\bm w\cdot\bm q\,I
      +\sum_{a,i=1}^3v_{ai}q_i\sigma_a+O(|\bm q|^2).
 \tag{8.10}\label{eq:weyl-linearization}
\]
When \(v\) is invertible, the local charge is
\(\operatorname{sgn}\det v\).  The charge remembers neither
\(\bm k_*\), \(E_*\), the tilt \(\bm w\), nor the positive quadratic form
\(v^{\mathsf T}v\) that sets the anisotropy of the cone.  In the untilted
linear model, the density-of-states coefficient is proportional to
\(|\det v|^{-1}\), while the charge retains only its sign.  This is a
particularly economical example of topology classifying stability without
determining low-energy physics.

Node position can itself enter a Hall response.  In the simplest
three-dimensional lattice Weyl semimetal with one pair of oppositely charged
nodes separated by \(\Delta\bm k\), the intrinsic anomalous Hall tensor has
the form
\[
       \sigma_{ij}
         =\frac{e^2}{2\pi h}\,
            \epsilon_{ijk}\Delta k_k ,
       \tag{8.10a}\label{eq:weyl-hall-separation}
\]
where \(\epsilon_{ijk}\) is the Levi--Civita symbol, up to
reciprocal-lattice ambiguities and contributions from other filled bands.
The local Chern charges determine the orientation of the Berry
sources and sinks, but not their separation.  Thus even a Hall response can
depend on geometric information in a gapless phase
\cite{TummuruEtAl2024}.

Dimension counting matters in hyperbolic reciprocal space.  For a generic
two-band family on \(\Jac(X)\), \(\dim_\RR B=2g\), so
\[
                \dim\mathcal Z=2g-3
                \tag{8.11}\label{eq:two-band-nodal-dimension}
\]
when the Pauli-vector map is transverse to zero.  In genus two, a generic
two-band degeneracy is therefore a nodal line, not an isolated Weyl point.
Dirac or Fermi points on two- or three-dimensional slices require symmetry,
parameter restriction, or additional tuning.  Hyperbolic magnetic models
do exhibit Dirac cones on suitable coordinate neighbourhoods
\cite{IkedaAokiMatsuki}; their locations and cone tensors remain geometric
data beyond their Berry phases or local charges.

For \(m=2\), five Clifford components give a codimension-five node whose
linking \(S^4\) carries a second Chern number.  This is the local structure
realized in the hyperbolic nonabelian semimetal of
\cite{TummuruEtAl2024}: the abelian spectrum can be gapped while a nodal
manifold appears in higher-rank representation space and is protected by
\(C_2\).  That work is therefore a positive demonstration of both halves of
the present thesis.  The second Chern number protects the node, while the
nonabelian character geometry, embedding of the nodal manifold, transverse
dispersion, the prescription used to pass from sector data to the regular
trace, and density-of-states scaling contain further physical information.

For a linearly dispersing codimension-\(p\) nodal manifold, the local
coarea argument gives the clean scaling
\[
                         \nu(E)\sim C_{\mathcal Z}|E|^{p-1},
                         \tag{8.12}\label{eq:nodal-dos-scaling}
\]
where \(C_{\mathcal Z}\) contains the volume of the nodal manifold and the
inverse transverse velocity determinant.  Global hyperbolic sector weights,
nonlinear dispersion, and singular strata can modify the observed exponent
or prefactor.  None of these metric quantities follows from the local Chern
charge.

\subsection{Semimetals, ordinary metals, and the scope of
  \texorpdfstring{\(K\)}{K}-theory}

The word ``gapless'' covers several mathematically different situations.
\Cref{tab:metal-regimes} records the distinction relevant here.

\begin{table}[ht]
\centering
\small
\begin{tabular}{@{}
 >{\raggedright\arraybackslash}p{0.18\textwidth}
 >{\raggedright\arraybackslash}p{0.31\textwidth}
 >{\raggedright\arraybackslash}p{0.40\textwidth}@{}}
\toprule
\textbf{regime} & \textbf{topological datum} & \textbf{additional physical data}\\
\midrule
gapped insulator
 & global occupied \(K\)-class; Chern pairings
 & gap size, dispersion, curvature profile, quantum metric\\
topological semimetal
 & local or relative class on a linking sphere
 & node position and energy, velocity, tilt, nodal embedding, density of states\\
Fermi-liquid metal
 & possible local or relative \(K\)-class protecting a stable Fermi manifold
 & Fermi-surface geometry, quasiparticle residue, lifetime, interactions\\
strange metal
 & generally no quasiparticle occupied bundle at the Fermi level
 & self-energy, spectral weight, relaxation, scaling, transport\\
\bottomrule
\end{tabular}
\caption{Topology and geometry in gapped and gapless regimes.}
\label{tab:metal-regimes}
\end{table}

Stable Fermi manifolds themselves admit \(K\)-theoretic classifications
\cite{HoravaFermiK}.  Such a class determines the stability type and charge,
not the embedded Fermi geometry or the coefficients of the effective
Hamiltonian.  Likewise, a semimetal is not specified by the statement that a
node is topologically protected.  One must also know whether the chemical
potential passes through it, whether other Fermi pockets are present, and how
the dispersion controls low-energy phase space.

\subsection{Beyond the single-particle framework: strange metals}

A topological semimetal and a strange metal should not be conflated.  The
former can be described by a nodal single-particle Hamiltonian.  The latter is
an interacting metallic regime in which a long-lived quasiparticle
description may fail.  Schematically, its retarded one-particle Green
function satisfies
\[
 [G^{\mathrm R}(\omega,x)]^{-1}
   =(\omega+\mu_{\mathrm F}+\ii0^+)I-H(x)-\Sigma^{\mathrm R}(\omega,x),
 \tag{8.13}\label{eq:interacting-green}
\]
and the spectral function
\[
 A(\omega,x)
    =-\frac1\pi\operatorname{Im}\Tr G^{\mathrm R}(\omega,x)
 \tag{8.14}\label{eq:spectral-function}
\]
depends essentially on the self-energy \(\Sigma^{\mathrm R}\).
Marginal-Fermi-liquid phenomenology provides a canonical example in which
anomalous scattering, rather than a band \(K\)-class, organizes the normal
state \cite{VarmaMarginalFermi}.

Hyperbolic geometry can plausibly influence an interacting problem through
the bare density of states, flatness, frustration, connectivity, quantum
metric, and interaction form factors.  These are concrete mechanisms to
investigate, but they do not by themselves imply strange-metal behaviour,
linear-in-temperature resistivity, or Planckian relaxation.  Establishing
such a phase requires a specified interacting model and control of the
self-energy, vertex corrections or current correlators, and thermodynamic or
transport scaling.  When the quasiparticle bundle ceases to be the principal
object, its \(K\)-class contains still less of the relevant physics, but band
geometry alone is not a replacement for a many-body theory.  This extension
clarifies the scope of the single-particle results.  We now ask whether
operator-algebraic or differential refinements of \(K\)-theory close the
remaining gap.

\section{Operator-algebraic and differential
  \texorpdfstring{\(K\)}{K}-theory}
\label{sec:operator-k}

\subsection{Three \texorpdfstring{\(K\)}{K}-theoretic questions}

One might respond to \Cref{thm:k-collapse} by replacing \(K^0(S)\) with a
more sophisticated \(K\)-theory.  That is often the correct move for
classifying phases, but it does not reconstruct the Hamiltonian.  We consider
the principal alternatives separately.

\paragraph{\(K^0(S)\).}
This is the recipient of the flat-bundle class.  At fixed rank and degree it
is constant, as proved above.

\paragraph{\(K^0(B)\).}
This classifies the stable occupied bundle over a chosen Brillouin or moduli
space \(B\).  It can detect Chern classes that \(K^0(S)\) does not.  It still
forgets dispersion and differential geometry by
\Cref{thm:qgt-not-k}.

\paragraph{\(K_*(\mathcal A)\).}
Here \(\mathcal A\) is an algebra of covariant observables, such as a reduced
group algebra, a crossed product, or an algebra built from a hull.  This is
the natural setting for noncommutative gap labels and disordered systems
\cite{BellissardKTheory,BellissardVanElstSchulzBaldes,ConnesNCG}.

These are not simply three successively larger groups.  They classify
different objects.  The class in \(K^0(S)\) belongs to a position-space flat
or holomorphic bundle; the class in \(K^0(B)\) belongs to a parameter-space
spectral bundle; and a class in \(K_*(\mathcal A)\) is represented by a
projection, unitary, or module over the observable algebra.  Relating them
requires an additional construction---for example a Poincar\'e transform, a
families index, or the choice of a Fermi projection.  There is no
model-independent map under which one of the three automatically contains
the other two.

It follows that the phrase ``use a more sophisticated \(K\)-theory'' is
incomplete until the classified object, algebra, and observable have been
specified.  Operator \(K\)-theory can retain disorder covariance and gap
labels that a fibrewise class in \(K^0(S)\) misses.  It still applies a
stable equivalence relation to a chosen projection or unitary and therefore
does not recover the self-adjoint element from which that projection arose.
These refinements retain different information; none is designed to be a
complete spectral invariant.

\subsection{Surface-group \texorpdfstring{\(C^*\)}{C-star}-algebras}

Because \(\Gamma_g\) acts freely, properly, and cocompactly on \(\HH^2\),
Green imprimitivity gives a Morita equivalence
\[
       C_0(\HH^2)\rtimes_r\Gamma_g
                 \ \sim_{\mathrm M}\ C(S).
       \tag{9.1}\label{eq:green-morita}
\]
Consequently
\[
 K_0\!\left(C_0(\HH^2)\rtimes_r\Gamma_g\right)\cong\ZZ\oplus\ZZ,
 \qquad
 K_1\!\left(C_0(\HH^2)\rtimes_r\Gamma_g\right)\cong\ZZ^{2g}.
 \tag{9.2}\label{eq:crossed-k}
\]
The reduced group algebra \(C_r^*(\Gamma_g)\) has the same group ranks.  The
Baum--Connes assembly map is an isomorphism for surface groups, and
spin\(^{c}\) Poincar\'e duality on \(S=B\Gamma_g\) identifies its domain with
the surface \(K\)-groups
\cite{BaumConnesHigson,HigsonKasparov,GreenImprimitivity}.

The twisted orbifold calculation gives a useful enlargement of this smooth
surface case.  Let \(\Gamma_{g,\bm\nu}\) be a cocompact Fuchsian group of
signature \((g;\nu_1,\ldots,\nu_n)\), and let \(\sigma\) be a multiplier with
trivial Dixmier--Douady invariant.  Marcolli and Mathai compute
\[
\begin{aligned}
 K_0\!\left(C^*(\Gamma_{g,\bm\nu},\sigma)\right)
    &\cong \ZZ^{\,2-n+\sum_{j=1}^n\nu_j},\\
 K_1\!\left(C^*(\Gamma_{g,\bm\nu},\sigma)\right)
    &\cong \ZZ^{\,2g},
\end{aligned}
\]
and the range of the canonical trace on the reduced twisted algebra as
\[
 \tau_{\sigma,*}\!
 \left(K_0\!\left(C_r^*(\Gamma_{g,\bm\nu},\sigma)\right)\right)
   =\ZZ+\theta\ZZ+\sum_{j=1}^n\frac1{\nu_j}\ZZ ,
\]
where \(\theta\) is the normalized pairing of \(\sigma\) with the orbifold
fundamental class \cite{MarcolliMathai1999}.  When \(n=0\), the \(K\)-groups
reduce to those of the torsion-free surface group and the trace range becomes
\(\ZZ+\theta\ZZ\).  Cone points add \(K_0\)-generators and rational trace
contributions.  This is a genuine hyperbolic counterpart of the
noncommutative-torus trace-range mechanism, but it retains the Fuchsian group
rather than replacing it by \(\ZZ^2\).

The adjacent constructions of Mathai--Wilkin and Mathai--Thiang make the same
point from two directions: a \(K\)-class becomes a physical conductance or
boundary index only after it is placed in a holomorphic families-index
construction or a bulk--boundary framework
\cite{MathaiWilkin2019,MathaiThiang2019}.

At fixed particle number \(N_{\mathrm p}\), the Marcolli--Seipp extension
replaces the Fuchsian group by the wreath product
\(\Gamma_{g,\bm\nu}^{\,N_{\mathrm p}}\rtimes
\mathfrak S_{N_{\mathrm p}}\) and uses the induced multiplier
\(\sigma_{N_{\mathrm p}}\).  The multiplier has trivial
Dixmier--Douady class, and the twisted and untwisted group-algebra
\(K\)-groups agree; summing the corresponding orbifold \(K\)-groups over
particle number produces the graded Fock-space structure mentioned in
\Cref{sec:fermi-hall} \cite{MarcolliSeipp2017}.  This is a substantial
many-particle enrichment of the classification target.  It does not,
however, turn the resulting \(K\)-classes or index pairings into a complete
invariant of a chosen many-particle Hamiltonian.

These finitely generated \(K\)-groups retain powerful stable information,
but they cannot encode the continuous range of spectral moments in
\eqref{eq:walk-trace} or the length data in \eqref{eq:twisted-heat}.  More
generally, the following elementary functional-calculus argument applies in
any unital \(C^*\)-algebra.

\begin{proposition}[A Fermi projection does not determine a spectrum]
\label{prop:functional-calculus}
Let \(h=h^*\in M_N(\mathcal A)\) have a gap containing a Fermi level \(\mu\),
and let
\[
                     p=\mathbf1_{(-\infty,\mu)}(h).
                     \tag{9.3}\label{eq:fermi-projection}
\]
If \(f:\RR\to\RR\) is continuous on \(\spec(h)\) and satisfies
\[
       (\lambda-\mu)\bigl(f(\lambda)-\mu\bigr)>0
       \qquad\text{for every }\lambda\in\spec(h),
\]
then
\[
           \mathbf1_{(-\infty,\mu)}(f(h))=p.
           \tag{9.4}\label{eq:same-fermi}
\]
Thus \(h\) and \(f(h)\) define the same class \([p]\in K_0(\mathcal A)\), while
their spectra are \(\spec(h)\) and \(f(\spec(h))\), which can differ.
\end{proposition}

\begin{proof}
The strict sign condition says both that a spectral value lies below \(\mu\)
exactly when its image does and that no image lies at \(\mu\).  Since
\(\spec(h)\) is compact, the continuous function
\(\lambda\mapsto|f(\lambda)-\mu|\) has a positive minimum there; hence
\(f(h)\) remains gapped at \(\mu\).  Continuous functional calculus then
identifies the two characteristic functions on \(\spec(h)\), so their
spectral projections are equal.
\end{proof}

\Cref{prop:functional-calculus} is the operator-algebraic version of
\eqref{eq:fixed-projector-H}.  It applies even when the full noncommutative
observable algebra is used.  The \(K\)-class of a Fermi projection determines
neither the energy scale nor the density of states within either spectral
component.

\subsection{Comparison with the Euclidean noncommutative torus}

The noncommutative two-torus provides the cleanest test of what is gained by
passing from a commutative Brillouin torus to an algebra of magnetic
translations.  For \(\theta\in\RR\), let \(A_\theta\) be the universal unital
\(C^*\)-algebra generated by unitaries \(\mathsf U\) and \(\mathsf V\) with
\[
                  \mathsf V\mathsf U
                    =\e^{2\pi\ii\theta}\mathsf U\mathsf V .
\]
Equivalently,
\(A_\theta\cong C(S^1)\rtimes_{\alpha_\theta}\ZZ\), where
\(\alpha_\theta\) is rotation through angle \(2\pi\theta\).  In a Euclidean
magnetic lattice, \(\theta\) is the flux per cell in flux-quantum units and
the two generators are magnetic translations.

For irrational \(\theta\), the Pimsner--Voiculescu exact sequence gives
\[
                    K_0(A_\theta)\cong\ZZ^2,
              \qquad K_1(A_\theta)\cong\ZZ^2
\]
\cite{PimsnerVoiculescu1980}.  These abstract groups are independent of
\(\theta\).  The canonical trace \(\tau_\theta\), defined on finite Fourier
sums by
\(\tau_\theta(\sum a_{mn}\mathsf U^m\mathsf V^n)=a_{00}\), retains more:
\[
             \tau_{\theta,*}\bigl(K_0(A_\theta)\bigr)
                    =\ZZ+\theta\ZZ .
\]
For projections \(p\in A_\theta\), the possible normalized trace values lie in
\((\ZZ+\theta\ZZ)\cap[0,1]\).  Rieffel's projections realize these values,
and the ordered, traced \(K_0\)-group consequently sees magnetic information
that the bare group \(\ZZ^2\) forgets \cite{RieffelRotation1981}.

The canonical derivations of the smooth noncommutative torus define a cyclic
two-cocycle.  Its pairing with \([p]\in K_0(A_\theta)\) gives the noncommutative
first Chern number and, under the usual physical hypotheses, the quantized
Hall conductance
\cite{ConnesNCG,BellissardVanElstSchulzBaldes}.  This is an exemplary case in
which \(K\)-theory, supplemented by a trace and a cyclic cocycle, captures
exactly the robust quantity it is meant to capture.

It still does not reconstruct a magnetic Hamiltonian
\(h=h^*\in M_N(A_\theta)\).  Even after \(A_\theta\), its smooth subalgebra,
the canonical trace, and the derivations have been fixed,
\Cref{prop:functional-calculus} produces \(f(h)\) with the same Fermi
projection and the same \(K\)-theoretic Hall data but a different spectrum.
Conversely, the trace range illustrates that algebra, trace, and \(K\)-class
are genuinely different layers: retaining only the abstract \(K\)-groups
forgets \(\theta\), while retaining the traced ordered group remembers more
but still not \(h\).

There is a categorical lesson as well.  Strong Morita equivalence identifies
the appropriate module categories and therefore preserves \(K\)-theory.
For irrational rotation algebras, \(A_\theta\) and \(A_{\theta'}\) are
strongly Morita equivalent precisely when \(\theta\) and \(\theta'\) lie in
the same \(\mathrm{GL}(2,\ZZ)\)-orbit under fractional linear transformations
\cite{RieffelRotation1981}.  Such an equivalence does not, by itself, select a
particular Hamiltonian, trace normalization, differential calculus, or
metric.  Morita-invariant classification and geometric reconstruction are
again different problems.

The hyperbolic magnetic analogue is not another copy of \(A_\theta\).
Projective surface-group translations satisfy
\[
          T_\gamma T_\delta
             =\sigma(\gamma,\delta)T_{\gamma\delta},
       \qquad
          \sigma\in Z^2(\Gamma_g,\UU(1)),
\]
and lead to the twisted reduced group algebra
\(C_r^*(\Gamma_g,\sigma)\).  The nonabelian surface relation and the
conjugacy classes entering \eqref{eq:walk-trace} remain essential; replacing
\(\Gamma_g\) by \(\ZZ^2\) would erase precisely that information.  The
continuous and discrete quantum Hall effects on the hyperbolic plane have
indeed been formulated through an imprimitivity algebra and a pairing of
\(K\)-theory with cyclic cohomology
\cite{CareyHannabussMathaiMcCann1998}.  Marcolli and Mathai make the
Euclidean--hyperbolic comparison explicit: for Euclidean magnetic
translations the twisted group algebra
\(C_r^*(\ZZ^2,\sigma)\) is the noncommutative torus \(A_\theta\), whereas a
hyperbolic magnetic model replaces \(\ZZ^2\) by a cocompact Fuchsian group
\cite{MarcolliMathai2006}.  Noncommutative tori are therefore a highly
relevant comparison and a model of successful noncommutative
classification, but they should not be identified with the hyperbolic
Brillouin or character space.

\subsection{What gap labels determine}

Let \(\mathcal A\) be a unital \(C^*\)-algebra with a trace \(\tau\), extended
to matrix algebras using the ordinary matrix trace.  A fixed normalization per
fundamental cell or per orbital may be imposed afterwards.  If
\(h=h^*\in M_N(\mathcal A)\) and \(E\notin\spec(h)\), continuous functional
calculus gives the Fermi projection
\[
             p_E(h):=\mathbf1_{(-\infty,E)}(h)\in M_N(\mathcal A).
             \tag{9.5}\label{eq:gap-fermi-projection}
\]
The trace induces
\[
     \tau_*:K_0(\mathcal A)\longrightarrow\RR,
     \qquad
     \mathcal N_h(E):=\tau\!\left(p_E(h)\right)
                     =\tau_*[p_E(h)].
     \tag{9.6}\label{eq:trace-pairing}
\]
In a covariant lattice model, \(\mathcal N_h(E)\) is the integrated density
of states at a gap, up to the chosen normalization.

\begin{proposition}[Scope and invariance of a gap label]
\label{prop:gap-label-scope}
With the notation above:
\begin{enumerate}[label=\textup{(\alph*)},leftmargin=2.5em]
\item \(\mathcal N_h(E)\) belongs to
  \(\tau_*(K_0(\mathcal A))\), and it is constant as \(E\) moves within a
  connected component of \(\RR\setminus\spec(h)\).

\item Let \(s\mapsto h_s\) be a norm-continuous path of self-adjoint elements,
  and let \(s\mapsto E_s\) be continuous.  If
  \[
       \inf_{s\in[0,1]}
       \operatorname{dist}\!\left(E_s,\spec(h_s)\right)>0,
       \tag{9.7}\label{eq:uniform-gap}
  \]
  then \(p_{E_s}(h_s)\) is a norm-continuous path of projections.  Its
  \(K_0\)-class and trace label are independent of \(s\).

\item The common class and trace label do not determine the position or width
  of the gap, nor the spectral measure on either side.  Indeed, the
  functional-calculus deformation \(h\mapsto f(h)\) of
  \Cref{prop:functional-calculus} preserves \(p_E(h)\) whenever \(f\)
  preserves the spectral cut at \(E\), while it can move both gap edges and
  reshape the spectrum.
\end{enumerate}
\end{proposition}

\begin{proof}
The projection in \eqref{eq:gap-fermi-projection} represents a class in
\(K_0(\mathcal A)\), which proves the first assertion in (a); moving \(E\)
inside one gap does not cross a spectral value and therefore does not change
the projection.  Under \eqref{eq:uniform-gap}, the Riesz projection can be
defined locally in \(s\) by a contour lying in the resolvent set.  Resolvent
continuity makes these local projections norm-continuous, and they agree on
overlaps.  A norm-continuous path is a homotopy of projections, proving (b).
Part (c) is precisely the functional-calculus mechanism of
\Cref{prop:functional-calculus}.
\end{proof}

For an aperiodic hull \(\Omega\) with an invariant probability measure
\(\mathfrak m\), the transformation-groupoid or crossed-product algebra
carries a trace \(\tau_{\mathfrak m}\) obtained by taking the coefficient at
the identity and averaging over \(\mathfrak m\).  Gap-labelling theorems
identify the allowed subgroup
\(\tau_{\mathfrak m,*}(K_0(\mathcal A))\).  This is a major success of
\(K\)-theory in periodic, quasiperiodic, and noncommutative settings
\cite{BellissardBovierGhez,KellendonkPutnam}.  Three points are important.

\begin{enumerate}[label=\textup{(\roman*)},leftmargin=2.6em]
\item The trace range is a group of \emph{allowed} labels.  A fixed
  Hamiltonian need not open a gap at every element of that group.
\item The label depends on the trace, hence on the invariant measure or
  normalization, not merely on the abstract groups \(K_*(\mathcal A)\).
\item If \(E\) lies in the spectrum, the discontinuous step function need not
  define a projection in \(\mathcal A\).  Mobility-gap formulations require
  additional localization and regularity hypotheses.
\end{enumerate}

A gap label therefore does not reconstruct the spectral measure between
gaps, decide localization, or give a gap's location or width.  Hyperbolic
periodic approximants illustrate both sides: \(K\)-theory predicts
topological gap classes, while converging boundary conditions and explicit
Hamiltonians are needed to decide which gaps are open and to calculate their
spectra \cite{LuxProdan2023}.  Because a hyperbolic lattice is nonamenable,
one must also justify how the trace is approximated by finite devices;
\Cref{prop:finite-cover-moments} supplies the elementary finite-range
statement used here.  These limitations are not shortcomings of gap
labelling.  Gap labelling determines the stable trace values associated with
open gaps; it is not designed to reconstruct the surrounding spectrum.

\subsection{Would differential \texorpdfstring{\(K\)}{K}-theory be enough?}

Differential \(K\)-theory does address one limitation of ordinary
\(K\)-theory: it refines a topological class by connection and
differential-form data.  On a smooth compact parameter manifold \(B\), it
comes with maps
\[
 I:\widehat K^0(B)\longrightarrow K^0(B),
 \qquad
 R:\widehat K^0(B)\longrightarrow
      \Omega_{\mathrm{cl}}^{\mathrm{even}}(B),
 \tag{9.8}\label{eq:differential-k-maps}
\]
where \(I\) forgets the differential refinement and \(R\) records its
curvature form.  They satisfy
\[
       [R(\widehat x)]_{\mathrm{dR}}
          =\ch\!\left(I(\widehat x)\right),
 \tag{9.9}\label{eq:differential-k-curvature}
\]
and one standard exact sequence is
\[
 K^{-1}(B)\xrightarrow{\ \ch\ }
 \frac{\Omega^{\mathrm{odd}}(B)}{\operatorname{im}\dd}
 \xrightarrow{\ a\ }\widehat K^0(B)
 \xrightarrow{\ I\ }K^0(B)\longrightarrow0.
 \tag{9.10}\label{eq:differential-k-sequence}
\]
In the Freed--Lott model, a cycle contains a Hermitian vector bundle, a
compatible connection, and an odd differential form modulo the appropriate
Chern--Simons relation \cite{FreedLott}.  Thus differential \(K\)-theory can
retain Chern--Weil curvature and flat or secondary holonomy information that
ordinary \(K^0(B)\) forgets.

For a gapped spectral projector \(P\), the occupied bundle and Berry
connection define a class
\[
 \widehat e_P
   :=[\mathscr E_P,h_P,\nabla^{\mathrm B},0]
      \in\widehat K^0(B).
 \tag{9.11}\label{eq:differential-band-class}
\]
This is a meaningful refinement, but it is still not a complete band datum.

\begin{proposition}[Differential \(K\)-theory does not determine band data]
\label{prop:differential-k-not-band}
\begin{enumerate}[label=\textup{(\alph*)},leftmargin=2.5em]
\item Fix the class \(\widehat e_P\).  It is compatible with infinitely
  many gapped Hamiltonians having different eigenvalue functions.
\item There are projector families with the same differential \(K\)-class
  and the same Berry connection but different quantum metrics.
\end{enumerate}
\end{proposition}

\begin{proof}
For (a), keep \(P\), \(h_P\), and \(\nabla^{\mathrm B}\) fixed and use the
Hamiltonians \(H_{\eps_-,\eps_+}\) in
\eqref{eq:fixed-projector-H}.  The functions \(\eps_-\) and \(\eps_+\) may
vary arbitrarily subject to the gap inequality.  Neither the occupied bundle
nor its connection changes, so \(\widehat e_P\) is fixed while the dispersion
changes.

For (b), use the real normalized sections \(u_a:S^1\to\CC^2\) from
\eqref{eq:metric-circle-family}.  They trivialize their spectral line bundles
globally and satisfy
\[
               \ii\braket{u_a}{\dd u_a}=0
               \qquad\text{for every }a.
 \tag{9.12}\label{eq:real-family-berry}
\]
Hence every induced Berry connection is the trivial connection and every
differential \(K\)-class \(\widehat e_{P_a}\) is the same.  Nevertheless,
\eqref{eq:circle-metric} gives
\[
               g_a^{\mathrm Q}=a^2\cos^2(k)\,\dd k^2,
 \]
which varies with \(a\).
\end{proof}

The proposition shows both what the refinement gains and what remains absent.
Differential \(K\)-theory can remember characteristic curvature forms and
secondary connection data, but it does not remember the eigenvalue functions,
the ambient embedding \(P:B\to\operatorname{Gr}(q,N)\), or the extrinsic
Fubini--Study metric pulled back by that embedding.  Nor does a class on \(B\)
by itself retain the hyperbolic metric and length spectrum of the real-space
surface, or a Higgs field and its spectral curve.

At the operator-algebraic level, a spectral triple
\((\mathcal A,\cH,D)\) carries substantially more metric information than
\(K_*(\mathcal A)\): the operator \(D\) supplies a differential calculus and,
in favourable commutative cases, a distance formula.  This reinforces rather
than evades the paper's conclusion, because the metric information resides in
\(D\), not in the \(K\)-groups alone.  Even a fixed spectral triple does not
select an arbitrary solid-state Hamiltonian \(h\in M_N(\mathcal A)\) unless
the model specifies a relation between \(h\) and \(D\).  The appropriate
repair is therefore to retain the geometric or spectral object from which a
\(K\)-class was extracted, with differential \(K\)-theory serving as a
valuable intermediate layer rather than a complete replacement.  The next
section tests this conclusion in finite and thermodynamic realizations, where
the choice of boundary condition and limiting trace becomes part of the
physical record.

\section{Quasicrystals, finite hyperbolic systems, and experiments}
\label{sec:quasicrystals-experiments}

\subsection{The quasicrystal comparison}

The conceptual parallel with quasicrystals is close but not exact.  For an
aperiodic pattern, one often replaces a Brillouin torus by a hull, a
transformation groupoid, and its \(C^*\)-algebra.  \(K\)-theory and trace
pairings label robust gaps, while the spectral measure and localization
depend on the pattern, frequencies, metric realization, and hopping law
\cite{BaakeGrimm,KellendonkPutnam}.  Two operators associated with the same
hull algebra can occupy the same \(K\)-class and still have different
spectra.

For example, let \(T\) be a finite-local-complexity pattern in \(\RR^d\), and
let \(\Omega_T\) be the closure of its translation orbit in the local
topology.  The transformation groupoid
\(\Omega_T\rtimes\RR^d\), or an equivalent transversal reduction, produces
the observable algebra \(\mathcal A_T\).  A finite-range covariant
Hamiltonian is a self-adjoint element of a matrix algebra over
\(\mathcal A_T\).  The hull specifies which local patterns may occur and how
they recur; it does not specify the hopping function assigned to those
patterns.  Consequently, changing hopping amplitudes, on-site terms, or
geometric weights can change the operator without changing \(\Omega_T\) or
\(\mathcal A_T\).  If a chosen Fermi gap remains open, the Fermi projection
stays in one \(K_0(\mathcal A_T)\)-class by
\Cref{prop:gap-label-scope}, although the bands or spectral measure on either
side of the gap may move.

The information loss may be displayed as
\[
 \bigl(\Omega_T,\mathcal A_T,h,\tau_{\mathfrak m}\bigr)
   \longrightarrow
 \bigl(\mathcal A_T,[p_E(h)],\tau_{\mathfrak m}\bigr)
   \longrightarrow
 \tau_{\mathfrak m,*}[p_E(h)].
 \tag{10.1}\label{eq:quasicrystal-forgetting}
\]
The first arrow forgets the particular Hamiltonian while retaining its
gapped projection; the second retains only the numerical gap label.  Each
step is useful, but neither is reversible.

Hyperbolic crystals share the failure of a two-dimensional Euclidean momentum
torus, but their replacement has a different origin.  A compact quotient
surface has a nonabelian fundamental group, its abelian characters form a
Jacobian, and higher-dimensional irreducible representations form
character-moduli spaces.  The geometry is not only the local arrangement of a
pattern.  It includes:
\[
\begin{gathered}
 \text{a hyperbolic metric, a complex structure, and a period matrix,}\\
 \text{together with moduli of flat or holomorphic bundles}.
\end{gathered}
\]
This makes hyperbolic matter a particularly transparent setting in which
topological and geometric data can be varied independently.

Thus a hull and a character variety should not be identified merely because
both replace an elementary Brillouin torus.  The hull organizes translates
and local configurations of an aperiodic pattern.  The character variety
organizes flat boundary conditions for the nonabelian fundamental group of a
compact quotient.  A genuinely aperiodic hyperbolic system may require both
structures, but the two sources of noncommutativity and the observables they
control remain distinct.

\subsection{Finite systems, boundary conditions, and the bulk trace}
\label{sec:finite-boundaries}

A disk-like patch of a hyperbolic lattice is not a bulk-dominated
approximation in the Euclidean sense.  For the standard regular hyperbolic
tilings, the universal-cover graph \(\widetilde{\mathcal Q}\) is
nonamenable: its edge-isoperimetric constant
\[
 h(\widetilde{\mathcal Q})
   :=\inf_{\substack{F\subset V(\widetilde{\mathcal Q})\\
                     0<|F|<\infty}}
       \frac{|\partial F|}{|F|}
   >0.
 \tag{10.2}\label{eq:hyperbolic-cheeger}
\]
Here \(\partial F\) is the set of edges leaving the finite vertex set \(F\).
Hence no exhaustion by open patches has a boundary-to-volume ratio tending
to zero.  Edge modes and boundary coordination can therefore contribute at
leading order to a normalized density of states.

One can retain this boundary deliberately rather than remove it.  Mathai and
Thiang use noncommutative T-duality to propose a fractional bulk--boundary
correspondence adapted to the geometry of the hyperbolic plane
\cite{MathaiThiang2019}.  That construction and the analysis below address
complementary questions: the former relates bulk and codimension-one boundary
indices, whereas the latter uses closed quotients to suppress edge
contributions and approximate the regular bulk trace.

Compact quotients and finite covers remove the physical boundary, but they
must also converge locally to the universal cover.  This condition has a
simple group-theoretic form.  Fix a word length \(|\cdot|\) on
\(\Gamma_g\), and let
\[
  \Gamma_n\triangleleft\Gamma_g,\qquad
  G_n:=\Gamma_g/\Gamma_n
\]
be finite-index normal subgroups and their deck groups.  Define the
group-theoretic injectivity length
\[
       \ell_n
        :=\min\{|\gamma|:\gamma\in\Gamma_n\setminus\{1\}\}.
 \tag{10.3}\label{eq:cover-injectivity-length}
\]
A sequence with \(\ell_n\to\infty\) contains no new relation of any fixed
word length for sufficiently large \(n\).  Geometrically, the corresponding
compact covers have injectivity radius tending to infinity on the scale
measured by the chosen lattice generators; cocompactness makes this word
metric comparable, up to uniform constants, with hyperbolic displacement.
Such sequences exist because surface groups are residually finite.  Indeed,
a nested normal chain with trivial intersection eventually excludes every
nonidentity element in each finite word ball, which is precisely
\(\ell_n\to\infty\).

Approximation of \(L^2\)-invariants by normalized invariants of finite
coverings goes back to L\"uck, while local or Benjamini--Schramm convergence
provides a broader spectral-measure formulation
\cite{Luck1994,AbertThomVirag2014}.  The result below is the elementary
finite-range specialization needed here.  Finite propagation strengthens
asymptotic moment convergence to an exact identity as soon as the quotient is
faithful on the word ball traversed by the relevant closed walks.  We include
the proof because it identifies that physical length scale directly.

After choosing a fundamental cell, any finite-range invariant Hamiltonian on
the universal cover can be written
\[
  \widetilde H
     =\sum_{|\gamma|\leq R}
        \lambda(\gamma)\otimes A_\gamma,
  \qquad
  A_{\gamma^{-1}}=A_\gamma^*,
 \tag{10.4}\label{eq:finite-range-group-operator}
\]
on \(\ell^2(\Gamma_g)\otimes\CC^N\).  Here \(N\) counts vertices and internal
orbitals in the cell, \(\lambda\) is the left regular representation, and
\(R\) is the hopping range in word length.  Let \(H_n\) be the finite matrix
obtained by replacing \(\lambda\) with the regular representation of \(G_n\).

\begin{proposition}[Exact finite-cover identity for fixed spectral moments]
\label{prop:finite-cover-moments}
Let \(\tau_{\Gamma_g}\) be the canonical group trace and
\(\tr_N=N^{-1}\Tr_N\).  If \(\ell_n>kR\), then
\[
   \frac{1}{N|G_n|}\Tr(H_n^k)
       =(\tau_{\Gamma_g}\otimes\tr_N)(\widetilde H^k).
 \tag{10.5}\label{eq:finite-cover-moment}
\]
Consequently, if \(\ell_n\to\infty\), the empirical spectral measures of
\(H_n\) converge weakly to the spectral measure of \(\widetilde H\) defined
by \(\tau_{\Gamma_g}\otimes\tr_N\).
\end{proposition}

\begin{proof}
Expand
\[
        \widetilde H^k
           =\sum_{|\gamma|\leq kR}
               \lambda(\gamma)\otimes B_\gamma^{(k)}.
\]
The infinite-volume trace extracts the coefficient of the identity:
\[
  (\tau_{\Gamma_g}\otimes\tr_N)(\widetilde H^k)
       =\tr_N B_1^{(k)}.
\]
On \(\ell^2(G_n)\), the normalized trace of left translation by the image of
\(\gamma\) is one if \(\gamma\in\Gamma_n\) and zero otherwise.  Therefore
\[
 \frac{1}{N|G_n|}\Tr(H_n^k)
     =\sum_{\substack{\gamma\in\Gamma_n\\|\gamma|\leq kR}}
        \tr_N B_\gamma^{(k)}.
\]
If \(\ell_n>kR\), only \(\gamma=1\) occurs, proving
\eqref{eq:finite-cover-moment}.  The norms of all \(H_n\) are bounded by
\(\sum_\gamma\|A_\gamma\|\).  Equality of every fixed moment for all
sufficiently large \(n\), followed by polynomial approximation on this
common compact spectral interval, gives weak convergence.
\end{proof}

\begin{corollary}[Quantitative trace approximation for continuous observables]
\label{cor:finite-cover-functional-calculus}
Set
\[
                 M:=\sum_{|\gamma|\leq R}\|A_\gamma\|.
\]
For \(f\in C([-M,M];\RR)\), let
\[
 E_d(f;[-M,M])
   :=\inf_{\substack{p\in\RR[x]\\ \deg p\leq d}}
          \|f-p\|_{\infty,[-M,M]}
 \tag{10.5a}\label{eq:best-polynomial-approximation}
\]
be its best uniform approximation error by real polynomials of degree at
most \(d\).  If \(\ell_n>dR\), then
\[
\left|
 \frac{1}{N|G_n|}\Tr\!\bigl(f(H_n)\bigr)
 -
 (\tau_{\Gamma_g}\otimes\tr_N)
      \bigl(f(\widetilde H)\bigr)
\right|
 \leq 2E_d(f;[-M,M]).
\tag{10.5b}\label{eq:continuous-trace-bound}
\]
In particular, the two normalized traces agree exactly for every polynomial
of degree at most \(d\).
\end{corollary}

\begin{proof}
Both \(\widetilde H\) and \(H_n\) have norm at most \(M\).  If
\(\deg p\leq d\), \Cref{prop:finite-cover-moments}, applied to each
positive-degree monomial (with the constant terms agreeing trivially), gives
\[
 \frac{1}{N|G_n|}\Tr\!\bigl(p(H_n)\bigr)
 =
 (\tau_{\Gamma_g}\otimes\tr_N)
      \bigl(p(\widetilde H)\bigr).
\]
The two normalized traces are states.  Hence, for every such \(p\),
the left-hand side of \eqref{eq:continuous-trace-bound} is at most
\[
 \|f(H_n)-p(H_n)\|
 +\|f(\widetilde H)-p(\widetilde H)\|
 \leq 2\|f-p\|_{\infty,[-M,M]}.
\]
Taking the infimum over \(p\) proves the bound.
\end{proof}

For arithmetic congruence covers, systolic geometry supplies the missing
relation between local injectivity and matrix size.

\begin{corollary}[Arithmetic congruence-tower rates]
\label{cor:arithmetic-congruence-rate}
Suppose that \(X=\HH^2/\Gamma_g\) is a compact arithmetic surface and that
\(\Gamma_n\triangleleft\Gamma_g\) is a principal congruence tower, with
\(G_n=\Gamma_g/\Gamma_n\) and \(|G_n|\to\infty\).  There are constants
\(a,b>0\), depending on \(X\) and the chosen word metric, such that
\[
                    \ell_n\geq a\log|G_n|-b.
\tag{10.5c}\label{eq:arithmetic-injectivity-growth}
\]
If \(R=0\), the two normalized traces agree for every continuous \(f\).
Assume \(R\geq1\) and, for all sufficiently large \(n\), set
\[
 d_n:=\left\lfloor
          \frac{a\log|G_n|-b-1}{R}
       \right\rfloor .
\tag{10.5d}\label{eq:arithmetic-polynomial-degree}
\]
Then every \(f\in C([-M,M];\RR)\) satisfies
\[
\left|
 \frac{1}{N|G_n|}\Tr\!\bigl(f(H_n)\bigr)
 -
 (\tau_{\Gamma_g}\otimes\tr_N)\!\bigl(f(\widetilde H)\bigr)
\right|
 \leq 2E_{d_n}(f;[-M,M]).
\tag{10.5e}\label{eq:arithmetic-trace-rate}
\]
If \(f\) extends holomorphically to a complex neighborhood of \([-M,M]\),
then the left-hand side is at most \(C|G_n|^{-\eta}\) for some
\(C,\eta>0\).  If \(s\geq1\) is an integer and
\(f\in C^s([-M,M])\), it is
\(O((\log|G_n|)^{-s})\).
\end{corollary}

\begin{proof}
If \(R=0\), then \(\widetilde H=I\otimes A_1\) and
\(H_n=I_{|G_n|}\otimes A_1\), which proves the first assertion.
The logarithmic systole theorem of Katz, Schaps, and Vishne gives
\[
 \operatorname{sys}_{\pi_1}(\HH^2/\Gamma_n)
      \geq \frac43\log g(\HH^2/\Gamma_n)-c
\]
for a constant \(c\) depending only on the arithmetic base surface
\cite{KatzSchapsVishne2007}.  Since
\[
 g(\HH^2/\Gamma_n)-1
       =|G_n|\bigl(g(X)-1\bigr),
\]
this is bounded below by \(c_0\log|G_n|-c_1\) for suitable
\(c_0,c_1>0\).  Fix a point \(o\in\HH^2\).  If \(A\) is the largest
displacement of \(o\) by a chosen finite set of generators, then
\[
 \operatorname{sys}_{\pi_1}(\HH^2/\Gamma_n)
 \leq \inf_{\gamma\in\Gamma_n\setminus\{1\}}d(o,\gamma o)
 \leq A\ell_n.
\]
After renaming constants, this proves
\eqref{eq:arithmetic-injectivity-growth}.  The definition of \(d_n\) gives
\(d_nR<a\log|G_n|-b\leq\ell_n\), so
\Cref{cor:finite-cover-functional-calculus} gives
\eqref{eq:arithmetic-trace-rate}.

For a function holomorphic in a complex neighborhood of the interval,
Bernstein approximation gives \(E_d\leq C_0\rho^{-d}\) for some
\(\rho>1\).  Since \(d_n\geq (a/R)\log|G_n|-C_1\), it follows that
\[
 E_{d_n}(f;[-M,M])
    \leq C_2|G_n|^{-a\log(\rho)/R}.
\]
For \(f\in C^s\), Jackson approximation gives
\(E_d=O(d^{-s})\), and hence the stated logarithmic rate.
\end{proof}

The first corollary turns local faithfulness into a quantitative error
estimate without imposing amenability; the arithmetic specialization shows
what additional systolic information buys.  Outside such quantitatively
controlled towers, there is no universal rate in the cover index \(|G_n|\)
without further information about the residual chain.  Neither bound by
itself controls a discontinuous Fermi projection or excludes a vanishing
fraction of spurious eigenvalues.

The finite quotient also has a precise sector interpretation.  The regular
representation decomposes into irreducibles of \(G_n\), and therefore
\[
 \frac{1}{N|G_n|}\Tr(H_n^k)
   =
 \sum_{\pi\in\widehat{G_n}}
    \frac{(\dim\pi)^2}{|G_n|}
    \left(
      \frac{1}{N\dim\pi}\Tr(H_\pi^k)
    \right),
 \tag{10.6}\label{eq:plancherel-sector-average}
\]
where \(H_\pi\) is the Bloch matrix obtained from
\eqref{eq:finite-range-group-operator} using \(\pi\).  A compact quotient
thus computes a Plancherel-weighted average over the representation sectors
factoring through \(G_n\).  It does not sample those sectors uniformly, and a
calculation in one abelian or nonabelian sector is not the same observable as
the regular finite-volume trace.  Different sequences of covers can have
different finite-size sector content even when they converge to the same
group-trace limit
\cite{MaciejkoRayan2022,LenggenhagerEtAl2023}.

Equation \eqref{eq:plancherel-sector-average} is a finite-group identity.  It
must not be promoted to a Peter--Weyl decomposition of the infinite-volume
operator into finite-dimensional irreducible representations of
\(\Gamma_g\).  The obstruction and the vanishing-weight statement can be
separated cleanly.  The first part of the next proposition is a classical
consequence of weak containment, Fell absorption, and Hulanicki's amenability
criterion \cite{Hulanicki1964}; its short proof is recorded to connect that
operator-algebraic fact directly with the finite-group weights in the second
part.

\begin{proposition}[No fixed finite-dimensional bulk sector]
\label{prop:no-finite-dimensional-bulk}
Let \(g\geq2\).
\begin{enumerate}[label=\textup{(\alph*)},leftmargin=2.5em]
\item No nonzero finite-dimensional unitary representation of \(\Gamma_g\)
  extends to a representation of \(C_r^*(\Gamma_g)\).  In particular, the
  regular representation of \(\Gamma_g\) has no nonzero finite-dimensional
  subrepresentation.
\item Let \(G_n=\Gamma_g/\Gamma_n\) be a locally faithful sequence as above.
  If a fixed \(d\)-dimensional irreducible representation factors through
  \(G_n\) for all sufficiently large \(n\), its individual atom in the
  finite-group Plancherel formula has weight
  \[
                         \frac{d^2}{|G_n|}\longrightarrow0.
  \]
\end{enumerate}
Neither assertion prevents Plancherel-weighted aggregates of representations
whose dimensions and multiplicities vary with \(n\) from converging to the
regular trace.
\end{proposition}

\begin{proof}
Let \(\lambda\) be the regular representation.  A unitary representation
\(\pi\) extends continuously to \(C_r^*(\Gamma_g)\) exactly when
\(\pi\) is weakly contained in \(\lambda\).  Suppose that such a nonzero
finite-dimensional \(\pi\) exists, and let \(\overline\pi\) be its conjugate.
Weak containment is preserved by tensoring, while Fell absorption gives
\[
       \lambda\otimes\overline\pi
           \cong \lambda^{\oplus\dim\pi}.
\]
The tensor product \(\pi\otimes\overline\pi\) contains the trivial
representation: under
\(\pi\otimes\overline\pi\cong\End(\CC^{\dim\pi})\), the identity operator is
invariant.  It would follow that the trivial representation is weakly
contained in \(\lambda\).  By Hulanicki's regular-representation criterion
for amenability, \(\Gamma_g\) would then be amenable
\cite{Hulanicki1964}.  This is impossible for a closed surface group of genus
\(g\geq2\), which is a non-elementary cocompact Fuchsian group and is
nonamenable.  This proves the first assertion.  A finite-dimensional
subrepresentation of \(\lambda\) would, in particular, extend to
\(C_r^*(\Gamma_g)\), so none exists.

For the second, local faithfulness \(\ell_n\to\infty\) forces
\(|G_n|\to\infty\): otherwise a subsequence would range over finite groups of
uniformly bounded order and could not be injective on arbitrarily large word
balls.  The displayed limit now follows from the exact Plancherel weight in
\eqref{eq:plancherel-sector-average}.
\end{proof}

A finite-dimensional unitary representation of \(\Gamma_g\) does extend to
the full group \(C^*\)-algebra, but by
\Cref{prop:no-finite-dimensional-bulk} it does not factor through the reduced
algebra in which the physical regular-representation spectrum lives.  Such a
sector can consequently contain spectral values absent from the bulk
spectrum \cite{LuxProdanCayley2024}.  Moreover,
\eqref{eq:plancherel-sector-average} has no canonical atomic Peter--Weyl
limit that assigns positive mass to any fixed finite-dimensional
\(\Gamma_g\)-sector.  Representations whose dimensions grow along the
quotient sequence can contribute collectively, but what converges
canonically is the normalized trace, or equivalently the spectral measure.
This is the operator-algebraic content of the coherent approximation
emphasized by Lux and Prodan
\cite{LuxProdanCayley2024,LuxProdan2023}.

The same conclusion appears from the hyperbolic-Bloch side.  Mosseri and
Vidal found that the rank-one Bloch-like spectrum carries a vanishing fraction
of the thermodynamic spectral weight \cite{MosseriVidal2023}.  Shankar and
Maciejko proved more generally that, for every fixed \(k\), the normalized
\(\UU(r)\) hyperbolic-Bloch moment \(m_k^{(r)}\) satisfies
\[
   (\tau_{\Gamma_g}\otimes\tr_N)(\widetilde H^k)
      =\lim_{r\to\infty}m_k^{(r)}.
 \tag{10.7}\label{eq:large-rank-hbt-moments}
\]
The average defining \(m_k^{(r)}\) uses the natural
Atiyah--Bott--Goldman, equivalently topological Yang--Mills, measure on the
rank-\(r\) character space.  Fixed \(r\) is therefore an approximation level,
not a superselection component of positive bulk weight
\cite{ShankarMaciejko2024}.

This averaged moment theorem is compatible with, but logically distinct
from, the asymptotic norm-recovery result for the hyperbolic Bloch transform
\cite{NagyRayan2024}.  The former identifies the limit of normalized traces
of powers of a Hamiltonian; the latter recovers inner products of
compactly supported states from a sequence of finite-rank transforms.
Together they support an infinite-rank reciprocal description, but neither
turns any fixed rank into a positive-weight component of the regular
representation.

\Cref{prop:finite-cover-moments} controls moments and weak convergence.
\Cref{cor:finite-cover-functional-calculus} adds a trace-error bound for
continuous functions, but neither result guarantees that a
proposed open interval is free of all finite-volume eigenvalues: a vanishing
fraction of spurious states is invisible to the limiting measure.  Detecting
a bulk gap requires uniform resolvent control or a converging
periodic-boundary construction adapted to the tiling, as emphasized in
\cite{LuxProdanCayley2024,LuxProdan2023}.  Moreover, a rate expressed in the
matrix size requires quantitative control of how \(\ell_n\) grows with
\(|G_n|\).  \Cref{cor:arithmetic-congruence-rate} supplies this control for
arithmetic principal congruence towers, but residual finiteness alone does
not.

This distinction affects the interpretation of every observable in the
paper.  Before comparing geometry with topology one must specify:

\begin{enumerate}[label=\textup{(\roman*)},leftmargin=2.6em]
\item whether the system is the infinite lattice, a compact quotient, or an
  open patch;
\item the subgroup chain, local injectivity scale, and boundary conditions;
\item whether the observable is sector-resolved, averaged with the
  finite-group Plancherel weights in
  \eqref{eq:plancherel-sector-average}, or defined thermodynamically by the
  canonical group trace;
\item whether couplings encode hyperbolic distances and angles or only graph
  adjacency.
\end{enumerate}

Once these choices are fixed, the nonfactorization results apply without
ambiguity.  Compactification controls the boundary problem; it does not
restore the geometric or spectral information discarded by a \(K\)-class.

\subsection{Six experimental tests}

The mathematical distinction can be tested with existing circuit and
synthetic-matter techniques.

\paragraph{Test 1: topology-preserving, geometry-varying couplings.}
Fabricate or program two systems with the same quotient graph and side
pairings.  In the first, choose couplings to approximate a prescribed
hyperbolic metric.  In the second, deform those couplings continuously while
keeping the graph, gap, and symmetry class fixed.  Compare resonance
frequencies and localization profiles.  The 2025 superconducting framework
of \cite{XuEtAl2025} is especially well suited to this test because its
capacitive couplings encode the metric.

\paragraph{Test 2: Wilson spectroscopy.}
Insert controllable phase twists
\(\bm\theta\in(\RR/2\pi\ZZ)^{2g}\) through the surface cycles and measure
\(\Tr(H_{\bm\theta}^n)\), or reconstruct it from resonances.  Fourier analysis
in \(\bm\theta\) recovers coefficients \(A_{n,\bm q}\) in
\eqref{eq:moment-fourier}.  A nonzero \(\bm q\) coefficient is a direct
witness that the spectrum varies across flat line bundles with the same
\(K^0(S)\)-class.

\paragraph{Test 3: quantum-metric tomography.}
Modulate two flux angles weakly and use integrated excitation rates to
estimate components of \(g^{\mathrm Q}\), following the periodic-driving
principle of \cite{OzawaGoldman}.  Adiabatic cycle measurements provide Berry
curvature.  The comparison tests \eqref{eq:det-bound} and
\eqref{eq:trace-bound}.

\paragraph{Test 4: holomorphicity relative to the Jacobian.}
Use the period matrix to choose complex pairs of flux controls.  On each
two-dimensional slice, test whether the trace bound is saturated.  Saturation
with a consistent orientation is evidence that the band map is
pseudoholomorphic relative to the chosen complex structure.  Repeating the
measurement after a geometry-changing but topology-preserving deformation
tests whether that holomorphic structure is dynamically selected.

\paragraph{Test 5: quantized versus partially filled Hall response.}
On a two-torus of boundary twists, prepare a fully occupied gapped band and
measure the integrated transverse response.  Then move the chemical potential
through the same band without changing its global \(K\)-class.  The first
measurement tests the integer in \eqref{eq:tknn-hall}; the second probes the
dispersion-weighted curvature in \eqref{eq:partial-hall}.  Two devices can
therefore agree in quantized Hall conductance and disagree away from the
quantized limit.

\paragraph{Test 6: nodal geometry at fixed charge.}
Tune a circuit model through a Dirac, Weyl-type, or higher-Clifford band
crossing and reconstruct its resonance cone.  Measure the node location,
tilt, transverse velocity matrix, and low-energy density of states.  A
gap-preserving deformation on a linking sphere leaves the local Chern charge
fixed while these quantities vary.  In a nonabelian hyperbolic semimetal, the
same protocol should be carried out across the representation sectors that
support the codimension-five nodal manifold
\cite{TummuruEtAl2024}.

Each test compares systems whose stable topological data can be held fixed.
The measured differences are therefore non-topological by construction and
arise from the retained geometric or Hamiltonian data, not from a residual
change of phase label.  These operational distinctions, together with the
finite-cover and regular-trace results, provide the ingredients for the
observable-dependent completion assembled in the final section.

\section{Observable-dependent geometric completions of topological band data}
\label{sec:completion}

\subsection{The information hierarchy}

The word ``hierarchy'' refers to factorization of information, not to a
ranking of mathematical importance.  The relevant notion can be stated
without choosing coordinates or a particular Hamiltonian model.

\begin{definition}[Information order and observable factorization]
\label{def:information-order}
Let \(\mathcal D\) and \(\mathcal D'\) be classes of physical records, taken
modulo the gauge equivalences appropriate to the problem.  A map
\(\pi:\mathcal D\to\mathcal D'\) is a \emph{forgetful map} if
\(\pi(D)\) is obtained from \(D\) by discarding specified structure.  We
write \(\mathcal D\succeq\mathcal D'\).  An observable
\(\mathcal O:\mathcal D\to Y\) is determined by \(\mathcal D'\) if there is a
map \(\overline{\mathcal O}:\mathcal D'\to Y\) such that
\[
                    \mathcal O=\overline{\mathcal O}\circ\pi .
\]
The loss is \emph{strict for \(\mathcal O\)} if two records have the same
image under \(\pi\) but different values of \(\mathcal O\).
\end{definition}

This definition is the formal content of the paper's nonfactorization
statements.  It also prevents a common ambiguity: saying that a \(K\)-class
``does not contain'' a quantity means that the quantity fails to descend
along a specified map to \(K\)-theory.

The set-level formulation is sufficient for all counterexamples in this
paper, but the natural domain is often a groupoid.  Let \(\mathbf D\) be the
groupoid whose objects are physical records and whose arrows are the chosen
gauge equivalences.  A structure-forgetting operation is then a functor
\[
                       F:\mathbf D\longrightarrow\mathbf D'.
\]
An observable with values in a category \(\mathbf Y\) descends through \(F\)
if there is a functor
\(\overline{\mathcal O}:\mathbf D'\to\mathbf Y\) and a natural isomorphism
\(\mathcal O\cong\overline{\mathcal O}\circ F\).  Ordinary numerical or
set-valued observables are recovered by taking \(\mathbf Y\) to be a discrete
groupoid.  Passing to connected components
\(\pi_0\mathbf D\) recovers \Cref{def:information-order}.  A functor can lose
information either by identifying nonisomorphic objects or by erasing
automorphisms; the latter is exactly what a coarse moduli space does at a
sector with nontrivial stabilizer.

For a smooth parameter region \(B\), let
\(\rho_B:B\to\cR_r(S)\) be the sector map introduced in
\Cref{def:geometric-band-datum}, and write \(H_x\) for the Hamiltonian in the
sector \(x\in B\).  Locally on a framed representation space, one may choose
a representative \(\rho_x\) of \(\rho_B(x)\).  The geometric input and the
derived band data may be abbreviated as
\begin{align*}
\mathfrak D_{\mathrm{geom}}
  &:=
  (S,\Gamma_g,J,h_J,\mathcal Q,\gamma,\mathsf t;\,
       B,\rho_B,H),\\
\mathfrak D_{\mathrm{band}}
  &:=
  \bigl(\{E_n\},P,\mathscr E_P,
        \nabla^{\mathrm B},F^{\mathrm B},g^{\mathrm Q}\bigr).
\end{align*}
Here \(J\) and \(h_J\) are related by uniformization rather than treated as
independent variables; \(\gamma=\{\gamma_e\}\) denotes the group labels of
\(\mathcal Q\); and \(\mathsf t\) denotes the on-site, hopping, and transport
coefficients.  The family \(H\) is included because distinct microscopic
inputs can yield unitarily equivalent Hamiltonians, while the converse
reconstruction is generally nonunique.

The parameter-space hierarchy is
\[
\begin{tikzcd}[row sep=1.65em]
\mathfrak D_{\mathrm{geom}}
  \arrow[d, "\text{spectral analysis}"]\\
\mathfrak D_{\mathrm{band}}
  \arrow[d, "\text{forget energies, embedding, and metric}"]\\
(\mathscr E_P,\nabla^{\mathrm B})
  \arrow[d, "\text{forget connection}"]\\
{\mathscr E_P}
  \arrow[d, "\text{stable class}"]\\
{[\mathscr E_P]\in K^0(B)}.
\end{tikzcd}
\tag{11.1}\label{eq:information-hierarchy}
\]
The first arrow discards the real-space realization and retains the selected
spectral data.  The second forgets the eigenvalue functions, the embedding
\(P:B\to\operatorname{Gr}(q,N)\), and the quantum metric, but retains the
abstract spectral bundle with its Berry connection.  The third retains only
the topological isomorphism class of that bundle.  The last passes to stable
equivalence.  Curvature \(F^{\mathrm B}\) is written in
\(\mathfrak D_{\mathrm{band}}\) although it is determined by
\(\nabla^{\mathrm B}\), because it is the directly integrated quantity in
Chern and Hall pairings.

There is a distinct sectorwise hierarchy over position space:
\[
\begin{tikzcd}[row sep=1.4em]
(S,\Gamma_g,J,h_J,\mathcal Q,\gamma,\mathsf t;\rho)
  \arrow[d, "\text{construct}"]\\
(E_\rho,\nabla_\rho,H_\rho)
  \arrow[d, "\text{forget }H_\rho"]\\
(E_\rho,\nabla_\rho)
  \arrow[d, "\text{forget }\nabla_\rho"]\\
{E_\rho}
  \arrow[d, "\text{stable class}"]\\
{[E_\rho]\in K^0(S)}.
\end{tikzcd}
\tag{11.2}\label{eq:real-space-hierarchy}
\]
The bundle \(E_\rho\to S\) in \eqref{eq:real-space-hierarchy} and the spectral
bundle \(\mathscr E_P\to B\) in \eqref{eq:information-hierarchy} have
different bases and answer different questions.  The former records a Bloch
sector as a flat bundle over physical position space.  The latter records a
family of eigenspaces over momentum or representation space.

Categorically, the central portion of
\eqref{eq:information-hierarchy} is the shadow on isomorphism classes of
functors
\[
 \mathsf{BandGeom}(B)
   \longrightarrow
 \mathsf{HermVect}^{\nabla}(B)
   \longrightarrow
 \mathsf{Vect}_{\CC}(B).
\]
Here the first groupoid retains an isolated spectral cluster together with
its energies, projector, and quantum geometry; the second retains a Hermitian
bundle with Berry connection; and the third retains only the underlying
complex vector bundle.  The final passage to \(K^0(B)\) is slightly
different: for compact \(B\), it takes the commutative monoid
\(\pi_0\mathsf{Vect}_{\CC}(B)\), under direct sum, to its Grothendieck group
\cite{AtiyahKTheory,KaroubiKTheory}.  Thus ``forget the connection'' and
``stabilize and group-complete'' should not be treated as the same categorical
operation.  The sectorwise hierarchy has the same form with the groupoid of
unitary local systems over \(S\) at its left.

The location of strictness depends on the base and on which equivalences have
already been imposed.  In \eqref{eq:real-space-hierarchy}, all degree-zero
\(E_\rho\) of fixed rank are smoothly isomorphic, so forgetting
\(\nabla_\rho\) is already a drastic collapse.  The final map
\(E_\rho\mapsto[E_\rho]\) causes no additional loss at fixed rank on a
surface, as explained after \eqref{eq:forgetful-flat-holo-k}.  In
\eqref{eq:information-hierarchy}, the final stabilization map may introduce
additional ambiguity when \(B\) has higher dimension, but the earlier arrows
are already strict for dispersion and differential-geometric observables.
Thus strictness depends on the base, the equivalences already imposed, and the
observable under consideration.  The theorem below records the losses
established by the preceding examples.

\begin{theorem}[Strictness of the hierarchy]
\label{thm:strict-hierarchy}
Within hyperbolic band theory, each of the following losses is strict in the
sense of \Cref{def:information-order}:

\begin{enumerate}[label=\textup{(\alph*)},leftmargin=2.5em]
\item \((E_\rho,\nabla_\rho)\mapsto E_\rho\): the representation and flat
  connection are not determined by the smooth bundle, and hence not by
  \([E_\rho]\in K^0(S)\);
\item \((E_\rho,\nabla_\rho,H_\rho)\mapsto[E_\rho]\): the spectrum is not
  determined by the flat-bundle \(K\)-class;
\item \((J,h_J)\mapsto(S,K^*(S))\): the hyperbolic metric and complex
  structure are not determined by the topology or its \(K\)-groups;
\item \((\{E_n\},P)\mapsto P\): dispersion is not determined by an occupied
  or isolated-band projector;
\item \((P,\nabla^{\mathrm B},g^{\mathrm Q})\mapsto[\mathscr E_P]\): Berry
  holonomy, curvature profile, and quantum metric are not determined by the
  spectral-bundle \(K\)-class;
\item \((E,\Phi)\mapsto[E]\): a Higgs field and its spectral curve are not
  determined by the underlying bundle \(K\)-class;
\item the position, embedding, tilt, and transverse velocity tensor of a
  topologically protected nodal manifold are not determined by its local
  \(K\)-theoretic charge.
\end{enumerate}
\end{theorem}

\begin{proof}
For \textup{(a)}, the positive-dimensional spaces \(\cR_r(S)\) contain
inequivalent flat connections, while \Cref{thm:k-collapse} and the
classification of smooth bundles on a surface identify all of their
underlying degree-zero bundles at fixed rank.  For \textup{(b)}, the
one-loop model in \Cref{ex:one-loop} gives an explicit interval of spectral
values over that single \(K\)-class; the general obstruction is
\Cref{thm:spectral-nonfactorization}.  Part \textup{(c)} is witnessed by the
pinching family in \Cref{thm:metric-nonfactorization}, and the stronger local
reconstruction statement is \Cref{thm:hodge-hessian}.

For \textup{(d)}, the Hamiltonians in
\eqref{eq:fixed-projector-H} have the same projector and arbitrary gapped
energy functions.  The circle families and curvature-redistribution
construction in \Cref{thm:qgt-not-k} prove \textup{(e)}.  Part \textup{(f)}
is \Cref{prop:higgs-collapse}, with the scaling family
\eqref{eq:higgs-scaling} as an explicit witness.  Finally,
\Cref{prop:clifford-nodes} proves \textup{(g)} by preserving the
linking-sphere degree while varying the differential and embedded geometry
of the node.
\end{proof}

The theorem does not identify a single universally preferred object at the
top of the hierarchy.  Rather, the data that must be retained depend on the
observable.  A quantized Chern pairing may factor through
\([\mathscr E_P]\), while a density of states, an optical matrix element, or
a metric-reconstruction experiment does not.  The next subsection makes this
dependence explicit.

\subsection{What should be retained}

There is no single minimal completion for every physical question.  A Chern
pairing, a density of states, and an optical transition rate factor through
different quotients of the same Hamiltonian data.  It is nevertheless useful
to specify a baseline record from which the principal single-particle
observables in this paper can be recovered.

\begin{definition}[Geometry-complete gapped band record]
\label{def:geometry-complete}
Let \(B\) be a smooth parameter region, let
\(\rho_B:B\to\cR_r(S)\) be a sector map, and let
\[
                  H:B\longrightarrow\operatorname{Herm}(N)
\]
be a smooth finite Hamiltonian family.  Choose a rank-\(q\) spectral cluster
with \(0<q<N\), separated from the rest of the spectrum.  Denote by
\(\boldsymbol E_{\mathcal I}(x)\) its unordered eigenvalue multiset, counted
with multiplicity, and by
\(\boldsymbol E_{\mathcal I^c}(x)\) the complementary eigenvalue multiset.
Define
\[
 \delta(x)
   :=\min_{\substack{
        E\in\boldsymbol E_{\mathcal I}(x)\\
        E'\in\boldsymbol E_{\mathcal I^c}(x)}}
       |E-E'|>0
\]
its gap to the complementary spectrum.  A
\emph{geometry-complete gapped band record} is
\[
\begin{aligned}
\mathfrak B_{\mathrm{gap}}
=\bigl(&
 \underbrace{S,\Gamma_g,J,h_J,\mathcal Q,\gamma,\mathsf t}
             _{\text{real-space realization}};\;
 \underbrace{B,\rho_B,\dd\mu_B}
             _{\text{sector parametrization and sampling}};\notag\\[-1mm]
&\underbrace{H,\mathcal I,\delta,
             \boldsymbol E_{\mathcal I},P}
             _{\text{spectral data}};\;
 \underbrace{\mathscr E_P,\nabla^{\mathrm B},
             F^{\mathrm B},g^{\mathrm Q}}
             _{\text{band geometry}}
\bigr).
\end{aligned}
\tag{11.3}\label{eq:complete-record}
\]
Here \(\mathcal I\) labels the chosen cluster,
\(P\) is its Riesz projector, and \(\dd\mu_B\) is included whenever sectors
are averaged to form a density of states or response.  In a purely
sectorwise question the measure may be omitted.
\end{definition}

The record in \eqref{eq:complete-record} is a finite-dimensional
control-family or sector-ensemble record.  When
\(B\subset\cR_r(S)\) has fixed finite rank, \(\dd\mu_B\) describes how that
family is scanned experimentally or sampled computationally.  It is not a
thermodynamic Plancherel measure for the nonamenable group \(\Gamma_g\).  For
the exact infinite-lattice density of states, the corresponding record must
instead retain
\[
  \widetilde H\in M_N(C_r^*(\Gamma_g)),\qquad
  \lambda:\Gamma_g\to\mathcal U(\ell^2\Gamma_g),\qquad
  \tau_{\Gamma_g},
\]
and, when an approximation is used, the coherent finite-cover chain or
large-rank \(\UU(r)\) sequence.  The bulk spectral measure
\(\nu_{\mathrm{bulk}}\) is characterized by
\[
 \int_{\RR}f(E)\,\dd\nu_{\mathrm{bulk}}(E)
   =(\tau_{\Gamma_g}\otimes\tr_N)\bigl(f(\widetilde H)\bigr)
\]
for continuous \(f\).  This regular-trace record and the finite-rank record
answer different questions, even when the latter converges to the former.

The use of the unordered multiset
\(\boldsymbol E_{\mathcal I}\), rather than globally labelled functions
\(E_n\), is deliberate.  Individual eigenvalue labels may permute around
internal degeneracies even though the cluster projector \(P\) remains
smooth.  On a local chart without such monodromy one may write
\(\boldsymbol E_{\mathcal I}=\{E_n\}_{n\in\mathcal I}\), recovering the
notation used in the response formulae.  Likewise, \(F^{\mathrm B}\) is the
Hermitian Berry curvature defined in
\eqref{eq:berry-connection-curvature}, while \(g^{\mathrm Q}\) is the
Grassmannian metric pulled back by \(P\).

The record in \eqref{eq:complete-record} is intentionally redundant.
\(H\) determines its spectral projectors, \(P\) determines
\(\mathscr E_P\) as a subbundle of \(B\times\CC^N\), and the ambient
Hermitian structure determines \(\nabla^{\mathrm B}\) and
\(g^{\mathrm Q}\).  Retaining the derived data makes the relevant
forgetful maps and experimentally accessible quantities explicit.  More
importantly, the left-hand portion cannot be reconstructed from the
right-hand portion: the same band geometry can have different real-space
realizations, and the same \(K\)-class can arise from different Hamiltonians,
sector measures, and hyperbolic metrics.

For a continuum model, \(H\) in \eqref{eq:complete-record} is replaced by a
smooth family of self-adjoint elliptic operators with a common domain, or by
an equivalent norm-resolvent formulation; \(\mathcal Q,\gamma,\mathsf t\) are
then replaced by the kinetic operator, potential, and geometric connection.
For a Higgs-coupled model, \((E,\Phi)\), the coupling that inserts \(\Phi\)
into \(H\), and the relevant spectral-curve data must be included.  Merely
recording a point of Higgs moduli without its physical coupling does not make
the holomorphic data observable.

The equivalence relation on records also has to be stated.  Smooth unitary
conjugation of \(H_x\), representation conjugation, and changes of quotient
vertex lifts are gauge equivalences; they transport \(P\),
\(\nabla^{\mathrm B}\), and the remaining tensors in the usual way.  A
mapping-class transformation is different.  It may be quotiented out in an
unmarked problem, but it should be retained in flux tomography because it
relabels the \(2g\) measured cycle directions.  The marking is precisely what
allows Hessians, period matrices, and Wilson coefficients from two devices
to be compared componentwise.  These flux coordinates provide a homology
marking, not a full Teichm\"uller marking: mapping classes in the Torelli
group act trivially on \(H_1(S;\ZZ)\) and are therefore invisible to this
comparison \cite{FarbMargalit}.

The amount of the record actually needed is observable-dependent:
\begingroup
\renewcommand{\arraystretch}{1.12}
\begin{center}
\small
\begin{tabular}{@{}
 >{\raggedright\arraybackslash}p{0.22\textwidth}
 >{\raggedright\arraybackslash}p{0.30\textwidth}
 >{\raggedright\arraybackslash}p{0.38\textwidth}@{}}
\toprule
\textbf{question} & \textbf{sufficient layer} & \textbf{additional datum that must be fixed}\\
\midrule
quantized Hall pairing
 & \([\mathscr E_P]\in K^0(B)\)
 & oriented two-cycle, charge convention, and a filled gapped cluster\\
density of states or Fermi geometry in a controlled sector family
 & \(\boldsymbol E_{\mathcal I}\) and \(\dd\mu_B\)
 & chemical potential, temperature, and sector-sampling prescription\\
thermodynamic bulk density of states
 & \(\widetilde H\in M_N(C_r^*(\Gamma_g))\) and
   \(\tau_{\Gamma_g}\)
 & coherent finite-cover or large-rank approximation scheme, if one is used\\
partially filled Hall response
 & \(\boldsymbol E_{\mathcal I}\), \(F^{\mathrm B}\), and
   \(\dd\mu_B\)
 & filling function, response directions, and a controlled sector family\\
adiabatic transport or localization
 & \(P\), \(\nabla^{\mathrm B}\), and \(g^{\mathrm Q}\)
 & path, control metric, and real-space localization convention\\
Wilson or closed-walk tomography
 & \(H\) over the sector scan
 & marking, group labels, and separation of hopping from holonomy\\
complex-structure reconstruction
 & Hessian of the bottom twisted band
 & homology-marked intersection lattice and the curvature-\(-1\) area normalization\\
\bottomrule
\end{tabular}
\end{center}
\endgroup

The table makes this observable dependence explicit.  Given a chosen family
\(\mathfrak O\) of observables, define
\[
 D_1\sim_{\mathfrak O}D_2
   \quad\Longleftrightarrow\quad
 \mathcal O(D_1)=\mathcal O(D_2)
 \ \text{for every }\mathcal O\in\mathfrak O .
 \tag{11.4}\label{eq:observable-equivalence}
\]
The quotient by \(\sim_{\mathfrak O}\) is the formally minimal information
space for that family.  It need not be represented by a familiar geometric
object, and different observable families yield incomparable quotients.
Equation \eqref{eq:complete-record} is therefore a practical common
refinement.  Categorical minimality belongs instead to the
observable-dependent quotient in \eqref{eq:observable-equivalence}.

For a gapless system, the record must be modified rather than merely
truncated.  At chemical potential \(\mu_{\mathrm F}\), one should retain
\[
\begin{aligned}
\mathfrak B_{\mathrm{gapless}}
=\bigl(&B,\dd\mu_B,H,\mu_{\mathrm F};
   \mathcal F_{\mu_{\mathrm F}},\mathcal Z;\,
   P|_{B\setminus\mathcal Z};\notag\\[-1mm]
 &\{[P|_{S_x^\perp}]\}_{x\in\mathcal Z},
   \dd_\perp H,\text{ higher normal jets as required}\bigr).
\end{aligned}
\tag{11.5}\label{eq:gapless-record}
\]
Here
\(\mathcal F_{\mu_{\mathrm F}}
  =\{x\in B:\mu_{\mathrm F}\in\spec(H_x)\}\)
is the Fermi level set, \(\mathcal Z\) is the band-degeneracy locus,
\(S_x^\perp\) is a small sphere in the normal directions linking
\(\mathcal Z\) at \(x\), and \(\dd_\perp H\) is the normal linearization
that supplies the velocity tensor.  The local class on \(S_x^\perp\)
records the nodal charge; it does not replace the embedding of
\(\mathcal Z\), the Fermi level set, or the normal derivative.  If a regular
Fermi surface rather than a nodal degeneracy is the object of interest, its
induced measure and \(|\nabla E_n|\) are retained as in
\eqref{eq:fermi-coarea}.  As in the gapped case, this is a finite-rank sector
record.  A thermodynamic gapless record must instead use the spectral measure
and projections of the regular operator, together with a specified large-rank
or finite-cover limit if one wishes to relate bulk scaling to finite-rank
nodal geometry.

Finally, the word ``complete'' remains relative to single-particle band
geometry.  It does not include every disorder realization, bath coupling, or
many-body state.  In an interacting system one must retain at least the
relevant Green function or self-energy, and generally the vertex functions
or current correlators required by the chosen response.  There is therefore
no geometry-complete description of a strange metal consisting only of a
noninteracting occupied bundle.  In such a regime even a single-particle
Hamiltonian \(H\), not only its \(K\)-class, is insufficient without
frequency-dependent dynamical data.

\subsection{Open mathematical problems}

The perspective suggests several problems that are both geometric and
physical.

\begin{enumerate}[label=\arabic*.,leftmargin=2.0em]
\item \textbf{Character-variety quantum geometry.}
  Determine the Berry and quantum-metric tensors of nonabelian Bloch bands on
  smooth character-variety strata.  Relate them to the Goldman symplectic form
  and to the K\"ahler metrics obtained through
  Narasimhan--Seshadri theory.

\item \textbf{Holomorphic band maps.}
  Classify finite hyperbolic Bloch Hamiltonians whose isolated-band projector
  is holomorphic relative to \(\Jac(X)\), or pseudoholomorphic on natural
  two-cycles.  Establish finite-range obstructions and approximation
  theorems analogous to those for K\"ahler Chern bands
  \cite{MeraOzawaFlat}.

\item \textbf{Metric reconstruction.}
  Identify which combinations of closed-walk moments or twisted heat traces
  recover geodesic-length and hopping data.  The correct question is not
  whether the spectrum determines the metric uniquely, but which geometric
  features are stable and experimentally resolvable.

\item \textbf{Families over moduli.}
  Construct the occupied-state object globally over representation and Higgs
  moduli, accounting for stabilizers and the absence of an ordinary universal
  bundle on coarse spaces.  Compare equivariant, stack-theoretic, and
  diffeological formulations.

\item \textbf{Noncommutative differential refinement.}
  Develop an invariant that combines gap topology with Wilson data, spectral
  moments, and quantum geometry for \(C_r^*(\Gamma_g,\sigma)\).  The
  noncommutative-torus comparison suggests retaining a trace and cyclic
  cocycles, while the present results require additional metric and
  Hamiltonian data.  A spectral triple or a differential refinement of a
  crossed-product cycle is a more plausible starting point than \(K\)-groups
  alone.

\item \textbf{Large-rank thermodynamic Bloch limit.}
  The starting point must be the target of the limiting process.  For a
  finite-range Hamiltonian
  \(\widetilde H\in M_N(C_r^*(\Gamma_g))\), the bulk density of states is
  defined by the canonical state
  \(\tau_{\Gamma_g}\otimes\tr_N\).  It is not obtained from an atomic
  Peter--Weyl measure on the finite-dimensional unitary dual.  Indeed, no
  nonzero finite-dimensional representation extends to
  \(C_r^*(\Gamma_g)\), and the finite-group Plancherel weight of any fixed
  \(d\)-dimensional sector, whenever it occurs in \(G_n\), is
  \(d^2/|G_n|\to0\); see
  \Cref{prop:no-finite-dimensional-bulk}.  The rank-one vanishing observed
  by Mosseri and Vidal is the first case of this thermodynamic distinction
  \cite{MosseriVidal2023,LuxProdanCayley2024}.

  Two finite-dimensional approximation mechanisms are nevertheless
  available, and they should not be conflated.  Along a locally faithful
  finite-quotient chain, every fixed polynomial trace is eventually
  \emph{exact}, with the continuous-functional-calculus error controlled by
  \eqref{eq:continuous-trace-bound}; for arithmetic principal congruence
  towers, \Cref{cor:arithmetic-congruence-rate} converts this into explicit
  quotient-size rates.  Separately, Shankar and Maciejko prove that the
  normalized \(\UU(r)\) hyperbolic-Bloch moment, averaged with the
  Atiyah--Bott--Goldman or topological Yang--Mills measure, converges for
  each fixed moment order as \(r\to\infty\)
  \cite{ShankarMaciejko2024}.  The asymptotic norm-recovery theorem of Nagy
  and the author for the hyperbolic Bloch transform is a third, compatible
  statement about states rather than density-of-states moments
  \cite{NagyRayan2024}.

  The open problem is therefore quantitative and structural, not merely the
  existence of a vague ``large-representation limit.''  Beyond the arithmetic
  congruence case, determine the relation between cover injectivity length,
  quotient size, and errors for resolvents, Fermi projections, and response
  functions; compute the
  large-rank corrections and the admissible joint limits in moment order
  and rank; identify which Berry, quantum-metric, Higgs, or
  character-variety data have regular-trace descendants; and decide whether
  finite-cover Plancherel averages, Yang--Mills averages, asymptotic Bloch
  transforms, and supercell constructions admit a common comparison framework
  \cite{LuxProdan2023,LenggenhagerEtAl2023}.

\item \textbf{Nodal geometry on character spaces.}
  Develop transversality and local-index theorems for Dirac- and
  Clifford-type nodes on smooth character-variety strata, then study how the
  nodal sets meet reducible or singular strata.  Determine which signatures
  of finite-rank nodal geometry persist under the large-rank or
  regular-representation limit and contribute to the bulk scaling, since no
  fixed finite-rank character space has positive thermodynamic spectral
  weight.

\item \textbf{Interacting hyperbolic metals.}
  Starting from a specified Hubbard, gauge-coupled, or other many-body model,
  determine how hyperbolic density of states, frustration, connectivity, and
  quantum geometry enter the self-energy, vertex functions, and transport.
  A proposed strange-metal or Planckian regime must be tested through
  quasiparticle lifetimes, spectral functions, and the relevant temperature
  and frequency scalings.

\end{enumerate}

\subsection{Conclusion}

The explanatory and theorem-level strands meet here.  Hyperbolic band theory
is moduli-theoretic and differential-geometric before it is \(K\)-theoretic:
rank-one sectors are flat line bundles organized by a polarized Jacobian,
nonabelian sectors are flat vector bundles organized by stratified character
spaces, and Higgs fields introduce spectral curves.  Holomorphy is
kinematical when it organizes these spaces, but dynamical when the Hodge star,
kinetic operator, metric-dependent couplings, transport, Higgs field, or
spectral projector responds to the complex structure.  It then controls
dispersion, heat traces, effective mass, Berry transport, quantum metric, and
pseudoholomorphic saturation.  The passage from geometric organization to
physical dependence is the explanatory centre of the manuscript.

The mathematical argument gives this explanation a precise form.  Its first
principal component is a common observable-factorization framework across
fibrewise \(K^0(S)\), occupied-family \(K^0(B)\), and operator-algebraic
\(K\)-theory, together with explicit strictness results for Wilson-weighted
spectra, Higgs data, band geometry, partially filled response, and Fermi and
nodal geometry.  Its second is a sector--bulk synthesis that places the
classical weak-containment obstruction, exact finite-propagation identities,
the continuous-functional-calculus estimate, its arithmetic congruence-tower
rates, and the large-rank thermodynamic results in one hyperbolic-band
framework.  Together they
identify the datum that is forgotten, exhibit an observable that still
depends on it, and specify when that dependence enters sector-resolved or
bulk physics.

\(K\)-theory nevertheless sees exact and indispensable parts of the
structure.  It records stable bundle classes, quantized Chern pairings, gap
labels, and robust nodal charges.  These are positive results for topology,
not exceptions to the argument.  Beyond this stable information, however, the
same class does not determine a Wilson-weighted spectrum, a gap width, Berry
holonomy, curvature distribution, quantum metric, Fermi geometry, or the
position, velocity, tilt, and embedding of a protected node.  In interacting
regimes such as strange metals, where quasiparticle bands may cease to be the
correct organizing objects, still more frequency-dependent dynamical and
correlation data are required.

The sector--bulk distinction is equally essential to the conclusion.
Finite-rank geometric structures govern sector-resolved observables and
finite periodic devices, but they do not form a Peter--Weyl decomposition of
the infinite-lattice spectrum.  The thermodynamic density of states is the
spectral measure of the regular representation: each fixed
finite-dimensional sector has zero bulk weight, while coherent finite-cover
and large-rank hyperbolic-Bloch sequences can recover its moments, with
explicit quotient-size rates for arithmetic congruence covers.  Thus the
geometric content of controlled sectors and the spectral content of the bulk
are related by specified limits, not by an atomic measure on finite-rank
representations.

Earlier results fit naturally into this synthesis.  The Kotani--Sunada
Hessian formula and the Hodge--Torelli reconstruction mechanism show exactly
how geometric information can reappear in ground-state spectral data.  The
recovered period data determine the underlying surface but not a full
Teichm\"uller marking.  Thermodynamic Bloch theory, meanwhile, supplies the
large-rank density-of-states limit.  This article brings those mechanisms
together with the discrete, holomorphic, quantum-geometric, finite-cover, and
operator-algebraic nonfactorization results developed here.  The result is a
unified account of which physical quantities stable topology determines and
which continue to depend on geometry, holomorphy, or the choice of
thermodynamic limit.

The appropriate conclusion is therefore constructive:
\[
\boxed{\begin{gathered}
\text{\(K\)-theory classifies what remains stable;}\\[-1mm]
\text{geometry determines much of what remains physical.}
\end{gathered}}
\]
Hyperbolic matter makes both halves of this statement unusually visible.

\appendix

\section{The \texorpdfstring{\(K\)}{K}-groups of a closed oriented surface}
\label{app:k-surface}

For completeness, we recall the computation used in
\Cref{thm:k-collapse}; standard foundations and conventions for complex
\(K\)-theory may be found in \cite{AtiyahKTheory,KaroubiKTheory}.  Give \(S\)
its standard CW structure with one zero-cell, \(2g\) one-cells, and one
two-cell attached by the word in \eqref{eq:surface-group}.  The
Atiyah--Hirzebruch spectral sequence for complex \(K\)-theory has
\[
             E_2^{p,q}=H^p(S;K^q(\mathrm{pt})).
             \tag{A.1}\label{eq:AHSS}
\]
Complex Bott periodicity gives \(K^{2k}(\mathrm{pt})=\ZZ\) and
\(K^{2k+1}(\mathrm{pt})=0\).  Because \(S\) has dimension two, there are no
possible nonzero differentials affecting total degrees zero or one.
Therefore the associated graded groups are
\begin{align}
\operatorname{gr}K^0(S)
   &\cong H^0(S;\ZZ)\oplus H^2(S;\ZZ)
     \cong\ZZ\oplus\ZZ, \notag\\
\operatorname{gr}K^1(S)
   &\cong H^1(S;\ZZ)
     \cong\ZZ^{2g}.
     \tag{A.2}\label{eq:graded-k}
\end{align}
All groups involved are free abelian, so the extension splits as a group.
The \(H^0\)-component is rank, and the \(H^2\)-component is the first Chern
class.  This proves \eqref{eq:k-surface}.

One may see the reduced statement directly.  Stabilized complex bundles are
classified by homotopy classes of maps to \(BU\).  On a two-complex, only the
two-stage Postnikov truncation matters:
\[
                       BU[2]\simeq K(\ZZ,2).
                       \tag{A.3}\label{eq:BU-two-stage}
\]
The resulting bijection
\[
             [S,BU]\cong H^2(S;\ZZ)
             \tag{A.4}\label{eq:BU-H2}
\]
is induced by \(c_1\).  In particular, a rank-\(r\) complex bundle of degree
zero is stably trivial.  Since complex vector bundles of fixed rank over a
surface are themselves classified by \(c_1\), it is also smoothly isomorphic
to the trivial rank-\(r\) bundle.  Its flat or holomorphic structure can
nevertheless be nontrivial, which is the distinction exploited in the main
text.

\section{Explicit genus-two models}
\label{app:genus-two}

\subsection{A scalar four-flux band}

For \(g=2\),
\[
 \Gamma_2=
 \langle a_1,b_1,a_2,b_2\mid[a_1,b_1][a_2,b_2]=1\rangle.
 \tag{B.1}\label{eq:genus-two-group}
\]
A character is determined by four angles
\[
 \chi(a_j)=\e^{\ii\theta_{a_j}},
 \qquad
 \chi(b_j)=\e^{\ii\theta_{b_j}}.
 \tag{B.2}\label{eq:four-fluxes}
\]
Consider a one-site quotient with loops carrying these four labels and one
additional loop labelled \(a_1a_2\).  With scalar hoppings
\(t_{a_j},t_{b_j},s\), its single energy is
\begin{align}
E(\bm\theta)
={}&v
+2\sum_{j=1}^2
  \left(
    |t_{a_j}|\cos(\theta_{a_j}+\phi_{a_j})
   +|t_{b_j}|\cos(\theta_{b_j}+\phi_{b_j})
  \right)\notag\\
&+2|s|\cos(\theta_{a_1}+\theta_{a_2}+\phi_s),
\tag{B.3}\label{eq:genus-two-scalar}
\end{align}
where the \(\phi\)'s are hopping phases.  Its gradient is a vector of
cycle-current responses.  For example,
\[
 \frac{\partial E}{\partial\theta_{a_1}}
 =-2|t_{a_1}|\sin(\theta_{a_1}+\phi_{a_1})
  -2|s|\sin(\theta_{a_1}+\theta_{a_2}+\phi_s).
 \tag{B.4}\label{eq:genus-two-velocity}
\]
All four-flux points represent line bundles with class
\([\cO_X]\in K^0(X)\).  The band can have arbitrary nonzero width.  Because
the Hilbert space is one-dimensional, its projector is constant and its
quantum metric vanishes.  This example cleanly separates spectral variation
from eigenstate geometry: it proves spectral nonfactorization without using
projector variation.  Nontrivial quantum geometry requires at least two
orbitals, as in the next example.

\subsection{A two-orbital band with quantum geometry}

Let
\[
        H(\bm\theta)=d_0(\bm\theta)I_2+
                     \bm d(\bm\theta)\cdot\bm\sigma,
        \qquad \bm d(\bm\theta)\neq0,
        \tag{B.5}\label{eq:two-band}
\]
where \(\bm\sigma=(\sigma_1,\sigma_2,\sigma_3)\) is the vector of Pauli
matrices and every component of \(\bm d\) is a finite Fourier polynomial in
the four fluxes, as in \eqref{eq:moment-fourier}.  The energies are
\[
                  E_\pm=d_0\pm|\bm d|.
                  \tag{B.6}\label{eq:two-band-energies}
\]
Writing \(\widehat{\bm d}=\bm d/|\bm d|\), the lower-band quantum geometry is
\begin{align}
 g_{\mu\nu}^{\mathrm Q}
   &=\frac14\,
     \partial_\mu\widehat{\bm d}\cdot
     \partial_\nu\widehat{\bm d},\notag\\
 F_{\mu\nu}
   &=\frac12\,
     \widehat{\bm d}\cdot
     \left(
       \partial_\mu\widehat{\bm d}\times
       \partial_\nu\widehat{\bm d}
     \right),
     \tag{B.7}\label{eq:two-band-qgt}
\end{align}
with the curvature sign reversed for the upper band.  A gap-preserving
deformation of \(\bm d\) leaves the \(K\)-class fixed but can change the
tensors in \eqref{eq:two-band-qgt}.  Precomposing
\(\widehat{\bm d}\) with a
diffeomorphism of the four-torus isotopic to the identity gives an explicit
curvature-redistribution family of the type used in
\Cref{thm:qgt-not-k}.

On a two-dimensional flux slice \(C\subset B_{\mathrm{ab}}\), the map
\(\widehat{\bm d}|_C:C\to S^2\cong\CC\mathrm P^1\) is
pseudoholomorphic exactly when the horizontal derivatives satisfy
\eqref{eq:horizontal-CR} with a consistent choice of sign.  The equality in
\eqref{eq:holomorphic-saturation} can therefore be checked directly from the
Fourier coefficients of a genus-two Bloch matrix.

\section*{Acknowledgements}

The author acknowledges Junaid Aftab, Johanna Erdmenger, Cameron Krulewski,
Joseph Maciejko, Ahmed Adel Mahmoud Taha Hassan, Matilde Marcolli, Ren\'e
Meyer, Riccardo Sorbello, and Ronny Thomale for insightful discussions on
these topics. The author extends a special gratitude
to Joseph Maciejko for a careful reading and for pointing out relevant results
in the hyperbolic setting regarding the thermodynamic Bloch limit. Crucial
steps in completing this work occurred during the Workshop on Quasicrystals
in Fundamental Physics, a programme of the International Centre for
Mathematical Sciences (ICMS) that was held during July 20--24, 2026 at the
Edinburgh Futures Institute (EFI).
The author is grateful to the organizers---Latham Boyle, Johanna Erdmenger,
and Justin Kulp---as well as the ICMS and EFI for facilitating a stimulating,
interdisciplinary environment.

The author acknowledges partial support from a Natural Sciences and
Engineering Research Council of Canada (NSERC) Discovery Grant.

\begingroup
\small
\bibliographystyle{unsrtnat}
\bibliography{BeyondKTheory}

\begin{thebibliography}{89}
\providecommand{\natexlab}[1]{#1}
\providecommand{\url}[1]{\texttt{#1}}
\expandafter\ifx\csname urlstyle\endcsname\relax
  \providecommand{\doi}[1]{doi: #1}\else
  \providecommand{\doi}{doi: \begingroup \urlstyle{rm}\Url}\fi

\bibitem[Kitaev(2009)]{KitaevPeriodicTable}
Alexei Kitaev.
\newblock Periodic table for topological insulators and superconductors.
\newblock In Vladimir Lebedev and Mikhail Feigel'man, editors, \emph{Advances
  in Theoretical Physics: Landau Memorial Conference}, volume 1134 of \emph{AIP
  Conference Proceedings}, pages 22--30. American Institute of Physics,
  Melville, NY, 2009.
\newblock \doi{10.1063/1.3149495}.

\bibitem[Prodan and Schulz-Baldes(2016)]{ProdanSchulzBaldes}
Emil Prodan and Hermann Schulz-Baldes.
\newblock \emph{Bulk and Boundary Invariants for Complex Topological
  Insulators: From {K}-Theory to Physics}.
\newblock Mathematical Physics Studies. Springer, Cham, 2016.
\newblock \doi{10.1007/978-3-319-29351-6}.

\bibitem[Narasimhan and Seshadri(1965)]{NarasimhanSeshadri}
M.~S. Narasimhan and C.~S. Seshadri.
\newblock Stable and unitary vector bundles on a compact {Riemann} surface.
\newblock \emph{Ann. of Math. (2)}, 82\penalty0 (3):\penalty0 540--567, 1965.
\newblock \doi{10.2307/1970710}.

\bibitem[Donaldson(1983)]{DonaldsonNS}
S.~K. Donaldson.
\newblock A new proof of a theorem of {Narasimhan} and {Seshadri}.
\newblock \emph{J. Differential Geom.}, 18\penalty0 (2):\penalty0 269--277,
  1983.
\newblock \doi{10.4310/jdg/1214437664}.

\bibitem[Kotani and Sunada(2000)]{KotaniSunada2000}
Motoko Kotani and Toshikazu Sunada.
\newblock Albanese maps and off diagonal long time asymptotics for the heat
  kernel.
\newblock \emph{Comm. Math. Phys.}, 209\penalty0 (3):\penalty0 633--670, 2000.
\newblock \doi{10.1007/s002200050033}.

\bibitem[Farb and Margalit(2012)]{FarbMargalit}
Benson Farb and Dan Margalit.
\newblock \emph{A Primer on Mapping Class Groups}, volume~49 of \emph{Princeton
  Mathematical Series}.
\newblock Princeton University Press, Princeton, 2012.
\newblock \doi{10.1515/9781400839049}.

\bibitem[Colbois et~al.(2025)Colbois, Provenzano, and
  Savo]{ColboisProvenzanoSavo2025}
Bruno Colbois, Luigi Provenzano, and Alessandro Savo.
\newblock Magnetic ground states and the conformal class of a surface.
\newblock \href{https://arxiv.org/abs/2503.16940} {arXiv:2503.16940}, 2025.
\newblock Preprint.

\bibitem[Maciejko and Rayan(2021)]{MaciejkoRayan2021}
Joseph Maciejko and Steven Rayan.
\newblock Hyperbolic band theory.
\newblock \emph{Sci. Adv.}, 7\penalty0 (36):\penalty0 eabe9170, 2021.
\newblock \doi{10.1126/sciadv.abe9170}.

\bibitem[Boettcher et~al.(2022)Boettcher, Gorshkov, Koll\'ar, Maciejko, Rayan,
  and Thomale]{BoettcherEtAl2022}
Igor Boettcher, Alexey~V. Gorshkov, Alicia~J. Koll\'ar, Joseph Maciejko, Steven
  Rayan, and Ronny Thomale.
\newblock Crystallography of hyperbolic lattices.
\newblock \emph{Phys. Rev. B}, 105\penalty0 (12):\penalty0 125118, 2022.
\newblock \doi{10.1103/PhysRevB.105.125118}.

\bibitem[Maciejko and Rayan(2022)]{MaciejkoRayan2022}
Joseph Maciejko and Steven Rayan.
\newblock Automorphic {Bloch} theorems for hyperbolic lattices.
\newblock \emph{Proc. Natl. Acad. Sci. USA}, 119\penalty0 (9):\penalty0
  e2116869119, 2022.
\newblock \doi{10.1073/pnas.2116869119}.

\bibitem[Kienzle and Rayan(2022)]{KienzleRayan2022}
Elliot Kienzle and Steven Rayan.
\newblock Hyperbolic band theory through {Higgs} bundles.
\newblock \emph{Adv. Math.}, 409:\penalty0 108664, 2022.
\newblock \doi{10.1016/j.aim.2022.108664}.

\bibitem[Nagy and Rayan(2024)]{NagyRayan2024}
\'Akos Nagy and Steven Rayan.
\newblock On the hyperbolic {Bloch} transform.
\newblock \emph{Ann. Henri Poincar\'e}, 25\penalty0 (3):\penalty0 1713--1732,
  2024.
\newblock \doi{10.1007/s00023-023-01336-8}.

\bibitem[Koc\'abov\'a and \v{S}\v{t}ov\'i\v{c}ek(2008)]{KocabovaStovicek2008}
Pavla Koc\'abov\'a and Pavel \v{S}\v{t}ov\'i\v{c}ek.
\newblock Generalized {Bloch} analysis and propagators on {Riemannian}
  manifolds with a discrete symmetry.
\newblock \emph{J. Math. Phys.}, 49\penalty0 (3):\penalty0 033518, 2008.
\newblock \doi{10.1063/1.2898484}.

\bibitem[Gruber(2001)]{Gruber2001}
Michael~J. Gruber.
\newblock Noncommutative {Bloch} theory.
\newblock \emph{J. Math. Phys.}, 42\penalty0 (6):\penalty0 2438--2465, 2001.
\newblock \doi{10.1063/1.1369122}.

\bibitem[Marcolli and Mathai(1999)]{MarcolliMathai1999}
Matilde Marcolli and Varghese Mathai.
\newblock Twisted index theory on good orbifolds, {I}: noncommutative {Bloch}
  theory.
\newblock \emph{Commun. Contemp. Math.}, 1\penalty0 (4):\penalty0 553--587,
  1999.
\newblock \doi{10.1142/S0219199799000213}.

\bibitem[Marcolli and Mathai(2001)]{MarcolliMathai2001}
Matilde Marcolli and Varghese Mathai.
\newblock Twisted higher index theory on good orbifolds, {II}: fractional
  quantum numbers.
\newblock \emph{Comm. Math. Phys.}, 217\penalty0 (1):\penalty0 55--87, 2001.
\newblock \doi{10.1007/s002200000351}.

\bibitem[Marcolli and Mathai(2006)]{MarcolliMathai2006}
Matilde Marcolli and Varghese Mathai.
\newblock Towards the fractional quantum {Hall} effect: a noncommutative
  geometry perspective.
\newblock In Caterina Consani and Matilde Marcolli, editors,
  \emph{Noncommutative Geometry and Number Theory}, volume E37 of \emph{Aspects
  of Mathematics}, pages 235--261. Vieweg, Wiesbaden, 2006.
\newblock \doi{10.1007/978-3-8348-0352-8_12}.

\bibitem[Rayan(2026)]{RayanQuantumMatter2026}
Steven Rayan.
\newblock Quantum matter from algebraic geometry and number theory.
\newblock In Yang-Hui He, Mo-Lin Ge, Cheng-Ming Bai, Jiakang Bao, and Edward
  Hirst, editors, \emph{Nankai Symposium on Mathematical Dialogues}, pages
  289--292. Springer Nature Singapore, Singapore, 2026.
\newblock \doi{10.1007/978-981-19-2328-9_33}.
\newblock CIM-DIALOGUES 2021.

\bibitem[Mosseri and Vidal(2023)]{MosseriVidal2023}
R.~Mosseri and J.~Vidal.
\newblock Density of states of tight-binding models in the hyperbolic plane.
\newblock \emph{Phys. Rev. B}, 108\penalty0 (3):\penalty0 035154, 2023.
\newblock \doi{10.1103/PhysRevB.108.035154}.

\bibitem[Lux and Prodan(2024)]{LuxProdanCayley2024}
Fabian~R. Lux and Emil Prodan.
\newblock Spectral and combinatorial aspects of {Cayley}-crystals.
\newblock \emph{Ann. Henri Poincar\'e}, 25\penalty0 (8):\penalty0 3563--3602,
  2024.
\newblock \doi{10.1007/s00023-023-01373-3}.

\bibitem[Lux and Prodan(2023)]{LuxProdan2023}
Fabian~R. Lux and Emil Prodan.
\newblock Converging periodic boundary conditions and detection of topological
  gaps on regular hyperbolic tessellations.
\newblock \emph{Phys. Rev. Lett.}, 131\penalty0 (17):\penalty0 176603, 2023.
\newblock \doi{10.1103/PhysRevLett.131.176603}.

\bibitem[Shankar and Maciejko(2024)]{ShankarMaciejko2024}
G.~Shankar and Joseph Maciejko.
\newblock Hyperbolic lattices and two-dimensional {Yang--Mills} theory.
\newblock \emph{Phys. Rev. Lett.}, 133\penalty0 (14):\penalty0 146601, 2024.
\newblock \doi{10.1103/PhysRevLett.133.146601}.

\bibitem[Lenggenhager et~al.(2023)Lenggenhager, Maciejko, and
  Bzdu\v{s}ek]{LenggenhagerEtAl2023}
Patrick~M. Lenggenhager, Joseph Maciejko, and Tom\'a\v{s} Bzdu\v{s}ek.
\newblock {Non-Abelian} hyperbolic band theory from supercells.
\newblock \emph{Phys. Rev. Lett.}, 131\penalty0 (22):\penalty0 226401, 2023.
\newblock \doi{10.1103/PhysRevLett.131.226401}.

\bibitem[Ikeda et~al.(2021)Ikeda, Aoki, and Matsuki]{IkedaAokiMatsuki}
Kazuki Ikeda, Shoto Aoki, and Yoshiyuki Matsuki.
\newblock Hyperbolic band theory under magnetic field and {Dirac} cones on a
  higher genus surface.
\newblock \emph{J. Phys.: Condens. Matter}, 33\penalty0 (48):\penalty0 485602,
  2021.
\newblock \doi{10.1088/1361-648X/ac24c4}.

\bibitem[Urwyler et~al.(2022)Urwyler, Lenggenhager, Boettcher, Thomale,
  Neupert, and Bzdu\v{s}ek]{UrwylerEtAl2022}
David~M. Urwyler, Patrick~M. Lenggenhager, Igor Boettcher, Ronny Thomale, Titus
  Neupert, and Tom\'a\v{s} Bzdu\v{s}ek.
\newblock Hyperbolic topological band insulators.
\newblock \emph{Phys. Rev. Lett.}, 129\penalty0 (24):\penalty0 246402, 2022.
\newblock \doi{10.1103/PhysRevLett.129.246402}.

\bibitem[Tummuru et~al.(2024)Tummuru, Chen, Lenggenhager, Neupert, Maciejko,
  and Bzdu\v{s}ek]{TummuruEtAl2024}
Tarun Tummuru, Anffany Chen, Patrick~M. Lenggenhager, Titus Neupert, Joseph
  Maciejko, and Tom\'a\v{s} Bzdu\v{s}ek.
\newblock Hyperbolic non-{Abelian} semimetal.
\newblock \emph{Phys. Rev. Lett.}, 132\penalty0 (20):\penalty0 206601, 2024.
\newblock \doi{10.1103/PhysRevLett.132.206601}.

\bibitem[Koll\'ar et~al.(2019)Koll\'ar, Fitzpatrick, and
  Houck]{KollarFitzpatrickHouck}
Alicia~J. Koll\'ar, Mattias Fitzpatrick, and Andrew~A. Houck.
\newblock Hyperbolic lattices in circuit quantum electrodynamics.
\newblock \emph{Nature}, 571:\penalty0 45--50, 2019.
\newblock \doi{10.1038/s41586-019-1348-3}.

\bibitem[Yu et~al.(2020)Yu, Piao, and Park]{YuPiaoPark}
Sunkyu Yu, Xianji Piao, and Namkyoo Park.
\newblock Topological hyperbolic lattices.
\newblock \emph{Phys. Rev. Lett.}, 125\penalty0 (5):\penalty0 053901, 2020.
\newblock \doi{10.1103/PhysRevLett.125.053901}.

\bibitem[Zhang et~al.(2022)Zhang, Yuan, Sun, Sun, and Zhang]{ZhangEtAl2022}
Weixuan Zhang, Hao Yuan, Na~Sun, Houjun Sun, and Xiangdong Zhang.
\newblock Observation of novel topological states in hyperbolic lattices.
\newblock \emph{Nat. Commun.}, 13:\penalty0 2937, 2022.
\newblock \doi{10.1038/s41467-022-30631-x}.

\bibitem[Xu et~al.(2025)Xu, Mahmoud, Gorgichuk, Thomale, Rayan, and
  Mariantoni]{XuEtAl2025}
Xicheng Xu, Ahmed~Adel Mahmoud, Noah Gorgichuk, Ronny Thomale, Steven Rayan,
  and Matteo Mariantoni.
\newblock A scalable superconducting circuit framework for emulating physics in
  hyperbolic space.
\newblock \href{https://arxiv.org/abs/2510.23827} {arXiv:2510.23827}, 2025.
\newblock Preprint.

\bibitem[L\"uck(1994)]{Luck1994}
Wolfgang L\"uck.
\newblock Approximating {$L^2$}-invariants by their finite-dimensional
  analogues.
\newblock \emph{Geom. Funct. Anal.}, 4\penalty0 (4):\penalty0 455--481, 1994.
\newblock \doi{10.1007/BF01896404}.

\bibitem[Ab\'ert et~al.(2014)Ab\'ert, Thom, and Vir\'ag]{AbertThomVirag2014}
Mikl\'os Ab\'ert, Andreas Thom, and B\'alint Vir\'ag.
\newblock {Benjamini--Schramm} convergence and pointwise convergence of the
  spectral measure.
\newblock Preprint, 2014.
\newblock URL \url{https://www.renyi.hu/~abert/luckapprox.pdf}.

\bibitem[Hulanicki(1964)]{Hulanicki1964}
Andrzej Hulanicki.
\newblock Groups whose regular representation weakly contains all unitary
  representations.
\newblock \emph{Studia Math.}, 24\penalty0 (1):\penalty0 37--59, 1964.
\newblock \doi{10.4064/sm-24-1-27-59}.

\bibitem[Katz et~al.(2007)Katz, Schaps, and Vishne]{KatzSchapsVishne2007}
Mikhail~G. Katz, Mary Schaps, and Uzi Vishne.
\newblock Logarithmic growth of systole of arithmetic {Riemann} surfaces along
  congruence subgroups.
\newblock \emph{J. Differential Geom.}, 76\penalty0 (3):\penalty0 399--422,
  2007.
\newblock \doi{10.4310/jdg/1180135693}.

\bibitem[Roy(2014)]{Roy2014}
Rahul Roy.
\newblock Band geometry of fractional topological insulators.
\newblock \emph{Phys. Rev. B}, 90\penalty0 (16):\penalty0 165139, 2014.
\newblock \doi{10.1103/PhysRevB.90.165139}.

\bibitem[Mera and Ozawa(2021{\natexlab{a}})]{MeraOzawaKahler}
Bruno Mera and Tomoki Ozawa.
\newblock {K}\"ahler geometry and {Chern} insulators: Relations between
  topology and the quantum metric.
\newblock \emph{Phys. Rev. B}, 104\penalty0 (4):\penalty0 045104,
  2021{\natexlab{a}}.
\newblock \doi{10.1103/PhysRevB.104.045104}.

\bibitem[Mera and Ozawa(2021{\natexlab{b}})]{MeraOzawaFlat}
Bruno Mera and Tomoki Ozawa.
\newblock Engineering geometrically flat {Chern} bands with {Fubini--Study}
  {K}\"ahler structure.
\newblock \emph{Phys. Rev. B}, 104\penalty0 (11):\penalty0 115160,
  2021{\natexlab{b}}.
\newblock \doi{10.1103/PhysRevB.104.115160}.

\bibitem[Farkas and Kra(1992)]{FarkasKra}
Hershel~M. Farkas and Irwin Kra.
\newblock \emph{Riemann Surfaces}, volume~71 of \emph{Graduate Texts in
  Mathematics}.
\newblock Springer, New York, second edition, 1992.
\newblock \doi{10.1007/978-1-4612-2034-3}.

\bibitem[Birkenhake and Lange(2004)]{BirkenhakeLange}
Christina Birkenhake and Herbert Lange.
\newblock \emph{Complex Abelian Varieties}, volume 302 of \emph{Grundlehren der
  mathematischen Wissenschaften}.
\newblock Springer, Berlin, second edition, 2004.
\newblock \doi{10.1007/978-3-662-06307-1}.

\bibitem[Atiyah and Bott(1983)]{AtiyahBott}
Michael~F. Atiyah and Raoul Bott.
\newblock The {Yang--Mills} equations over {Riemann} surfaces.
\newblock \emph{Philos. Trans. Roy. Soc. London Ser. A}, 308\penalty0
  (1505):\penalty0 523--615, 1983.
\newblock \doi{10.1098/rsta.1983.0017}.

\bibitem[Goldman(1984)]{GoldmanSymplectic}
William~M. Goldman.
\newblock The symplectic nature of fundamental groups of surfaces.
\newblock \emph{Adv. Math.}, 54\penalty0 (2):\penalty0 200--225, 1984.
\newblock \doi{10.1016/0001-8708(84)90040-9}.

\bibitem[Atiyah and Singer(1971)]{AtiyahSingerFamilies}
Michael~F. Atiyah and Isadore~M. Singer.
\newblock The index of elliptic operators. {IV}.
\newblock \emph{Ann. of Math. (2)}, 93\penalty0 (1):\penalty0 119--138, 1971.
\newblock \doi{10.2307/1970756}.

\bibitem[Selberg(1956)]{Selberg}
Atle Selberg.
\newblock Harmonic analysis and discontinuous groups in weakly symmetric
  {Riemannian} spaces with applications to {Dirichlet} series.
\newblock \emph{J. Indian Math. Soc. (N.S.)}, 20:\penalty0 47--87, 1956.

\bibitem[Buser(1992)]{Buser}
Peter Buser.
\newblock \emph{Geometry and Spectra of Compact Riemann Surfaces}, volume 106
  of \emph{Progress in Mathematics}.
\newblock Birkh\"auser, Boston, 1992.
\newblock \doi{10.1007/978-0-8176-4992-0}.

\bibitem[Attar and Boettcher(2022)]{AttarBoettcher2022}
Alexander Attar and Igor Boettcher.
\newblock {Selberg} trace formula in hyperbolic band theory.
\newblock \emph{Phys. Rev. E}, 106\penalty0 (3):\penalty0 034114, 2022.
\newblock \doi{10.1103/PhysRevE.106.034114}.

\bibitem[R\"uckriemen(2013)]{RueckriemenBloch}
Ralf R\"uckriemen.
\newblock Recovering quantum graphs from their {Bloch} spectrum.
\newblock \emph{Ann. Inst. Fourier (Grenoble)}, 63\penalty0 (3):\penalty0
  1149--1176, 2013.
\newblock \doi{10.5802/aif.2786}.

\bibitem[Garc\'ia-Raboso and Rayan(2015)]{GarciaRabosoRayan}
Alberto Garc\'ia-Raboso and Steven Rayan.
\newblock Introduction to nonabelian {Hodge} theory: Flat connections, {Higgs}
  bundles and complex variations of {Hodge} structure.
\newblock In Radu Laza, Matthias Sch\"utt, and Noriko Yui, editors,
  \emph{Calabi--Yau Varieties: Arithmetic, Geometry and Physics}, volume~34 of
  \emph{Fields Institute Monographs}, pages 131--171. Springer, New York, 2015.
\newblock \doi{10.1007/978-1-4939-2830-9_5}.

\bibitem[Hitchin(1987)]{HitchinSelfDuality}
Nigel~J. Hitchin.
\newblock The self-duality equations on a {Riemann} surface.
\newblock \emph{Proc. London Math. Soc. (3)}, 55\penalty0 (1):\penalty0
  59--126, 1987.
\newblock \doi{10.1112/plms/s3-55.1.59}.

\bibitem[Simpson(1992)]{SimpsonHiggsLocal}
Carlos~T. Simpson.
\newblock {Higgs} bundles and local systems.
\newblock \emph{Inst. Hautes \'Etudes Sci. Publ. Math.}, 75:\penalty0 5--95,
  1992.
\newblock \doi{10.1007/BF02699491}.

\bibitem[Beauville et~al.(1989)Beauville, Narasimhan, and Ramanan]{BNR}
Arnaud Beauville, M.~S. Narasimhan, and S.~Ramanan.
\newblock Spectral curves and the generalised theta divisor.
\newblock \emph{J. Reine Angew. Math.}, 398:\penalty0 169--179, 1989.
\newblock \doi{10.1515/crll.1989.398.169}.

\bibitem[Ramanan(1973)]{Ramanan1973}
S.~Ramanan.
\newblock The moduli spaces of vector bundles over an algebraic curve.
\newblock \emph{Math. Ann.}, 200:\penalty0 69--84, 1973.
\newblock \doi{10.1007/BF01578292}.

\bibitem[Biswas and Sengupta(2018)]{BiswasSengupta2018}
Indranil Biswas and Tathagata Sengupta.
\newblock Brauer group of the moduli spaces of stable vector bundles of fixed
  determinant over a smooth curve.
\newblock \emph{Bull. Sci. Math.}, 144:\penalty0 55--63, 2018.
\newblock \doi{10.1016/j.bulsci.2018.02.001}.

\bibitem[Azam and Rayan(2026{\natexlab{a}})]{AzamRayanQuiver2026}
Mahmud Azam and Steven Rayan.
\newblock Moduli stacks of quiver bundles with applications to {Higgs} bundles.
\newblock \emph{J. Algebra}, 708:\penalty0 333--379, 2026{\natexlab{a}}.
\newblock \doi{10.1016/j.jalgebra.2026.05.018}.

\bibitem[Azam and Rayan(2026{\natexlab{b}})]{AzamRayan2026}
Mahmud Azam and Steven Rayan.
\newblock A diffeological perspective on non-{Abelian} {Hodge} theory.
\newblock \href{https://arxiv.org/abs/2606.16772} {arXiv:2606.16772},
  2026{\natexlab{b}}.
\newblock Preprint.

\bibitem[Ikeda and Rayan(2026)]{IkedaRayan2026}
Kazuki Ikeda and Steven Rayan.
\newblock Quantum entanglement, stratified spaces, and topological matter:
  Towards entanglement-sensitive {Langlands} data.
\newblock \emph{Rep. Prog. Phys.}, 89\penalty0 (6):\penalty0 067601, 2026.
\newblock \doi{10.1088/1361-6633/ae73b6}.

\bibitem[Kato(1995)]{Kato}
Tosio Kato.
\newblock \emph{Perturbation Theory for Linear Operators}.
\newblock Classics in Mathematics. Springer, Berlin, second edition, 1995.
\newblock \doi{10.1007/978-3-642-66282-9}.

\bibitem[Provost and Vall\'ee(1980)]{ProvostVallee}
J.~P. Provost and G.~Vall\'ee.
\newblock {Riemannian} structure on manifolds of quantum states.
\newblock \emph{Comm. Math. Phys.}, 76\penalty0 (3):\penalty0 289--301, 1980.
\newblock \doi{10.1007/BF02193559}.

\bibitem[Berry(1984)]{Berry}
M.~V. Berry.
\newblock Quantal phase factors accompanying adiabatic changes.
\newblock \emph{Proc. Roy. Soc. London Ser. A}, 392\penalty0 (1802):\penalty0
  45--57, 1984.
\newblock \doi{10.1098/rspa.1984.0023}.

\bibitem[Simon(1983)]{SimonHolonomy}
Barry Simon.
\newblock Holonomy, the quantum adiabatic theorem, and {Berry}'s phase.
\newblock \emph{Phys. Rev. Lett.}, 51\penalty0 (24):\penalty0 2167--2170, 1983.
\newblock \doi{10.1103/PhysRevLett.51.2167}.

\bibitem[Liu et~al.(2025)Liu, Qiang, Lu, and Xie]{LiuQiangLuXie}
Tianyu Liu, Xiao-Bin Qiang, Hai-Zhou Lu, and X.~C. Xie.
\newblock Quantum geometry in condensed matter.
\newblock \emph{Natl. Sci. Rev.}, 12\penalty0 (3):\penalty0 nwae334, 2025.
\newblock \doi{10.1093/nsr/nwae334}.

\bibitem[Xiao et~al.(2010)Xiao, Chang, and Niu]{XiaoChangNiu}
Di~Xiao, Ming-Che Chang, and Qian Niu.
\newblock {Berry} phase effects on electronic properties.
\newblock \emph{Rev. Mod. Phys.}, 82\penalty0 (3):\penalty0 1959--2007, 2010.
\newblock \doi{10.1103/RevModPhys.82.1959}.

\bibitem[Sun et~al.(2024)Sun, Chen, Bzdu\v{s}ek, and Maciejko]{SunEtAl2024}
Canon Sun, Anffany Chen, Tom\'a\v{s} Bzdu\v{s}ek, and Joseph Maciejko.
\newblock Topological linear response of hyperbolic {Chern} insulators.
\newblock \emph{SciPost Phys.}, 17\penalty0 (5):\penalty0 124, 2024.
\newblock \doi{10.21468/SciPostPhys.17.5.124}.

\bibitem[Marzari and Vanderbilt(1997)]{MarzariVanderbilt}
Nicola Marzari and David Vanderbilt.
\newblock Maximally localized generalized {Wannier} functions for composite
  energy bands.
\newblock \emph{Phys. Rev. B}, 56\penalty0 (20):\penalty0 12847--12865, 1997.
\newblock \doi{10.1103/PhysRevB.56.12847}.

\bibitem[Marzari et~al.(2012)Marzari, Mostofi, Yates, Souza, and
  Vanderbilt]{MarzariReview}
Nicola Marzari, Arash~A. Mostofi, Jonathan~R. Yates, Ivo Souza, and David
  Vanderbilt.
\newblock Maximally localized {Wannier} functions: Theory and applications.
\newblock \emph{Rev. Mod. Phys.}, 84\penalty0 (4):\penalty0 1419--1475, 2012.
\newblock \doi{10.1103/RevModPhys.84.1419}.

\bibitem[Ozawa and Goldman(2018)]{OzawaGoldman}
Tomoki Ozawa and Nathan Goldman.
\newblock Extracting the quantum metric tensor through periodic driving.
\newblock \emph{Phys. Rev. B}, 97\penalty0 (20):\penalty0 201117, 2018.
\newblock \doi{10.1103/PhysRevB.97.201117}.

\bibitem[Peotta and T\"orm\"a(2015)]{PeottaTorma}
Sebastiano Peotta and P\"aivi T\"orm\"a.
\newblock Superfluidity in topologically nontrivial flat bands.
\newblock \emph{Nat. Commun.}, 6:\penalty0 8944, 2015.
\newblock \doi{10.1038/ncomms9944}.

\bibitem[Thouless et~al.(1982)Thouless, Kohmoto, Nightingale, and den
  Nijs]{TKNN}
David~J. Thouless, Mahito Kohmoto, M.~Peter Nightingale, and Marcel den Nijs.
\newblock Quantized {Hall} conductance in a two-dimensional periodic potential.
\newblock \emph{Phys. Rev. Lett.}, 49\penalty0 (6):\penalty0 405--408, 1982.
\newblock \doi{10.1103/PhysRevLett.49.405}.

\bibitem[Niu et~al.(1985)Niu, Thouless, and Wu]{NiuThoulessWu}
Qian Niu, David~J. Thouless, and Yong-Shi Wu.
\newblock Quantized {Hall} conductance as a topological invariant.
\newblock \emph{Phys. Rev. B}, 31\penalty0 (6):\penalty0 3372--3377, 1985.
\newblock \doi{10.1103/PhysRevB.31.3372}.

\bibitem[Mathai and Wilkin(2019)]{MathaiWilkin2019}
Varghese Mathai and Graeme Wilkin.
\newblock Fractional quantum numbers via complex orbifolds.
\newblock \emph{Lett. Math. Phys.}, 109\penalty0 (11):\penalty0 2473--2484,
  2019.
\newblock \doi{10.1007/s11005-019-01190-y}.

\bibitem[Mathai and Thiang(2019)]{MathaiThiang2019}
Varghese Mathai and Guo~Chuan Thiang.
\newblock Topological phases on the hyperbolic plane: fractional bulk--boundary
  correspondence.
\newblock \emph{Adv. Theor. Math. Phys.}, 23\penalty0 (3):\penalty0 803--840,
  2019.
\newblock \doi{10.4310/ATMP.2019.v23.n3.a5}.

\bibitem[Marcolli and Seipp(2017)]{MarcolliSeipp2017}
Matilde Marcolli and Kyle Seipp.
\newblock Twisted index theory on orbifold symmetric products and the
  fractional quantum {Hall} effect.
\newblock \emph{Adv. Theor. Math. Phys.}, 21\penalty0 (2):\penalty0 451--501,
  2017.
\newblock \doi{10.4310/ATMP.2017.v21.n2.a3}.

\bibitem[Zhang et~al.(2023)Zhang, Di, Zheng, Sun, and
  Zhang]{ZhangEtAlSecondChern}
Weixuan Zhang, Fengxiao Di, Xingen Zheng, Houjun Sun, and Xiangdong Zhang.
\newblock Hyperbolic band topology with non-trivial second {Chern} numbers.
\newblock \emph{Nat. Commun.}, 14:\penalty0 1083, 2023.
\newblock \doi{10.1038/s41467-023-36767-8}.

\bibitem[Ho\v{r}ava(2005)]{HoravaFermiK}
Petr Ho\v{r}ava.
\newblock Stability of {Fermi} surfaces and {$K$}-theory.
\newblock \emph{Phys. Rev. Lett.}, 95\penalty0 (1):\penalty0 016405, 2005.
\newblock \doi{10.1103/PhysRevLett.95.016405}.

\bibitem[Varma et~al.(1989)Varma, Littlewood, Schmitt-Rink, Abrahams, and
  Ruckenstein]{VarmaMarginalFermi}
Chandra~M. Varma, Peter~B. Littlewood, Stephan Schmitt-Rink, Elihu Abrahams,
  and Allan~E. Ruckenstein.
\newblock Phenomenology of the normal state of {Cu--O} high-temperature
  superconductors.
\newblock \emph{Phys. Rev. Lett.}, 63\penalty0 (18):\penalty0 1996--1999, 1989.
\newblock \doi{10.1103/PhysRevLett.63.1996}.

\bibitem[Bellissard(1986)]{BellissardKTheory}
Jean Bellissard.
\newblock K-theory of {$C^*$}-algebras in solid state physics.
\newblock In T.~C. Dorlas, N.~M. Hugenholtz, and M.~Winnink, editors,
  \emph{Statistical Mechanics and Field Theory: Mathematical Aspects}, volume
  257 of \emph{Lecture Notes in Physics}, pages 99--156. Springer, Berlin,
  1986.
\newblock \doi{10.1007/3-540-16777-3_74}.

\bibitem[Bellissard et~al.(1994)Bellissard, van Elst, and
  Schulz-Baldes]{BellissardVanElstSchulzBaldes}
Jean Bellissard, Andreas van Elst, and Hermann Schulz-Baldes.
\newblock The noncommutative geometry of the quantum {Hall} effect.
\newblock \emph{J. Math. Phys.}, 35\penalty0 (10):\penalty0 5373--5451, 1994.
\newblock \doi{10.1063/1.530758}.

\bibitem[Connes(1994)]{ConnesNCG}
Alain Connes.
\newblock \emph{Noncommutative Geometry}.
\newblock Academic Press, San Diego, 1994.

\bibitem[Baum et~al.(1994)Baum, Connes, and Higson]{BaumConnesHigson}
Paul Baum, Alain Connes, and Nigel Higson.
\newblock Classifying space for proper actions and {$K$}-theory of group
  {$C^*$}-algebras.
\newblock In Robert~S. Doran, editor, \emph{{$C^*$}-Algebras: 1943--1993},
  volume 167 of \emph{Contemporary Mathematics}, pages 240--291. American
  Mathematical Society, Providence, RI, 1994.
\newblock \doi{10.1090/conm/167/1292018}.

\bibitem[Higson and Kasparov(2001)]{HigsonKasparov}
Nigel Higson and Gennadi Kasparov.
\newblock {$E$}-theory and {$KK$}-theory for groups which act properly and
  isometrically on {Hilbert} space.
\newblock \emph{Invent. Math.}, 144\penalty0 (1):\penalty0 23--74, 2001.
\newblock \doi{10.1007/s002220000118}.

\bibitem[Green(1978)]{GreenImprimitivity}
Philip Green.
\newblock The local structure of twisted covariance algebras.
\newblock \emph{Acta Math.}, 140:\penalty0 191--250, 1978.
\newblock \doi{10.1007/BF02392308}.

\bibitem[Pimsner and Voiculescu(1980)]{PimsnerVoiculescu1980}
Mihai Pimsner and Dan Voiculescu.
\newblock Exact sequences for {$K$}-groups and {Ext}-groups of certain
  cross-product {$C^*$}-algebras.
\newblock \emph{J. Operator Theory}, 4\penalty0 (1):\penalty0 93--118, 1980.

\bibitem[Rieffel(1981)]{RieffelRotation1981}
Marc~A. Rieffel.
\newblock {$C^*$}-algebras associated with irrational rotations.
\newblock \emph{Pacific J. Math.}, 93\penalty0 (2):\penalty0 415--429, 1981.
\newblock \doi{10.2140/pjm.1981.93.415}.

\bibitem[Carey et~al.(1998)Carey, Hannabuss, Mathai, and
  McCann]{CareyHannabussMathaiMcCann1998}
Alan~L. Carey, Keith~C. Hannabuss, Varghese Mathai, and Paul McCann.
\newblock Quantum {Hall} effect on the hyperbolic plane.
\newblock \emph{Comm. Math. Phys.}, 190\penalty0 (3):\penalty0 629--673, 1998.
\newblock \doi{10.1007/s002200050255}.

\bibitem[Bellissard et~al.(1992)Bellissard, Bovier, and
  Ghez]{BellissardBovierGhez}
Jean Bellissard, Anton Bovier, and Jean-Michel Ghez.
\newblock Gap labelling theorems for one-dimensional discrete {Schr}\"odinger
  operators.
\newblock \emph{Rev. Math. Phys.}, 4\penalty0 (1):\penalty0 1--37, 1992.
\newblock \doi{10.1142/S0129055X92000029}.

\bibitem[Kellendonk and Putnam(2000)]{KellendonkPutnam}
Johannes Kellendonk and Ian~F. Putnam.
\newblock Tilings, {$C^*$}-algebras, and {$K$}-theory.
\newblock In Michael Baake and Robert~V. Moody, editors, \emph{Directions in
  Mathematical Quasicrystals}, volume~13 of \emph{CRM Monograph Series}, pages
  177--206. American Mathematical Society, Providence, RI, 2000.
\newblock \doi{10.1090/crmm/013/07}.

\bibitem[Freed and Lott(2010)]{FreedLott}
Daniel~S. Freed and John Lott.
\newblock An index theorem in differential {$K$}-theory.
\newblock \emph{Geom. Topol.}, 14\penalty0 (2):\penalty0 903--966, 2010.
\newblock \doi{10.2140/gt.2010.14.903}.

\bibitem[Baake and Grimm(2013)]{BaakeGrimm}
Michael Baake and Uwe Grimm.
\newblock \emph{Aperiodic Order, Volume 1: A Mathematical Invitation}, volume
  149 of \emph{Encyclopedia of Mathematics and its Applications}.
\newblock Cambridge University Press, Cambridge, 2013.
\newblock \doi{10.1017/CBO9781139025256}.

\bibitem[Atiyah(1967)]{AtiyahKTheory}
Michael~F. Atiyah.
\newblock \emph{{K}-Theory}.
\newblock W. A. Benjamin, New York, 1967.

\bibitem[Karoubi(1978)]{KaroubiKTheory}
Max Karoubi.
\newblock \emph{{K}-Theory: An Introduction}, volume 226 of \emph{Grundlehren
  der mathematischen Wissenschaften}.
\newblock Springer, Berlin, 1978.
\newblock \doi{10.1007/978-3-540-79890-3}.

\end{thebibliography}
\endgroup

\end{document}